\documentclass[12pt]{article}

\usepackage{amsmath,amssymb,amsthm, amstext} 
\usepackage{bbm}
\usepackage{stmaryrd}
\usepackage{commath}
\usepackage[utf8]{inputenc}
\usepackage{enumitem}
\usepackage{dsfont}
\usepackage{natbib}
\usepackage{mathrsfs} 
\usepackage[ruled]{algorithm2e}
\usepackage{appendix}
\usepackage{float}
\usepackage{mathtools}
\usepackage{xr}
\newenvironment{policy}[1][htb]
{%
\begin{algorithm}[#1]%
}{\end{algorithm}}

\usepackage{mathrsfs} 
\usepackage[colorlinks, pdfstartview=FitH, citecolor = black, linkcolor  = black]{hyperref}

\usepackage{graphicx}
\usepackage{setspace}
\setdisplayskipstretch{}

\renewcommand{\P}{\mathbb{P}}

\renewcommand{\rho}{\varrho}

\newcommand\eps{\varepsilon}

\newcommand\N{\mathbb{N}}
\newcommand\E{\mathbb{E}}
\newcommand\R{\mathbb{R}}
\newcommand\Z{\mathbb{Z}}

\DeclareMathOperator*{\argmax}{arg\,max}

\newtheorem{theorem}{Theorem}[section]  
\newtheorem{assumption}[theorem]{Assumption}

\newtheorem{corollary}[theorem]{Corollary}

\newtheorem{lemma}[theorem]{Lemma}

\theoremstyle{definition}
\newtheorem{example}[theorem]{Example}
\newtheorem{remark}[theorem]{Remark}

\begin{document}
\title{Treatment recommendation with distributional targets\footnote{
We thank Quentin Badolle for excellent research assistance in an early stage of this project.
}
}



\author{
\begin{tabular}{c}
Anders Bredahl Kock \\ 
\small	University of Oxford \\
\small	CREATES, Aarhus University\\
\small	10 Manor Rd, Oxford OX1 3UQ
\\
\small	{\small	\href{mailto:anders.kock@economics.ox.ac.uk}{anders.kock@economics.ox.ac.uk}} 
\end{tabular}
\and
\begin{tabular}{c}
David Preinerstorfer\footnote{Corresponding author.} \\ 
{\small	SEPS-SEW} \\ 
{\small	University of St.~Gallen} \\
{\small	Varnb\"uelstrasse 14, 9000 St.~Gallen } \\ 
{\small	 \href{mailto:david.preinerstorfer@unisg.ch}{david.preinerstorfer@unisg.ch}}
\end{tabular}
\and
\begin{tabular}{c}
Bezirgen Veliyev \\ 
{\small	CREATES, Department of Economics} \\ 
\small Aarhus University \\ 
{\small	Fuglesangs Alle 4, 8210 Aarhus V.} \\ 
{\small	\href{mailto:bveliyev@econ.au.dk}{bveliyev@econ.au.dk}}
\end{tabular}
}

\date{March, 2022}

\maketitle

\begin{abstract}

We study the problem of a decision maker who must provide the best possible treatment recommendation based on an experiment. The desirability of the outcome distribution resulting from the policy recommendation is measured through a functional capturing the distributional characteristic that the decision maker is interested in optimizing. This could be, e.g., its inherent inequality, welfare, level of poverty or its distance to a desired outcome distribution. 
If the functional of interest is not quasi-convex or if there are constraints, the optimal recommendation may be a \emph{mixture} of treatments. This vastly expands the set of recommendations that must be considered. We characterize the difficulty of the problem by obtaining maximal expected regret lower bounds. Furthermore, we propose two (near) regret-optimal policies. The first policy is static and thus applicable irrespectively of subjects arriving sequentially or not in the course of the experimentation phase. The second policy can utilize that subjects arrive sequentially by successively eliminating inferior treatments and thus spends the sampling effort where it is most needed. 


\end{abstract}

%

\pagebreak

\onehalfspacing

\section{Introduction}

We consider a decision maker who wants to run an experiment to identify the ``best" among a set of candidate treatments. While most previous work has focused on the case of targeting the distribution with the highest mean outcome, our framework lets the decision maker use a custom-built \emph{distributional characteristic}. This characteristic summarizes the targeted properties of the outcome distributions, e.g., their poverty-, welfare-, or inequality-implications. Such generality is particularly crucial for socio-economic decision making, as targeting solely the highest mean outcome may result in excessive inequality or poverty, which can be avoided by using an appropriately designed functional. For a discussion concerning the importance of studying distributional effects of a policy beyond the mean in the context of welfare reform research we refer to \cite{bitler2006mean}. 

Compared to targeting the treatment with the highest expectation, allowing for general distributional characteristics substantially increases the complexity of the decision maker's problem. This is the case in particular if the functional used to compare treatments is not quasi-convex. Then, there is no guarantee that there always exists a single best treatment that is at least as good as any \emph{mixture} of treatments. Thus, the decision maker has to determine which mixture of treatments is best. Here, a vector of mixture weights corresponds to the proportion of the population to which each treatment is rolled out after the experimentation phase. Note that the functional is generally not quasi-convex as soon as one of its ``components'' is not quasi-convex. If, for example, the policy maker seeks to target a distribution that combines a high expectation with low poverty, the resulting functional is not quasi-convex if the poverty measure used is not quasi-convex. 

A practically important aspect that we incorporate is that the decision maker may be obliged to respect \emph{constraints} concerning the type of mixture that can be implemented. For example, there could be treatments that cannot be given to the whole population, or treatments that have to be given to at least a certain proportion of the population. Another example of restrictions arises when there are groups of treatments that are incompatible in the sense that they cannot be implemented jointly, because, e.g., they depend on different types of infrastructures that are too costly to maintain simultaneously. Even if the functional is quasi-convex, such as the mean, the decision maker can be confronted with constraints of the just-mentioned type. Then also in this case mixtures of treatments have to be considered.

Concerning functional targets, related papers are~\cite{kock2017optimal},~\cite{cassel2018general} and~\cite{kpv1}. Note, however, that the results developed there consider a completely different problem of exploration-exploitation type. The regret in these papers is cumulative in the sense that i) every subject not assigned to the best treatment contributes to the regret, and ii) a loss incurred for one subject can not be compensated by future assignments. This is in contrast to the present paper, where the goal is to roll out a mixture of treatments after an experimental phase, in which the decision maker can learn which mixture is best. Mistakes in assignments during the experimental phase do not contribute to the regret, which only measures the quality of the recommended mixture. Thus, it is solely the distribution across subjects in the roll-out phase that matters. Our new objective requires new policies. Indeed, it is known already for the case of targeting the mean functional that an \emph{upper} bound on the exploration-exploitation regret of a policy implies a \emph{lower} bound on its regret in the pure exploration problem (studied in the present work) targeting the quality of a final recommendation, cf.~\cite{bubeck2009pure}. Thus, policies that work well in the exploration-exploitation settings of~\cite{kock2017optimal},~\cite{cassel2018general} and~\cite{kpv1} need not do so in our setting and new policies must be developed for our recommendation problem. This need for new policies is amplified by the fact that the just cited works do not allow the decision maker to target mixtures of treatments, nor to incorporate constraints (such as capacity  or incompatibility constraints).
Although incorporating constraints could also be interesting in the papers just cited, it is debatable whether targeting a mixture per se is even sensible in exploration-exploitation problems, where the goal is to assign every individual to the best treatment. 

In the present article we establish theoretical results to assist a decision maker who needs to give a recommendation based on an experiment. We study optimality properties of policies in terms of their maximal expected regret.  Here, regret is defined as the difference between the distributional characteristic of the optimal mixture of treatments (over the feasible set of weights considered) and the distributional characteristic of the weights recommended by the policy. 

Before we discuss how our results are related to the literature, we summarize our contributions: 
\begin{enumerate}
	\item We develop general lower bounds on the maximal expected regret of static and sequential policies. Obtaining lower bounds allows us to understand how the size of the sample used for experimentation, the number of treatments, and structural aspects of the feasible set of mixture weights affect the difficulty of the problem. The size of these lower bounds depends intrinsically on the structure of the set of feasible mixture weights.
	Thus, establishing useful lower bounds is a delicate matter for severely restricted subsets. 
	\item We study static assignment policies and investigate under which conditions such policies are optimal.
	\item We investigate optimality properties of a sequential assignment policy that attempts to eliminate ``inferior'' (groups of) treatments during the experimentation phase. This elimination strategy can help to target the sampling effort to where it is most needed, which we also illustrate in our numerical results.
\end{enumerate}

\subsection{Related literature}

Our article draws on ideas from the multi-armed bandit literature. Important early contributions include~\cite{thomp},~\cite{robbins1952some}, \cite{gittins1979bandit}, and~\cite{lai1985asymptotically}; cf.~\cite{cesa2006prediction}, ~\cite{bubeck2012regret} and \cite{lattimore2019bandit} for introductions to the subject and for further references. A large part of this literature is focused on balancing an exploration-exploitation trade-off. In contrast, our framework resembles that of a (fixed budget) ``pure exploration'' problem, in that the goal is to give the best recommendation by the end of the experimentation phase. Suboptimal assignments made during this phase do not enter the regret function. The pure-exploration literature dates back at least to~\cite{bubeck2009pure}, cf.~also \cite{audibert2010best}, \cite{karnin2013almost}, \cite{jamieson2014lil}, \cite{carpentier2016tight}, \cite{kaufmann2016complexity}, and  \cite{kasy2019adaptive}. In all articles just mentioned, the goal is to find the treatment with the highest expectation in unconstrained situations. An article that suggests a sequential batch-elimination policy in the pure exploration setting for distributional targets is~\cite{tran2014functional}. The article does not study maximal expected regret optimality properties of the policy introduced nor related performance lower bounds. Furthermore, the policy introduced only allows targeting the best individual treatment, and does not allow the decision maker to target the best mixture of treatments, which is crucial for functionals that are not quasi-convex or for constrained problems.

The vector of mixture weights the decision maker can recommend takes its values in the subset of the unit simplex described by the imposed constraints. It is tempting to interpret this problem as a functional target version of the standard continuum (or~$\mathcal{X}$-armed) bandit problem which is primarily concerned with targeting the mean (e.g., \cite{agrawal1995continuum}, \cite{kleinberg2008multi} and \cite{bubeck2011x}). This, however, would completely overlook the delicate role played by the mixtures in the problem we are dealing with: The decision maker can always assign a subject to only one out of finitely many treatments, not to a mixture of those treatments. That is, in continuum-armed bandit parlance, the sampling in the experimentation phase can exclusively be done from the vertices of the simplex (even though these vertices may not satisfy the constraints imposed on the final recommendation). For general continuum-armed bandit problems, such a sampling scheme is obviously not enough to find the optimal recommendation, as there is no hope to learn a function in the interior of the simplex from its values on the vertices. However, the problem we consider is such that learning the underlying finitely many outcome distributions of the treatments permits us to solve the recommendation problem, because the target function we are optimizing is a (nonlinear) function of the mixture of the unknown treatment outcome distributions. Also note that, equipped with the special additional structure just explained, a standard continuum-armed bandit problem targeting the mean functional would reduce to a pure-exploration problem with finitely many arms. Thus, it is precisely the nonlinearity of the target functional and the presence of constraints that creates problems of the type considered in the present article. 

Our paper is also related to the work on statistical treatment rules initiated by \cite{manski2004statistical}, and developed further by		 \cite{dehejia2005program}, \cite{stoye2009minimax}, \cite{hirano2009asymptotics}, \cite{stoye2012minimax}, \cite{bhattacharya2012inferring}, \cite{Manski10518}, \cite{kitagawa2018should}, \cite{athey2017efficient}. The central question attacked in these papers is to learn, in a non-sequential setting, whom to assign to the treatment based on observed covariates. The non-asymptotic results in this literature are focused on designing treatment rules maximizing the conditional mean. One exception is 
\cite{kitagawa2019equality}, who studied a class of rank dependent social welfare functions as targets. Although these social welfare functions cover interesting functionals such as the (extended) Gini-welfare measure, they are quasi-convex, which greatly simplifies the space of treatment rules as no mixtures need to be considered. In addition, the authors focused exclusively on static policies. However,~\cite{kitagawa2019equality} allow the treatment choice to be based on a vector of covariates. In this sense our results are complimentary and none is a subcase of the other.\footnote{We provide a heuristic discussion of how to incorporate covariates into our framework in Section \ref{sec:cov}.} 

For many classes of distributional characteristics classical inference problems, such as constructing a non-sequential asymptotically efficient estimator or test, or partial identification issues concerning distributional effects have been studied; e.g., \cite{biewen2002bootstrap}, \cite{davidson2007asymptotic}, \cite{BarrettDonald2009}, \cite{rothe2010nonparametric, rothe2012partial}, \cite{stoyspread}, \cite{dufour2017permutation}. Our goal is conceptually different, in that we analyze the specific decision problem of identifying the best treatment based on an experiment of a given size. As discussed in~\cite{Manski10518} this task is only weakly related to significance testing, which is typically focused on controlling Type I and Type II errors. From a technical perspective, we do not use asymptotic approximations, but we mainly use (finite-sample) concentration inequalities of plug-in estimators for the functional under consideration, and information-theoretic arguments to derive our maximal expected regret lower bounds.

\bigskip

The structure of the remaining article is as follows: We introduce the framework in Section~\ref{sec:framework}, then we present lower bounds for maximal expected regret in Section~\ref{sec:lowbound}. Static assignment policies are discussed in Section~\ref{sec:nonseq}, whereas sequential policies are treated in Section~\ref{sec:seq}. In Section~\ref{sec:num} we summarize the results of three numerical experiments. We provide an informal discussion of how covariates can be dealt with in Section~\ref{sec:cov}. All proofs are collected in Appendices~\ref{app:aux},~\ref{app:prfs1},~\ref{app:es},~\ref{sec:SEP}, and~\ref{sec:elimwincp}.

\section{Framework}\label{sec:framework}

In this section we formally describe the observational structure, the distributional target and the objective of the decision maker, as well as the policies we consider. 

\subsection{Observational structure}

The potential outcome of assigning subject~$t$ to treatment~$i\in\mathcal{I}:=\cbr[0]{1,\hdots,K}$, $K \geq 2$, shall be denoted by~$Y_{i,t}$. This potential outcome will be interpreted as a draw from an \emph{unknown} cumulative distribution function (cdf)~$F^i$, i.e., the cdf obtained by rolling out treatment~$i$ to an infinitely large population. We shall assume that~$F^i \in \mathscr{D} \subseteq D_{cdf}([a,b])$ holds for~$i = 1, \hdots, K$, where~$a < b$ are real numbers,~$D_{cdf}([a,b])$ denotes the set of all cdfs~$F$ on~$\R$ such that~$F(a-) = 0$ and~$F(b) = 1$, and~$\mathscr{D}$ is the common ``parameter space'' of the unknown outcome distributions. As such, the set~$\mathscr{D}$ describes the assumptions one is willing to put on the unknown cdfs~$F^1, \hdots, F^K$. For example,~$\mathscr{D}$ could be a set of cdfs satisfying certain smoothness conditions. The potential outcome vector of subject~$t$ will be denoted by~$Y_t := (Y_{1,t}, \hdots, Y_{K, t})$. Furthermore, for every~$t$, we let~$G_t$ be a random variable which can be used for randomization in assigning the~$t$-th subject. We think of the \emph{randomization measure}, i.e., the distribution of~$G_t$, as being fixed, e.g., the uniform distribution on~$[0, 1]$. 

Throughout we impose the following assumption.

\begin{assumption}\label{as:dgp}
The random vectors~$Y_t$ for~$t \in \N$ are independent and identically distributed (i.i.d.); the sequence of random variables~$G_t$ for~$t \in \N$ is i.i.d., and is independent of the sequence~$Y_t$. Furthermore,~$\mathscr{D}$ is convex (and non-empty).
\end{assumption}

Note that no assumptions are imposed concerning the dependence between the components of each random vector~$Y_t$.

\subsection{Target of the decision maker}\label{sec:objective}

Rolling out multiple treatments by randomly assigning treatment~$i$ with proportion~$\delta_i$ leads to the population (mixture) cdf~$\sum_{i = 1}^K \delta_i F^i$, which we write as~$\langle \delta, \mathbf{F} \rangle$ for~$\mathbf{F} = (F^1, \hdots, F^K)$. Here,~$\delta_i \in [0, 1]$ and~$\sum_{i = 1}^K \delta_i = 1$. In order to judge which vector of proportions~$\delta$ is ``best," we need to compare the quality of cdfs. To this end, we assume that the decision maker is interested in maximizing a certain characteristic of the cdf, measured by a pre-specified~\emph{functional}
\begin{equation}
	\mathsf{T}: D_{cdf}([a,b]) \to \R.
\end{equation}
Given~$\mathsf{T}$, the cdf~$F \in D_{cdf}([a,b])$ is considered as being better than~$G \in D_{cdf}([a,b])$ if~$\mathsf{T}(F) > \mathsf{T}(G)$. 

\begin{remark}
We think of the treatment recommendation as being rolled out to a large population, in which a proportion~$\delta_i$ of individuals is assigned to treatment~$i$. Thus, every single subject is only assigned to one treatment, but not every subject need to be assigned to the same treatment.\footnote{In particular, we do not think of the vector of proportions~$\delta$ as meaning that every individual is assigned a ``dose" of~$\delta_i$ of treatment~$i$ --- the treatments can not be divided.} The policy maker then wishes to choose a vector of proportions that maximizes~$\delta\mapsto\mathsf{T}(\langle \delta, \mathbf{F} \rangle)$, cf.~\eqref{eqn:simplmax} below.
\end{remark}

Technically, the main assumption we impose on~$\mathsf{T}$ is as follows, where for two cdfs~$F$ and~$G$ we denote~$\|F - G\|_{\infty} = \sup_{x \in \R} |F(x) - G(x)|$ (the same symbol is used for the supremum norm on~$\R^K$). 
\begin{assumption}\label{as:MAIN}
	The functional~$\mathsf{T}: D_{cdf}([a,b]) \to \R$ and the non-empty set~$\mathscr{D} \subseteq D_{cdf}([a,b])$ satisfy
	\begin{equation}\label{eqn:lipcondAS}
		|\mathsf{T}(F) - \mathsf{T}(G)|\leq C\|F- G\|_{\infty} \quad \text{ for every } \quad F \in \mathscr{D} \text{ and every } G \in D_{cdf}([a,b])
	\end{equation}
	for some~$C > 0$.
\end{assumption}

For expositional purposes it is best to keep~$\mathsf{T}$ abstract. Several examples are briefly discussed in the following remark, and in more detail in Section~\ref{sec:funex} further below.
\begin{remark}[Examples]\label{rem:ex}
Assumption~\ref{as:MAIN} is satisfied for many inequality measures (e.g., the Schutz-coefficient, Gini-index, linear inequality measures, generalized entropy family, Atkinson-indices, or Kolm-indices), welfare measures based on these inequality measures (e.g., the Gini-welfare measure), poverty measures (e.g., the headcount ratio, or Sen- and Foster-families of poverty measures), quantiles, U-functionals, L-functionals, or generalized trimmed mean functionals. The assumption was introduced in~\cite{kpv1}; cf.~their Remarks 2.2-2.4 for some discussion, and see their Section~2.1 and Appendices~E and~G for a detailed discussion of examples.
\end{remark}

\begin{remark}
In practice a policy maker often seeks to find a balance between several policy objectives. For example, one may wish to trade off a high expectation with low poverty. Such considerations are encompassed in the present framework: For functionals~$\mathsf{T}_1,\hdots,\mathsf{T}_l$ on~$D_{cdf}([a,b])$ and a function~$g$ defined on the range of these one can define~$F\mapsto \mathsf{T}(F)$ as
\begin{align*}
\mathsf{T}(F)=g\left(\mathsf{T}_1(F),\hdots,\mathsf{T}_l(F)\right).
\end{align*}
If~$\mathsf{T}_1,\hdots,\mathsf{T}_l$ satisfy Assumption~\ref{as:MAIN} and~$g$ is Lipschitz continuous, then also~$\mathsf{T}$ satisfies Assumption~\ref{as:MAIN}. Finally, note that~$\mathsf{T}$ will generally not be quasi-convex as soon as one of the building blocks~$\mathsf{T}_1,\hdots,\mathsf{T}_l$ is not quasi-convex (unless~$g$ has further special structure). For example, numerous poverty measures are not quasi-convex, cf.~Example~\ref{eqn:pov} below. 
\end{remark}
\medskip

\emph{If the decision maker knew}~$F^1, \hdots, F^K$, the goal would be to target a vector of weights~$\delta$ that maximizes~$\mathsf{T}(\langle \delta, \mathbf{F}\rangle)$. Denoting by~$\mathscr{M}_K$ the (non-empty) set of weight vectors the decision maker is restricted to work with, the task would then be to find an element of\footnote{Note that the set in Equation~\eqref{eqn:simplmax} is non-empty for every~$\mathbf{F} \in \mathscr{D} \times \hdots \times \mathscr{D}$ if, e.g.,~$\mathscr{D}$ is convex,~$\mathsf{T}$ is continuous on~$\mathscr{D}$ w.r.t.~the supremum metric induced from~$D_{cdf}([a,b])$ (that is for all~$F\in\mathscr{D}$ and all~$\eps>0$ there exists a~$\delta>0$ such that for all~$G\in\mathscr{D}$ satisfying~$||F-G||_\infty<\delta$ it holds that~$|\mathsf{T}(F)-\mathsf{T}(G)|<\eps$.) and~$\mathscr{M}_K$ is a (non-empty) closed subset of the standard simplex of~$\R^K$. This is guaranteed under Assumptions~\ref{as:dgp}, \ref{as:MAIN}, and~\ref{as:M}.}
\begin{equation}\label{eqn:simplmax}
	\argmax_{\delta \in \mathscr{M}_K} \mathsf{T}(\langle \delta, \mathbf{F} \rangle).
\end{equation}
Leaving aside for a moment that~$\mathbf{F}$ is unknown to the decision maker and that solving the optimization problem in~\eqref{eqn:simplmax} is only one part of the problem, we now introduce some notation used throughout this article and discuss in more detail frequently encountered constraints imposed upon the decision maker.

In the ``unrestricted'' case the set~$\mathscr{M}_K$ of weights coincides with the standard simplex in~$\R^K$, which we denote as~$$\mathscr{S}_K  := \{\delta = (\delta_1, \hdots, \delta_K)' \in [0, 1]^K: \delta_1 + \hdots + \delta_K = 1 \} \subseteq \R^K.$$

An important example of a constrained set is~$\mathscr{M}_K = \mathscr{E}_K := \{e_1(K), \hdots, e_K(K)\}$, where~$e_i(K)$ denotes the~$i$-th element of the standard basis in~$\R^K$. This set describes the situation in which the decision maker can only recommend a single treatment out of the~$K$ treatments. Other frequently encountered constraints are \emph{capacity constraints}, imposing upper- or lower-bound restrictions on single weights~$\delta_i$ or, more generally, imposing upper- or lower-bound restrictions on the total weight put on a combination of a subset of treatments; and \emph{similarity constraints}, leading to upper bounds on the absolute difference between weights. It is clear that capacity constraints or similarity constraints can typically be incorporated by imposing linear inequality restrictions on~$\delta$.

In the following example we shall discuss \emph{incompatibility constraints}, which are present whenever certain groups of treatments are ``incompatible," in the sense that treatments belonging to different groups cannot be rolled out jointly. This is relevant if, e.g., groups of treatments depend on different types of infrastructures, which cannot be administered together due to budget limitations. 
\begin{example}\label{ex:compa}
	Let~$m$ be an integer such that~$1 < m \leq K$, let~$\{A_1, \hdots, A_m\}$ be a partition of~$\mathcal{I}$ (i.e., the incompatible groups of treatments), and set 
	\begin{equation}\label{eqn:part}
		\mathscr{M}_K = \bigcup_{j = 1}^m \mathscr{M}_{A_j, K},~ \text{ where } \emptyset \neq \mathscr{M}_{A_j, K} \subseteq \mathscr{S}_{A_j, K} := \{\delta \in \mathscr{S}_K : \delta_k = 0 \text{ if } k \notin A_j \}.
	\end{equation}
	As a special case, note that the partition~$A_j = \{j\}$ for~$j = 1, \hdots, K$, leads to~$\mathscr{M}_K = \mathscr{E}_K$, i.e., when each treatment constitutes its own group. Observe furthermore that in addition to the incompatibility constraints described by the partition, additional constraints might be imposed through the choices of~$\mathscr{M}_{A_j, K}$.
\end{example}
The structure of~$\mathscr{M}_K$ described in Equation~\eqref{eqn:part} is already very flexible. For most results, however, we shall only need to impose the following assumption ensuring that the problem is non-trivial, and that the set in~\eqref{eqn:simplmax} is non-empty (for~$\mathsf{T}$ continuous on the convex set~$\mathscr{D}$).
\begin{assumption}\label{as:M}
	The set~$\mathscr{M}_K \subseteq \mathscr{S}_K$ contains at least two elements and is closed.	
\end{assumption}
Before we discuss policies that a decision maker may employ to learn an element of~\eqref{eqn:simplmax}, we comment on how the restricted set~$\mathscr{M}_K$ imposed upon the decision maker can sometimes be reduced without loss under certain conditions on~$\mathsf{T}$ and~$\mathscr{D}$. Furthermore, we discuss three examples of functionals.

\subsubsection{Special cases where~$\mathscr{M}_K$ can be reduced}\label{sec:qconv}

A set of potential weights~$\mathscr{M}_K$ may sometimes without loss be reduced by exploiting properties of the specific functional~$\mathsf{T}$ and parameter space~$\mathscr{D}$, both of which are known to the decision maker (while the cdfs~$\mathbf{F} = (F^1, \hdots, F^K)$ are unknown). This is the case if~$\mathsf{T}$ is such that for \emph{any}~$\mathbf{F} \in \mathscr{D} \times \hdots \times \mathscr{D}$ the set of maximizers in~\eqref{eqn:simplmax} contains an element of~$\mathscr{M}^*_K \subsetneqq \mathscr{M}_K$, where~$\mathscr{M}^*_K$ does not depend on~$\mathbf{F}$. In such a case, nothing is lost by restricting~$\mathscr{M}_K$ to~$\mathscr{M}^*_K$, i.e., by targeting~$\argmax_{\delta \in \mathscr{M}^*_K} \mathsf{T}(\langle \delta, \mathbf{F}\rangle)$. 

The leading example of such a reduction is the case where~$\mathsf{T}$ restricted to the convex set~$\mathscr{D}$ is continuous and is quasi-convex, that is
\begin{equation}\label{eqn:qc}
	\max(\mathsf{T}(F), \mathsf{T}(G)) \geq \max_{\alpha \in [0, 1]} \mathsf{T}(\alpha F + (1-\alpha) G), \text{ for every } F, G \in \mathscr{D}.
\end{equation}
In this case one may restrict~$\mathscr{M}_K$ to the set of its extreme points (i.e., the subset of elements of~$\mathscr{M}_K$ that cannot be written as a strict convex combination of two other elements of~$\mathscr{M}_K$).\footnote{This follows from quasi-convexity implying that~$\max_{\delta \in \mathscr{M}_K} \mathsf{T}(\langle \delta, \mathbf{F}\rangle) = \max_{\delta \in \mathrm{conv}(\mathscr{M}_K)} \mathsf{T}(\langle \delta, \mathbf{F}\rangle)$, for~$\mathrm{conv}(\mathscr{M}_K)$ the convex hull of~$\mathscr{M}_K$, together with~$\max_{\delta \in \mathrm{conv}(\mathscr{M}_K)} \mathsf{T}(\langle \delta, \mathbf{F}\rangle)$ being attained at an extreme point of~$\mathrm{conv}(\mathscr{M}_K)$, e.g., Theorem 4.6.3 in~\cite{MR2459665}, and thus at an extreme point of~$\mathscr{M}_K$.}
For example, if~$\mathscr{E}_K \subseteq \mathscr{M}_K$, one can then replace~$\mathscr{M}_K$ by~$\mathscr{E}_K$. That is, mixtures do not need to be taken into account and it suffices to consider the individual treatments. Note further that if~$\mathscr{E}_K \not \subseteq \mathscr{M}_K$, e.g., because one or more treatments have capacity constraints, the set of extreme points of~$\mathscr{M}_K$ has a more complicated structure, and may even be infinite or not closed, in which case one may work with its closure to enforce Assumption~\ref{as:M}.

Whether such reductions are possible depends on~$\mathsf{T}$ and~$\mathscr{D}$. We note that even if Equation~\eqref{eqn:qc} holds, not reducing~$\mathscr{M}_K$ obviously does not lead to a loss concerning the maximal attainable value of~$\mathsf{T}$.

\subsubsection{Examples of functionals}\label{sec:funex}

\begin{example}\label{ex:gini}
	A functional for which Assumption~\ref{as:MAIN} holds (with~$C = 2(b-a)$ and~$\mathscr{D} = D_{cdf}([a, b])$, cf.~the discussion after Lemma~E.9 in~\cite{kpv1}) is the Gini-welfare measure
	\begin{equation}
		F \mapsto \mu(F) - \frac{1}{2} \int \int | x_1 - x_2 | dF(x_1) dF(x_2), \text{ where } \mu(F) = \int x dF(x).
	\end{equation}
	The functional can equivalently be written as~$a + \int_{a}^b (1-F(x))^2 dx$, from which it easily follows that it is quasi-convex with~$\mathscr{D} = D_{cdf}([a,b])$.
\end{example}

\begin{example}\label{eqn:pov}
	One important poverty measure, for which Assumption~\ref{as:MAIN} holds (e.g., with~$a = 0$,~$b = 1$, ~$\mathscr{D}$ equal to the subset of cdfs in~$D_{cdf}([0, 1])$ with a density bounded from above by~$s > 0$, and~$C = 1 + s/2$, cf.~Lemma~E.10 in~\cite{kpv1}), but which is not quasi-convex in general, is the (negative) headcount ratio with poverty line equaling half the mean
	\begin{equation}\label{eqn:headcount}
		F \mapsto -F(\mu(F)/2).
	\end{equation}
	Here we multiply by~$-1$, as we aim at maximizing~$\mathsf{T}$, and one typically wants to find the treatment combination leading to the smallest fraction of ``poor'' individuals. Other poverty measures that satisfy Assumption~\ref{as:MAIN} but are not quasi-convex (in general) include the ones by~\cite{sen1976} and~\cite{foster84}.
\end{example}

\begin{example}\label{ex:distt}
	Consider finally a situation in which the decision maker has an ``ideal'' cdf~$F^* \in \mathscr{D} \subseteq D_{cdf}([a,b])$ in mind, and intends to assign the treatments in such a way that~$\langle \delta, \mathbf{F} \rangle$ resembles this ideal cdf~$F^*$ as closely as possible. Then, one could work with the functional
	\begin{equation}
		F \mapsto -\|F - F^*\|_{\infty}.
	\end{equation}
	By the inverse triangle inequality, Assumption~\ref{as:MAIN} is seen to be satisfied with~$C = 1$ (for any~$\mathscr{D} \subseteq D_{cdf}([a,b])$, and any pair of real numbers~$a < b$). Similar as in the previous example, it can be shown that~\eqref{eqn:qc} is not generally satisfied.
\end{example}

\subsection{Policies}\label{sec:policies}

The decision maker's targets~$\argmax_{\delta \in \mathscr{M}_K} \mathsf{T}(\langle \delta, \mathbf{F} \rangle)$ are not readily accessible, because the cdfs~$\mathbf{F}$ are unknown. Therefore, before rolling out a certain combination of treatments, the decision maker needs to learn~$\mathbf{F}$ in order to reach a recommendation concerning which~$\delta$ is best. The \emph{objective} of the present paper is to devise optimal policies for this problem, i.e., optimal ways of obtaining such recommendations based on the results of assigning (statically or sequentially) a small number of subjects to treatments in the course of an experimentation phase. 

Essentially, a policy is a prescription concerning two steps:  the first step concerns exploring the efficacies of the available treatments based on an experiment involving~$n$ subjects. Here, a policy needs to describe how to assign these subjects to one out of the~$K$ treatments. These assignments may be \emph{static}, e.g., based on an exogenous random draw from the set of all partitions of~$n$ subjects, or \emph{sequential}, where the assignment of the~$t$-th subject depends on the outcomes of the previously observed subjects. Note that sequential assignments can be designed so as to adaptively focus their sampling effort to where it is needed most. In the second phase, after the outcomes of all~$n$ subjects have been observed, the policy prescribes which element~$\hat{\delta} \in \mathscr{M}_K$ to recommend, aiming at a recommendation~$\hat{\delta}$ satisfying~$\mathsf{T}(\langle \hat{\delta}, \mathbf{F} \rangle) \approx \max_{\delta \in \mathscr{M}_K} \mathsf{T}(\langle \delta, \mathbf{F} \rangle)$.

Formally, a policy~$\pi$ with recommendations in~$\mathscr{M}_K$ is a triangular array~$\{ \pi_{n,t}: n \in \N, n \geq K, t = 1, \hdots, n+1\}$ of measurable functions, such that for every natural number~$n \geq K$ 
\begin{align*}
	\pi_{n,t} : &\left([a,b] \times \mathbb{R}\right)^{t-1} \times \mathbb{R} \to \mathcal{I} = \{1, \hdots, K\} \quad \text{ for } t = 1, \hdots, n, \\
	\pi_{n,n+1} : & \left([a,b] \times \mathbb{R}\right)^{n} \times \mathbb{R} \to \mathscr{M}_K.
\end{align*}
For~$t = 1, \hdots, n$ the outcome of~$\pi_{n,t}$ is to be interpreted as the assignment of subject~$t$, whereas the outcome of~$\pi_{n,n+1}$ is to be interpreted as the final recommendation. 

The input of~$\pi_{n,t}$ is denoted by~$(Z_{t-1}, G_t)$ and is recursively defined, where~$(Z_0, G_1) := (G_1)$ and one defines the~$2(t-1)$ dimensional random vector
\begin{equation}\label{eqn:zdef}
	Z_{t-1} = (Y_{\pi_{n, t-1}(Z_{t-2}, G_{t-1}), t-1}, G_{t-1}, \hdots, Y_{\pi_{n, 1}(G_{1}), 1}, G_1),
\end{equation}
i.e.,~$Z_{t-1}$ is the complete history of outcomes and randomizations observed before subject~$t$ arrives. Note that each assignment, and the final recommendation, can incorporate a draw from an exogenous random variable. 

The objective is to use the outcomes of the~$n$ subjects observed to give a recommendation~$\pi_{n,n+1}$ in~$\mathscr{M}_K$, such that~$\mathsf{T}(\langle \pi_{n,n+1}(Z_n, G_{n+1}), \mathbf{F} \rangle)$ is close to~$\max_{\delta \in \mathscr{M}_K} \mathsf{T}(\langle \delta, \mathbf{F} \rangle)$. Therefore, we will evaluate policies based on their \emph{regret}
\begin{equation}\label{eqn:reg2}
	r_n(\pi, {\mathscr{M}_K}) = r_n(\pi, \mathscr{M}_K; \mathbf{F}; Z_n, G_{n+1})  := \max_{\delta \in \mathscr{M}_K} \mathsf{T}(\langle \delta, \mathbf{F} \rangle) - \mathsf{T}(\langle \pi_{n,n+1}(Z_n, G_{n+1}), \mathbf{F} \rangle ).
\end{equation}
Note that~$r_n(\pi, {\mathscr{M}_K})$ measures the ``out-of-sample'' performance of the recommendation~$\pi_{n,n+1}(Z_n, G_{n+1})$. The regret we use in the present paper does not measure how ``well'' the~$n$ subjects used for obtaining the recommendation are assigned, it only incorporates the quality of the final recommendation. This makes sense in our framework, as~$n$ is assumed to be small compared to the whole population. In problems where~$n$ is relatively large, however, one may want to work with a different regret criterion also incorporating losses made during the experimentation phase. This problem is fundamentally different, and we refer the reader to~\cite{kpv1}, where an ``in-sample'' theory for a corresponding ``individual-specific regret criterion'' was developed. 

\section{Performance lower bounds}\label{sec:lowbound}

In this section we present lower bounds on maximal expected regret for policies as discussed in Section~\ref{sec:policies}. The lower bounds are given under weak conditions on the structure of~$\mathscr{M}_K$, allowing for rich forms of constraints. Besides delivering insights into the difficulty of the problem, the lower bounds will be instrumental in asserting that the specific policies introduced in later sections are optimal. 

To rule out trivial situations, we need an assumption which guarantees a minimal amount of variation of the functional~$\mathsf{T}$ evaluated on~$\mathscr{D}$ (if~$\mathsf{T}$ is constant on~$\mathscr{D}$ the regret of any policy is constant equal to~$0$, an uninteresting situation).
\begin{assumption}\label{as:lbcov}
	The functional~$\mathsf{T}: D_{cdf}([a,b]) \to \R$ satisfies Assumption~\ref{as:MAIN}, and~$\mathscr{D}$ contains two elements~$H_1$ and~$H_2$, such that 
	\begin{equation*}
		J_{\tau} := \tau H_1 + (1-\tau)H_2 \in \mathscr{D} \quad \text{ for every }\tau \in [0, 1],
	\end{equation*}
	and such that for some~$c_- > 0$ we have
	\begin{equation}\label{eqn:sublinASX}
		\mathsf{T}(J_{\tau_2}) - \mathsf{T}(J_{\tau_1}) \geq c_-(\tau_2-\tau_1) \quad \text{ for every }
		\tau_1 \leq \tau_2 \text{ in } [0, 1].
	\end{equation}
\end{assumption}
As already noted in~\cite{kpv2}, where this assumption was introduced, we emphasize that Equation~\eqref{eqn:sublinASX} in Assumption~\ref{as:lbcov} is satisfied if, e.g.,~$\tau \mapsto \mathsf{T}(J_{\tau})$ is continuously differentiable on~$[0,1]$ with an everywhere positive derivative. Note that Assumption~\ref{as:lbcov} is weak, as one is free to choose~$H_1$ and~$H_2$. In particular, it is typically satisfied as long as~$\mathscr{D}$ is reasonably large. For instance, as can be easily seen, this is the case for the examples given in Section~\ref{sec:funex}.

We start with a lower bound into which~$\mathscr{M}_K$ enters only via its squared diameter (with respect to the Euclidean norm~$\|\cdot\|$), which we shall abbreviate (under Assumption~\ref{as:M}) as
\begin{equation}
	\mathrm{diam}(\mathscr{M}_K) := \max\{\|\nu - \gamma\|: (\nu, \gamma) \in \mathscr{M}_K \times \mathscr{M}_K\}.
\end{equation}
Obviously,~$\mathrm{diam}(\mathscr{M}_K) = 0$ if and only if~$\mathscr{M}_K$ is a singleton, an uninteresting case which is ruled out in Assumption~\ref{as:M}. The lower bound is as follows.

\begin{theorem}\label{thm:lbn}
	Suppose Assumptions~\ref{as:dgp},~\ref{as:MAIN},~\ref{as:M} and~\ref{as:lbcov} hold. 
	Then there exists a constant~$c > 0$, independent of~$K$,~$n$ and~$\mathscr{M}_K$, such that for every policy~$\pi$ with recommendations in~$\mathscr{M}_K$, and any randomization measure~$\P_G$, it holds that
	\begin{equation}\label{eqn:rlowon}
		\sup_{F^1, \hdots, F^K \in \{J_{\tau}:\tau \in [0,1]\}} \mathbb{E}[r_n(\pi, {\mathscr{M}_K})] \geq c \times \mathrm{diam}^2(\mathscr{M}_K)/ \sqrt{n}, ~~\text{ for every } n \geq K,
	\end{equation}
	where the supremum is taken over all potential outcome vectors with independent marginals and cdfs in~$\{J_{\tau}:\tau \in [0,1]\}$.
\end{theorem}
The lower bound in Theorem~\ref{thm:lbn} actually holds over the \emph{parametric} subset~$\{J_{\tau}:\tau \in [0,1]\}$ (obtained from Assumption~\ref{as:lbcov}) of the potentially much larger and nonparametric set~$\mathscr{D}$. This also holds for all other bounds in this section. 

The main strength of Theorem~\ref{thm:lbn} is that it delivers a useful lower bound under \emph{minimal} assumptions on~$\mathscr{M}_K$: as soon as~$\mathscr{M}_K$ contains~$2$  elements that are bounded away from each other (uniformly over~$K$), the theorem shows that there exists no policy with maximal expected regret decreasing at a rate faster than~$1/\sqrt{n}$. The minimal assumption comes with a cost: namely that the theorem is silent about how the number of treatments affects the worst-case behavior of a policy. To make this more concrete, note for example that~$\mathscr{M}_2 = \{e_1(2), e_2(2)\}$ and~$\mathscr{M}_3 = \{e_1(3), e_2(3), e_3(3)\}$ both have diameter~$\sqrt{2}$, and thus both lead to the same lower bound in the previous theorem, even though~$\mathscr{M}_3$ obviously is more complex than~$\mathscr{M}_2$.

The remaining part of this section thus establishes lower bounds that also incorporate the dependence on~$K$ (in an optimal way, as the next section will show). It is clear that more refined structural properties of~$\mathscr{M}_K$ now need to enter the picture: the lower bounds need to reflect that the decision problem is more difficult if~$\mathscr{M}_K$ is very much ``spread out'' over~$\mathscr{S}_K$ instead of, e.g., being concentrated in a single corner of the simplex. 

To formulate the second result in this section, we now define the quantity~$\kappa(\mathscr{M}_K)$, which summarizes the structural properties of~$\mathscr{M}_K$ that enter into our general lower bound. The characteristic~$\kappa(\mathscr{M}_K)$ captures, in a way tailored towards our method of proof, how ``spread out'' the set~$\mathscr{M}_K$ is over the simplex, and (under Assumption~\ref{as:M}) is defined as
\begin{equation*}
	\kappa(\mathscr{M}_K) := \sup \left\{
	\left[
	\max_{\delta \in \mathscr{M}_K} v'\delta - \sup_{\delta \in \mathscr{T}} v'\delta\right] \wedge \min_{i = 1}^K
	\left[
	\max_{\delta \in \mathscr{M}_K}
	(v'\delta + w_i\delta_i)
	-
	\sup_{\delta \in  \mathscr{M}_K \setminus \mathscr{T}} (v'\delta + w_i\delta_i)
	\right]\right\},
\end{equation*}
where the outer supremum is taken over all nonempty Borel sets~$\mathscr{T} \subsetneqq \mathscr{M}_K$, all~$v \in [-1,1]^K$, and all~$w_i \in [-1,1]$ ($i = 1, \hdots, K$). It is easy to see  that~$\kappa(\mathscr{M}_K) \geq 0$ under Assumption~\ref{as:M}. 

Note that~$\kappa(\mathscr{M}_K)$ is large if there exists a vector~$v$, such that the maximum of~$\delta \mapsto v'\delta$ decreases substantially upon imposing the restriction~$\delta \in \mathscr{T}$, but where imposing the \emph{complementary} restriction~$\delta \notin \mathscr{T}$ decreases the maximum substantially upon suitably modifying \emph{any} single coordinate of~$v$. Since~$\kappa(\mathscr{M}_K)$ is a rather abstract quantity, the following result also provides two lower bounds for~$\kappa(\mathscr{M}_K)$ that are easy to interpret.
\begin{theorem}\label{Thm: UniformLowermixed}
	Suppose Assumptions~\ref{as:dgp},~\ref{as:MAIN},~\ref{as:M} and~\ref{as:lbcov} hold. 
	Then there exists a constant~$c > 0$, independent of~$K$,~$n$ and~$\mathscr{M}_K$, such that for every policy~$\pi$ with recommendations in~$\mathscr{M}_K$, and any randomization measure~$\P_G$, it holds that
	\begin{equation}\label{eqn:rlow3}
		\sup_{F^1, \hdots, F^K \in \{J_{\tau}:\tau \in [0,1]\}} \mathbb{E}[r_n(\pi, {\mathscr{M}_K})] \geq c  \kappa(\mathscr{M}_K) \sqrt{K/n}, ~~\text{ for every } n \geq K,
	\end{equation}
	where the supremum is taken over all potential outcome vectors with independent marginals and cdfs in~$\{J_{\tau}:\tau \in [0,1]\}$. Furthermore, 
	\begin{equation}\label{eqn:lq}
		\kappa(\mathscr{M}_K) \geq  \frac{1}{272} \times \max_{\delta \in \mathscr{M}_K} \|\delta\|^2 \times \min_{j = 1}^K (\max_{\delta \in \mathscr{M}_K} \delta_j - \min_{\delta \in \mathscr{M}_K} \delta_j)^3;
	\end{equation}
	and~$\kappa(\mathscr{M}_K)$ is also lower bounded by the squared positive part of~$\min_{j = 1}^K \max_{\delta \in \mathscr{M}_K} \delta_j - 1/2$. 
\end{theorem}
An important situation in Theorem~\ref{Thm: UniformLowermixed} is~$\mathscr{M}_K \supseteq \mathscr{E}_K$. In this case, the second lower bound on~$\kappa(\mathscr{M}_K)$ provided in the last sentence of the theorem implies~$\kappa(\mathscr{M}_K) \geq 1/4$. Together with Equation~\eqref{eqn:rlow3} this proves that no policy can have  maximal expected regret decreasing at a faster rate than~$\sqrt{K/n}$. The same is true for any~$\mathscr{M}_K$ as long as~$\min_{j = 1}^K \max_{\delta \in \mathscr{M}_K} \delta_j - 1/2$ is bounded away from~$0$ (uniformly in~$K$). 

For the special case in which the functional~$\mathsf{T}$ is the mean functional,~$a = 0$,~$b = 1$, and~$\mathscr{M}_K = \mathscr{E}_K$, a~$\sqrt{K/n}$ lower bound is given in the distribution-free lower bound mentioned in~\cite{bubeck2009pure}. The lower bound in Theorem~\ref{Thm: UniformLowermixed} is a non-trivial generalizations of the lower bound in~\cite{bubeck2009pure}: besides working with a general functional~$\mathsf{T}$ and operating over a corresponding line segment~$\{J_{\tau}:\tau \in [0,1]\}$ in~$\mathscr{D}$ along which~$\mathsf{T}$ varies sufficiently, we do not impose the condition~$\mathscr{M}_K = \mathscr{E}_K$. This makes the situation substantially more delicate, as the recommendation then takes its values in a potentially complicated set~$\mathscr{M}_K$. 

At this point, one may wonder whether the lower bound obtained in Theorem~\ref{Thm: UniformLowermixed} can actually be attained in case~$\mathscr{M}_K \neq \mathscr{E}_K$ for the following reason: Note that if~$\mathscr{M}_K = \mathscr{E}_K$ the decision maker ``only'' needs to learn the~$K$~values~$\mathsf{T}(F^1), \hdots,\mathsf{T}(F^K)$ to determine~$\argmax_{\delta \in \mathscr{M}_K} \mathsf{T}(\langle \delta, \mathbf{F} \rangle)$, whereas in case, e.g.,~$\mathscr{M}_K = \mathscr{S}_K$, one needs to learn~$\mathsf{T}(\langle \delta, \mathbf{F} \rangle)$ for \emph{all}~$\delta \in \mathscr{S}_K$.  In the next section we shall see that lower bounds of rate~$\sqrt{K/n}$ are attainable (up to logarithmic terms in~$K$). 

In situations where~$\mathscr{M}_K$ restricts a weight to~$\delta_j \leq 1/2$, and the discussion right after Theorem~\ref{Thm: UniformLowermixed} thus does not apply, the lower bound for~$\kappa(\mathscr{M}_K)$ given in Equation~\eqref{eqn:lq} (and a simple argument) shows that
\begin{equation*}
	\max_{\delta \in \mathscr{M}_K}\|\delta\|^2  \times
	\min_{j = 1}^K (\max_{\delta \in \mathscr{M}_K} \delta_j - \min_{\delta \in \mathscr{M}_K} \delta_j)^3 \geq \min_{j = 1}^K (\max_{\delta \in \mathscr{M}_K} \delta_j - \min_{\delta \in \mathscr{M}_K} \delta_j)^5;
\end{equation*}
as long as this lower bound is bounded away from~$0$ (uniformly in~$K$), no policy exists that has maximal expected regret decreasing at a faster rate than~$\sqrt{K/n}$. This is a substantial generalization beyond the case already discussed after Theorem~\ref{Thm: UniformLowermixed}, and in particular covers most cases of practical relevance where constraints are put on each treatment. Note in particular that if~$\mathscr{M}_K$ is such that~$\max_{\delta \in \mathscr{M}_K} \delta_j - \min_{\delta \in \mathscr{M}_K} \delta_j > 0$ for every~$j = 1, \hdots, K$, i.e., if every~$\delta_j$ can take at least two different values, then the lower bound for~$\kappa(\mathscr{M}_K)$ given in Equation~\eqref{eqn:lq} is strictly greater than~$0$. 

Note also that in case~$\max_{\delta \in \mathscr{M}_K} \delta_j - \min_{\delta \in \mathscr{M}_K} \delta_j = 0$ for a treatment~$j$, this treatment always has to be assigned with the same weight. In this case the lower bound in Equation~\eqref{eqn:lq} is~$0$, and furthermore also~$\kappa(\mathscr{M}_K) = 0$ holds.\footnote{That~$\kappa(\mathscr{M}_K) = 0$ if~$\max_{\delta \in \mathscr{M}_K} \delta_j - \min_{\delta \in \mathscr{M}_K} \delta_j = 0$ for treatment~$j$ follows upon noting that for any choice of~$\mathscr{T}$,~$v \in [-1,1]^K$, and~$w_j \in [-1,1]$, either~$\max_{\delta \in \mathscr{M}_K} v'\delta = \sup_{\delta \in \mathscr{T}} v'\delta$ or~$\max_{\delta\in \mathscr{M}_K} (v'\delta + w_j \delta_j) = \sup_{\delta \in \mathscr{M}_K \setminus \mathscr{T}} (v'\delta + w_j \delta_j)$.} Hence, Theorem~\ref{Thm: UniformLowermixed} does \emph{not} deliver a~$\sqrt{K/n}$ lower bound; it is nevertheless worth mentioning that Theorem~\ref{thm:lbn} still delivers a lower bound of order~$1/\sqrt{n}$ as long as~$\mathscr{M}_K$ is not a singleton. In this situation, one could expect a lower bound to hold true that incorporates, instead of~$K$, the number of treatments for which~$\max_{\delta \in \mathscr{M}_K} \delta_j - \min_{\delta \in \mathscr{M}_K} \delta_j \neq 0$. Such a result can be established using Theorem~\ref{prop: UniformLowermixed} in Appendix~\ref{sec:LB}, which fills in this ``gap'' between Theorems~\ref{thm:lbn} and~\ref{Thm: UniformLowermixed}, but comes at the expense of more complicated notation. Since this case is very special, we shall not discuss it further.

\section{Static assignment policies}\label{sec:nonseq}

We now consider a class of \emph{static assignment} (SA) policies incorporating an ``empirical-success'' recommendation rule. Given balanced assignments in the experimentation phase, these policies will be shown to be minimax expected regret optimal, in the sense that they attain the~$\sqrt{K/n}$ lower bound established in Theorem~\ref{Thm: UniformLowermixed} (up to~$\sqrt{\log(K)}$). In the following discussion we assume that~$n \geq K$, that~$\mathsf{T}$ is continuous on the convex set~$\mathscr{D}$, and that~$\mathscr{M}_K$ satisfies Assumption~\ref{as:M}.

An~SA policy proceeds as follows: First, it allocates subject~$t = 1, \hdots, n$ to treatment~$j$ according to whether~$t \in \Pi_{n, j}$, where~$\Pi_n := (\Pi_{n, 1}, \hdots, \Pi_{n, K})$ is a  partition of~$\{1, \hdots, n\}$. For simplicity, we treat the partition as fixed; random partitions that are obtained by an exogenous randomization mechanism can be easily accommodated via conditioning. Once all~$n$ outcomes are observed, one estimates every~$F^j$ by the corresponding empirical cdf
\begin{equation}
	\hat{F}_{n, \Pi_n}^j(\cdot) := |\Pi_{n, j}|^{-1} \sum_{t \in \Pi_{n, j}} \mathds{1} \{ Y_{j, t} \leq \cdot \},
\end{equation}
where~$|\Pi_{n, j}|$ denotes the cardinality of~$\Pi_{n, j}$. The obvious way to proceed would now be to search for a maximizer of~$\mathsf{T}(\langle \delta, \mathbf{\hat{F}}_{n, \Pi_{n}} \rangle)$ over~$\mathscr{M}_K$, where~$\mathbf{\hat{F}}_{n, \Pi_{n}} := (\hat{F}_{n, \Pi_n}^1, \hdots, \hat{F}_{n, \Pi_n}^K)$. However, without further assumptions on~$\mathsf{T}$, there is no guarantee that such a maximizer exists, because~$\delta \mapsto \mathsf{T}(\langle \delta, \mathbf{\hat{F}}_{n, \Pi_{n}} \rangle)$ need not be continuous (even under Assumption~\ref{as:MAIN}; noting that the coordinates of $\mathbf{\hat{F}}_{n, \Pi_{n}}$ are not necessarily elements of~$\mathscr{D}$). To avoid adding additional assumptions, we shall work with an approximate maximum, based on a discretization of~$\mathscr{M}_K$, i.e., a non-empty and finite set~$\mathscr{M}^{n}_K \subseteq \mathscr{M}_K$. Based on the discretization~$\mathscr{M}_K^n$ the SA policy then recommends 
\begin{equation}\label{eqn:selmin}
	\min \argmax_{\delta \in \mathscr{M}^{n}_K} \mathsf{T}(\langle \delta, \mathbf{\hat{F}}_{n, \Pi_n}\rangle),
\end{equation}
where the minimum in~\eqref{eqn:selmin} is taken as a concrete way of breaking ties and where the ambient set~$\mathscr{S}_K$ is equipped with the lexicographic order.\footnote{Here, for two vectors $a \neq b$ of real numbers, $a$ is lexicographically smaller than $b$ if $a_{i^*} < b_{i^*}$ for the smallest index $i^*$ such that $a_{i^*} \neq b_{i^*}$.} SA policies will be denoted by the generic symbol~$\hat{\pi}$, and are summarized in Policy~\ref{pol:es} for later reference. Note that the assignments of the subjects~$t = 1, \hdots, n$ neither depend on~$Z_{t-1}$ (i.e., the policy is static) nor incorporate external randomization~$G_t$, which are therefore dropped as arguments from the policy in the summary in Policy~\ref{pol:es}. Note also that the recommendation given does not incorporate an external randomization~$G_{n+1}$, which we therefore drop from the notation in Policy~\ref{pol:es} as well.
\begin{policy}\label{pol:es}
	\caption{Static Assignment Policy~$\hat{\pi}$}
	
	\smallskip
	
	\textbf{Input:}~$n\in \N$,~$K \leq n$, partition~$\Pi_n$ of~$\{1, \hdots, n\}$, discretization~$\mathscr{M}_K^n$ of~$\mathscr{M}_K$ \\
	
	\smallskip
	
	\For{$t = 1, \hdots,n$}{
		$\hat{\pi}_{n, t} = \sum_{i = 1}^K i \mathds{1}_{\Pi_{n,i}}(t)$ 
	}
	$\hat{\pi}_{n,n+1}(Z_n) = \min \argmax_{\delta \in \mathscr{M}^n_K} \mathsf{T}(\langle \delta, \mathbf{\hat{F}}_{n, \Pi_n}\rangle)$  
\end{policy}

\subsection{Choosing a discretization}\label{sec:disc}

Associate to any discretization~$\mathscr{M}_K^n \subseteq \mathscr{M}_K$ its (worst case) ``optimization error"
\begin{equation}\label{eqn:opte}
\varepsilon(n) = \varepsilon(n, \mathscr{D}, \mathscr{M}_K,\mathscr{M}_K^n) := \sup_{\mathbf{F} \in \mathscr{D} \times \hdots \times \mathscr{D}} \big|\max_{\delta \in \mathscr{M}_K} \mathsf{T}(\langle \delta, \mathbf{F} \rangle) - 
\max_{\delta \in \mathscr{M}_K^n} \mathsf{T}(\langle \delta, \mathbf{F} \rangle) \big|, 
\end{equation}
i.e., the maximal loss possible by optimizing over~$\mathscr{M}_K^n$ rather than~$\mathscr{M}_K$. In case~$\mathscr{M}_K$ has finitely many elements, one may choose~$\mathscr{M}_K^n = \mathscr{M}_K$ implying~$\varepsilon(n, \mathscr{D}, \mathscr{M}_K,\mathscr{M}_K^n) = 0$. If~$\mathscr{M}_K$ is not finite (and cannot be reduced to a finite set, cf.~Section~\ref{sec:qconv}), or if~$\mathscr{M}_K$ is finite but large, controlling the optimization error is often more conveniently done by choosing a fine enough discretization. Note that under Assumption~\ref{as:MAIN} and for~$\mathscr{D}$ convex 
\begin{equation}\label{eqn:rese}
	\varepsilon(n, \mathscr{D}, \mathscr{M}_K,\mathscr{M}_K^n) \leq CK \sup_{\delta \in \mathscr{M}_K} \inf_{\gamma \in \mathscr{M}_K^n}\|\delta - \gamma \|_{\infty}  =: CK \rho(\mathscr{M}_K, \mathscr{M}_K^n),
\end{equation}
and we refer to~$\rho(\mathscr{M}_K, \mathscr{M}_K^n)$ as the (worst-case) ``resolution error'' of the discretization~$\mathscr{M}_K^n$ of~$\mathscr{M}_K$. Since~$\mathscr{M}_K$ is assumed to be closed throughout, a finite~$\mathscr{M}^n_K \subseteq \mathscr{M}_K$ with~$\rho(\mathscr{M}_K, \mathscr{M}_K^n)$ (and hence~$\varepsilon(n)$) as small as one wishes always exists. To give a concrete example, a common way of constructing discretizations is discussed next for the unrestricted case~$\mathscr{M}_K = \mathscr{S}_K$ and for sets~$\mathscr{M}_K$ incorporating incompatibility constraints (but are otherwise unrestricted).
\begin{example} \label{ex:disc}
	Let~$\mathscr{M}_K$ be as in~\eqref{eqn:part} (potentially with~$m = 1$) with~$\mathscr{M}_{A_j, K} = \mathscr{S}_{A_j, K}$ for~$j = 1, \hdots, m$. For~$P \in \N$ define
	\begin{equation}\label{eqn:simpgrid}
		\mathscr{S}_{K, P} := \{
		\delta \in \mathscr{S}_K: P \delta \in \mathbb{Z}^K\},
	\end{equation}
	i.e., the intersection of the simplex~$\mathscr{S}_K$ and~$P^{-1} \mathbb{Z}^K$. Theorem~7 of~\cite{bomze2014rounding} implies that for every~$x \in \mathscr{S}_K$ there exists an~$y \in \mathscr{S}_{K, P}$, such that~$\|x-y\|_{\infty} \leq P^{-1}(1-1/K)$. The same theorem, but applied to the simplex~$\mathscr{S}_{A_i,K}$, shows that for every~$x \in \mathscr{S}_{A_i, K}$ there exists an~$y \in \mathscr{S}_{A_i, K} \cap \mathscr{S}_{K, P}$ such that~$\|x - y\|_{\infty} \leq P^{-1} (1-1/|A_i|)$. It thus follows that for every~$\rho_n > 0$ the discretization
	\begin{equation}\label{eqn:disc}
		\mathscr{M}_K^n := \mathscr{M}_K \cap \mathscr{S}_{K, \lceil \rho_n^{-1} \rceil} =  \bigcup_{i = 1}^m \mathscr{S}_{A_i, K} \cap \mathscr{S}_{K, \lceil \rho_n^{-1} \rceil}
	\end{equation}
	satisfies~$\rho(\mathscr{M}_K, \mathscr{M}_K^n) \leq \rho_n$ and~$\mathscr{M}_K^n \subseteq \mathscr{M}_K$.
\end{example}

\subsection{Maximal expected regret upper bound for~$\hat{\pi}$}

We now prove an upper bound on the maximal expected regret of~$\hat{\pi}$. One question we have not discussed so far is the measurability of~$\hat{\pi}$, which holds under the following weak condition as shown in the proof of Theorem~\ref{thm:NMAPup} in Appendix~\ref{app:es}.
\begin{assumption}\label{as:MB}
	For every~$m \in \N$, and every~$\delta \in \mathscr{S}_m$ the function on~$[a,b]^m$ defined via~$x \mapsto \mathsf{T}(\sum_{j = 1}^m \delta_j \mathds{1}\cbr[0]{x_j \leq \cdot})$ is Borel measurable.
\end{assumption}
\begin{theorem}\label{thm:NMAPup}
	Suppose Assumptions~\ref{as:dgp},~\ref{as:MAIN},~\ref{as:M}, and~\ref{as:MB} hold. Then the SA policy~$\hat{\pi}$ with partition~$\Pi_n$,~discretization~$\mathscr{M}_K^n$ and~$\beta_n := \min_{j = 1, \hdots, K}|\Pi_{n,j}|$ satisfies
	\begin{equation}\label{eqn:rup2}
		\sup_{\substack{F^i \in \mathscr{D} \\ i = 1, \hdots, K }} \mathbb{E}\left[r_n(\hat{\pi}, \mathscr{M}_K)\right] \leq \varepsilon(n) + 3.01 \times C ~\sqrt{
			\log(K)	/\beta_n}, \text{ for every } n \geq K.
	\end{equation}
\end{theorem}
Theorem~\ref{thm:NMAPup} shows that working with a discretization having maximal optimization error~$\varepsilon(n) \leq \sqrt{K/n}$ (cf.~Equation~\eqref{eqn:opte}) together with ``balanced partitions,'' i.e., partitions satisfying~$|\Pi_{n,j}| \geq \lfloor n/K \rfloor$ for every~$n \in \N$ and every~$j = 1, \hdots, K$, results in an SA policy that attains the lower bound established in Theorem~\ref{Thm: UniformLowermixed}, up to a~$\sqrt{\log(K)}$ factor (and multiplicative constants). Hence, for such a choice of partition and discretization, the SA policy is (near) minimax expected regret optimal. Note that the optimality statement also holds in the high-dimensional regime where the number of treatments grows with~$n$. By Equation~\eqref{eqn:rese}, a discretization satisfying~$\varepsilon(n) \leq \sqrt{K/n}$ can be obtained by choosing~$\mathscr{M}_K^n$ such that its resolution error~$\rho(\mathscr{M}_K, \mathscr{M}_K^n) \leq \sqrt{K/n}/(CK)$, cf.~Example~\ref{ex:disc} for specific constructions. 

In the special case where~$\mathsf{T}$ is the mean functional (which is quasi-convex) and~$\mathscr{M}_K = \mathscr{E}_K$, the SA policy~$\hat{\pi}$ with~$\mathscr{M}_K^{n} = \mathscr{E}_K$ and balanced assignment reduces to the ``uniform allocation policy'' in~\cite{bubeck2009pure} and to the empirical-success rule studied in~\cite{manski2004statistical} and~\cite{Manski10518}. In this special case Assumption~\ref{as:MAIN} is satisfied with~$C = b-a$,~$\varepsilon(n) = 0$, and Theorem~\ref{thm:NMAPup} delivers and upper bound with the same dependence on~$K$ and~$n$ as the upper bound in the just mentioned articles, but with an additional factor~3.01. This additional factor is due to our different method of proof that needs to deal with the possibility that in general~$\mathscr{E}_K \neq \mathscr{M}_K$ (and that~$\mathsf{T}$ is not the mean functional).

\section{Sequential elimination policies}\label{sec:seq}
The static assignment policy~$\hat{\pi}$ assigns each subject~$t =1, \hdots, n$ to a  treatment in~$\mathcal{I}$ according to the partition~$\Pi_n$. This partition is fixed by the decision maker before the first assignment is made. Thus, the policy~$\hat{\pi}$ is static. In the present section we consider policies that are sequential, in the sense that the assignment of subject~$t \in \{1, \hdots, n\}$ can depend on the outcomes of all previously assigned subjects~$1, \hdots, t-1$. Intuitively, this opens up the following opportunity: by sequentially monitoring the performance of the treatments, one can target the sampling effort to where it is most useful, rather than deciding up front to assign each treatment, e.g., equally often.

We already know from Theorem~\ref{thm:NMAPup} that the maximal expected regret lower bound from Theorem~\ref{Thm: UniformLowermixed} is attainable (up to a~$\sqrt{\log(K)}$ term) in the class of static policies. Therefore, not much can be gained in terms of worst-case expected regret from sequential policies. Nevertheless, it is plausible that a policy which exploits that subjects arrive sequentially can improve on the static policy~$\hat{\pi}$ for many potential outcome distributions~$\mathbf{F} = (F^1, \hdots, F^K) \in \mathscr{D} \times \hdots \times \mathscr{D}$, without having a higher worst-case expected regret. For the policies suggested in the present article this is confirmed by Theorems~\ref{thm:FSApartition} and~\ref{thm:FSAcheck} and the numerical results in Section~\ref{sec:num}.

Essentially, we propose sequential policies which are based on the following rationale: \emph{stop assigning ``inferior'' treatments as soon as possible, and do not include treatments in the recommendation once eliminated}. While it is clear in the special case of~$\mathscr{M}_K = \mathscr{E}_K$ that all treatments~$i$ not contained in~$\argmax_{j \in \mathcal{I}} \mathsf{T}(F^j)$ are inferior, it is less clear in the general case what an inferior treatment is. We shall consider treatment~$i$ to be inferior, if \emph{all} elements of~$\argmax_{\delta \in \mathscr{M}_K} \mathsf{T}(\langle \delta, \mathbf{F} \rangle)$ have zero~$i$-th coordinate, i.e., if
\begin{equation}\label{eqn:top1}
	\argmax_{\delta \in \mathscr{M}_K} \mathsf{T}(\langle \delta, \mathbf{F} \rangle) \cap \{\delta \in \mathscr{M}_K: \delta_i > 0\} = \emptyset.
\end{equation}
In this case treatment~$i$ does not contribute to any maximizer. Note that the maximum in the previous display is well defined if~$\mathsf{T}$ is continuous on~$\mathscr{D}$, convex, and~$\mathscr{M}_K$ is as in Assumption~\ref{as:M}, which we shall assume throughout this section. Since~$\mathbf{F}$ is unknown and~$\mathsf{T}$ is continuous, estimation uncertainty implies that an inferior treatment~$i$ is only empirically detectable if it is actually ``strongly inferior'' in the sense that
\begin{equation}\label{eqn:intbd}
	\argmax_{\delta \in \mathscr{M}_K} \mathsf{T}(\langle \delta, \mathbf{F} \rangle) \cap \overline{\{\delta \in \mathscr{M}_K: \delta_i > 0\}} = \emptyset
\end{equation}
(that is~$\{\delta \in \mathscr{M}_K: \delta_i > 0\}$ in~\eqref{eqn:top1} is replaced by its closure). It is easy to see that structural assumptions need to be put on~$\mathscr{M}_K$ in order to guarantee that strongly inferior treatments exist for some~$\mathbf{F}$. As an example, no strongly inferior treatments exist for~$\mathscr{M}_K = \mathscr{S}_K$, since then~$\overline{\{\delta \in \mathscr{S}_K : \delta_i > 0\}} = \mathscr{S}_K$ for every~$i$. In such cases, attempting  to eliminate treatments based on assessing~\eqref{eqn:intbd} can never result in efficiency gains over~$\hat{\pi}$, because~\eqref{eqn:intbd} is then never satisfied.

In the construction of our sequential policies it thus only makes sense to consider~$\mathscr{M}_K$ that do not a priori rule out the existence of strongly inferior treatments. We shall therefore mainly focus on~$\mathscr{M}_K = \bigcup_{j = 1}^m \mathscr{M}_{A_j, K}$ for a partition~$A_1, \hdots, A_m$ of~$\mathcal{I}$ as in Example~\ref{ex:compa}, which covers many situations of practical relevance (a sequential policy that also works without this structural assumption is considered in Section~\ref{sec:beyond}).  Since~$\mathsf{T}$ is continuous on~$\mathscr{D}$, convex, the condition
\begin{equation}\label{eqn:maxsep}
	\max_{\delta \in \mathscr{M}_K} \mathsf{T}(\langle \delta, \mathbf{F} \rangle) = \max_{\delta \in \bigcup_{l \neq j} \mathscr{M}_{A_l, K}} \mathsf{T}(\langle \delta, \mathbf{F} \rangle) > \max_{\delta \in  \mathscr{M}_{A_j, K}} \mathsf{T}(\langle \delta, \mathbf{F} \rangle)
\end{equation}
is then sufficient for \emph{all} treatments in~$A_j$ to be strongly inferior, because the closure of~$\mathscr{M}_{A_j, K}$ is contained in~$\mathscr{S}_{A_j, K}$, which is disjoint with~$\bigcup_{l \neq j} \mathscr{M}_{A_l, K}$.
%
In the special case~$\mathscr{M}_K = \mathscr{E}_K$, amounting to~$m = K$, the relation in~\eqref{eqn:maxsep} is actually equivalent to
\begin{equation}\label{eqn:maxsepE}
	\max_{i \in \mathcal{I}} \mathsf{T}(F^i) = \max_{l \neq j} \mathsf{T}(F^l) > \mathsf{T}(F^j).
\end{equation}

The policies introduced in the present section eliminate strongly inferior treatments based on verifying an empirical equivalent of~\eqref{eqn:maxsep}. Once the data firmly suggest~\eqref{eqn:maxsep},  all treatments in~$A_j$ are eliminated. Our policy carefully needs to avoid concluding~\eqref{eqn:maxsep} prematurely, in case this inequality is false. Whether or not the validity of~\eqref{eqn:maxsep} can actually be confirmed depends on the difference between the two maxima in~\eqref{eqn:maxsep} in relation to sample size. In case all treatments in~$A_1$ and all treatments in~$A_2$ are strongly inferior, it can for example happen that~\eqref{eqn:maxsep} can be concluded for~$j = 1$ at a much earlier stage than for~$j = 2$. Therefore, in a sequential policy, the set of indices over which~\eqref{eqn:maxsep} is assessed evolves during the sampling process, and depends on the previously eliminated treatments.

Before we proceed to the policies we suggest, we need some notation. Given a policy~$\pi$,~$n \in \N$ and~$t = 1, \hdots, n$, we denote the number of times treatment~$i$ has been assigned up to time~$t$ by
\begin{equation}\label{eqn:Sintdef}
	S_{i,n}(t):= \sum_{s = 1}^t \mathds{1}\{ \pi_{n,s}(Z_{s-1}, G_s) = i\}.
\end{equation}
On the event~$\{S_{i,n}(t) > 0\}$ we define the empirical cdf based on the outcomes of all subjects in~$\{1, \hdots, t\}$ that have been assigned to treatment~$i$
\begin{equation}\label{eq:Fhat}
	\hat{F}_{i, t, n}(z) := S^{-1}_{i,n}(t) \sum_{\substack{1 \leq s \leq t \\
			\pi_{n, s}(Z_{s-1}, G_s) = i
	}} \mathds{1}\{Y_{i, s} \leq z\}, \quad \text{ for every } z \in \R;
\end{equation}
we leave~$\hat{F}_{i, t, n}$ undefined on~$\{S_{i,n}(t) = 0\}$. Note that the random sampling times~$s$ such that~$\pi_{n, s}(Z_{s-1}, G_s) = i$ depend on previously observed outcomes. Finally, we define~$\hat{\mathbf{F}}_{t,n}=(\hat{F}_{1, t, n},\hdots,\hat{F}_{K, t, n})$ on the event~$\{S_{i,n}(t) > 0 \text{ for every } i = 1, \hdots, K\}$.

Equipped with this notation, we can now introduce and discuss sequential policies. To ease the exposition, we start with the simplest case~$\mathscr{M}_K = \mathscr{E}_K$, which is then extended to all~$\mathscr{M}_K$ as in Example~\ref{ex:compa}. Finally, we shall go beyond incompatibility constraints and consider a sequential policy, which can be applied without $\mathscr{M}_K$ possessing any particular structure (whether anything can be gained in practice compared to a non-sequential policy then depends on whether the set $\mathscr{M}_k$ allows for the existence of strongly inferior treatments). 

\subsection{The case~$\mathscr{M}_K=\mathscr{E}_K$}

We shall refer to the sequential policy introduced as the \emph{sequential elimination} (SE) policy~$\tilde{\pi}$. In the SE policy, the treatments are assigned in rounds. In every round~$r$, all treatments~$\mathcal{I}_{r-1} \subseteq \mathcal{I}$, say, that have not been eliminated in one of the previous rounds are assigned exactly once. Elimination is based on checking whether the data observed so far firmly suggest Equation~\eqref{eqn:maxsepE}. This is done as follows: at the end of round~$r$ we eliminate all treatments~$i \in \mathcal{I}_{r-1}$ for which
\begin{equation}\label{eq:FSEp}
	\max_{j \in \mathcal{I}_{r-1}} \mathsf{T}(\hat{F}_{j,t,n}) > \mathsf{T}(\hat{F}_{i,t,n}) + u_\eta(r,n);
\end{equation}
i.e., one defines~$\mathcal{I}_r$ as the subset of elements of~$\mathcal{I}_{r-1}$ that do not satisfy~\eqref{eq:FSEp}. Here~$u_\eta(r,n):\N^2\to[0,\infty)$ is a threshold function giving a ``critical value'' determining whether a treatment can be eliminated. We shall base our policies upon the threshold 
\begin{equation}
	\label{eqn:ueta}u_\eta(r,n)=C\sqrt{\frac{1+\eta}{r}\left[0.5\log(n)+\log(rK)\right]}, \text{ for a tuning parameter } \eta>0,
\end{equation}
and where~$C$ is the constant from~Assumption~\ref{as:MAIN}. Note that the more rounds~$r$ have been completed, the better the unknown cdfs can be estimated, and thus the more aggressively~$\tilde{\pi}$ eliminates treatments. 

The decision maker may not want to check~\eqref{eq:FSEp} and update~$\mathcal{I}_r$ after every single round. For example, subjects may arrive in batches. We allow for this by assuming that the decision maker a priori decides on a subset of ``elimination rounds"~$\mathcal{R}\subseteq\cbr[0]{1,\hdots,\lfloor \frac{n}{2}\rfloor}$, after which~\eqref{eq:FSEp} is assessed, and treatments are potentially eliminated. Note that there are at most~$\lfloor n/2 \rfloor$ rounds at the beginning of which there are at least two treatments left. We shall denote by~$\underline{r}:=\min \mathcal{R}$ the first round after which elimination takes place, and assume throughout that~$\underline{r} \leq \lfloor n/K \rfloor$, since otherwise an elimination round can never be reached. To provide some examples and motivation consider the following relevant special cases:
\begin{enumerate}
	\item $\mathcal{R}=\cbr[0]{1,\hdots,\lfloor \frac{n}{2}\rfloor}$: This amounts to checking after every round whether any of the remaining treatments can be eliminated.
	\item $\mathcal{R}=\cbr[0]{\underline{r},\hdots,\lfloor \frac{n}{2}\rfloor}$ for some $1 < \underline{r} \leq \lfloor n/K \rfloor$: This amounts to allowing for a ``burn-in'' phase before checking (after every subsequent round) whether treatments can be eliminated. This is relevant in case one does not want to eliminate any treatments based on very few assignments. 
	\item $\mathcal{R}=k\N\cap\cbr[0]{1,\hdots,\lfloor \frac{n}{2}\rfloor}$ for some~$k\in\cbr[0]{1,\hdots,\lfloor \frac{n}{K}\rfloor}$: This choice amounts to checking after every~$k$th round whether treatments can be eliminated. 
\end{enumerate}

A detailed description of the policy~$\tilde{\pi}$, concretizing the explanation above, is given in Policy \ref{pol:FSE_corners}. Here, we also need to take care of the possibility that after round~$r$ has been completed, there might be less subjects left than treatments in~$\mathcal{I}_r$, not allowing for a further complete round of assignments. This is why the policy is separated in an outer ``while''-loop and an outer if expression. As long as enough subjects are left to assign all remaining treatments once, the ``while''-loop proceeds in rounds as discussed above, where after each round in~$\mathcal{R}$ treatments may be eliminated. Once there are more treatments than subjects left, all remaining subjects (if there are any left) are assigned to a subset of the remaining treatments. This happens within the outer if expression. Finally, the recommendation is based on an empirical-success rule, where the minimum is taken as a concrete way of breaking ties. The policy does not use external randomization~$G_t$, which is therefore suppressed in the notation. 
\begin{policy}\label{pol:FSE_corners}
	\caption{Sequential Elimination Policy~$\tilde{\pi}$}
	
	\smallskip
	
	\textbf{Input:}~$n \in \N$,~$K \leq n$, set of elimination rounds~$\mathcal{R}\subseteq\cbr[0]{1,\hdots,\lfloor \frac{n}{2}\rfloor}$,~$\eta>0$\\
	
	\smallskip
	
	\textbf{Set:}~$t\leftarrow0$,~$r\leftarrow0$,~$\mathcal{I}_0 \leftarrow \mathcal{I}$
	
	\smallskip
	
	\While{$n-t\geq |\mathcal{I}_r|$}{
		\For{$i\in\mathcal{I}_r$}{
			$t\leftarrow t+1$ \\
			$\tilde{\pi}_{n,t}(Z_{t-1})=i$ 
		}
		$r\leftarrow r+1$\\
		\If{$r\in \mathcal{R}$}{
			$\mathcal{I}_{r}\leftarrow  \{j\in\mathcal{I}_{r-1}:\max_{i\in\mathcal{I}_{r-1}}\mathsf{T}(\hat{F}_{i,t,n})-\mathsf{T}(\hat{F}_{j,t,n})\leq u_\eta(r,n)\}$
		}
		\Else{$\mathcal{I}_r\leftarrow \mathcal{I}_{r-1}$}
	}
	\If{$t<n$}{
		Let~$(j_{1}, \hdots, j_{|\mathcal{I}_r|})$ be the elements of~$\mathcal{I}_r$, ordered from smallest to largest.
		
		\For{$i=1,\hdots,(n-t)$}{$\tilde{\pi}_{n, t+i}(Z_{t+i-1})= j_i$}
	}
	$\tilde{\pi}_{n,n+1}(Z_n)=\min\argmax \cbr[1]{\mathsf{T}(\hat{F}_{i,n,n}):i\in \mathcal{I}_{r}}$
\end{policy}

Most importantly, our theoretical results in the general case in the next section imply that Policy~\ref{pol:FSE_corners} is minimax regret optimal up to logarithmic factors. This follows from the upper bound on expected maximal regret established in Theorem \ref{thm:FSApartition} together with the lower bound established in Theorem~\ref{Thm: UniformLowermixed}. We abstain from formulating a theorem for the special case~$\mathscr{M}_K = \mathscr{E}_K$. 

In pure exploration problems targeting the mean and~$\mathscr{M}_K=\mathscr{E}_K$, sequential policies have been studied in~\cite{audibert2010best} and \cite{karnin2013almost}. Furthermore, \cite{tran2014functional} studied these policies for a family of quasi-convex functionals in the case of~$\mathscr{M}_K=\mathscr{E}_K$. The just mentioned policies, however, fix a set of elimination times upfront at which a pre-specified (nonzero) number of treatments \emph{must} be eliminated. Policy~\ref{pol:FSE_corners}, on the other hand, decides in a data-driven way if and when a treatment may be eliminated.

\subsection{General incompatibility constraints}\label{sec:FSEm}

We shall now extend the policy from~$\mathscr{M}_K = \mathscr{E}_K$ to the more general case where~$\mathscr{M}_K = \bigcup_{j = 1}^m \mathscr{M}_{A_j, K}$ for a partition~$A_1, \hdots, A_m$ of~$\mathcal{I}$, cf.~Example~\ref{ex:compa}. In this more general case, the policy proceeds in rounds in the same way as detailed in the previous section. Again, the decision maker needs to a priori specify a subset of ``elimination rounds''~$\mathcal{R}$, in which treatments can potentially be eliminated. Checking whether treatments can be eliminated is a bit more involved and explained in the following.

Recall from the discussion around~\eqref{eqn:maxsep} that the treatments in~$A_j$ are strongly inferior if 
\begin{equation}\label{eqn:wtc}
	\max_{\delta \in \bigcup_{l \neq j} \mathscr{M}_{A_l, K}} \mathsf{T}(\langle \delta, \mathbf{F} \rangle) > \max_{\delta \in  \mathscr{M}_{A_j, K}} \mathsf{T}(\langle \delta, \mathbf{F} \rangle).
\end{equation}
In order to empirically check whether this inequality holds (and whether the set of treatments~$A_j$ can be eliminated) we again rely on a discretization of~$\mathscr{M}_K$. To this end fix a discretization~$\mathscr{M}_K^n = \bigcup_{j =1}^m \mathscr{M}_{A_j, K}^n$ such that~$\emptyset \neq \mathscr{M}_{A_j, K}^n \subseteq \mathscr{M}_{A_j,K}$ for every~$j = 1, \hdots, m$.
%
Given~$\mathscr{M}_K^n$ of this form, our policy aims to verify whether Equation~\eqref{eqn:wtc} holds by checking its discretized version
\begin{equation}\label{eqn:elim2} 
	\max_{\delta \in \bigcup_{l \neq j} \mathscr{M}_{A_l, K}^n} \mathsf{T}(\langle \delta, \mathbf{F} \rangle) > \max_{\delta \in  \mathscr{M}_{A_j, K}^n} \mathsf{T}(\langle \delta, \mathbf{F} \rangle).
\end{equation}
Intuitively, we empirically check this condition in two steps: (i) for every~$j$ and in every elimination round we ``remove'' those elements~$\gamma \in \mathscr{M}_{A_j, K}^n$, for which the data firmly suggest~$\max_{\delta \in \bigcup_{l \neq j} \mathscr{M}_{A_l, K}^n} \mathsf{T}(\langle \delta, \mathbf{F} \rangle) >  \mathsf{T}(\langle \gamma, \mathbf{F} \rangle)$; and (ii) we eliminate all treatments in~$A_j$, once all~$\gamma \in \mathscr{M}_{A_j, K}^n$ have been~removed. As in the previous subsection, if a treatment is eliminated, it is no longer assigned, and does not contribute to the final recommendation. 

To make this two-step check more precise, we denote by~$\mathscr{M}_{A_j, K, r}^n$ the subset of elements of~$\mathscr{M}_{A_j, K}^n$, which have not been removed after~$r$ rounds. After each elimination round~$r \in \mathcal{R}$, we remove all~$\gamma \in \mathscr{M}_{A_j, K, r}^n$ for which
\begin{equation}\label{eq:FSEpmixing}
	\max_{\delta \in \bigcup_{l \neq j} \mathscr{M}_{A_l, K, r}^n} \mathsf{T}(\langle \delta, \hat{\mathbf{F}}_{t,n} \rangle) >  \mathsf{T}(\langle \gamma, \hat{\mathbf{F}}_{t,n} \rangle) + u_{\eta}(r, n),
\end{equation}
where~$u_\eta(r,n)$ is as defined in Equation~\eqref{eqn:ueta}, and~$\hat{\mathbf{F}}_{t,n}$ was defined in Equation~\eqref{eq:Fhat} (in case the set over which the maximum in the previous display is taken is empty, no removal takes place). If \emph{all} elements in~$\mathscr{M}_{K,A_j}^n$ are removed, we eliminate \emph{all} treatments in~$A_j$, and no longer assign them.

A detailed description of~$\tilde{\pi}$ (for general partitions) is given in Policy \ref{pol:FSE_partition}. The structure is the same as in Policy~\ref{pol:FSE_corners}, but elimination is based on the two-step procedure just explained. The recommendation is again based on an empirical-success rule. Furthermore, the policy does not depend on external randomization, which is therefore dropped from the notation. It is easy to verify that for~$m = K$, i.e.,~$\mathscr{M}_K = \mathscr{E}_K$, Policy~\ref{pol:FSE_partition} reduces to Policy~\ref{pol:FSE_corners}.  
\begin{policy}[t!]\label{pol:FSE_partition}
	\caption{Sequential Elimination~$\tilde{\pi}$}
	
	\smallskip
	
	\textbf{Input:}~$n \in \N$,~$K \leq n$, partition~$\cbr[0]{A_1,\hdots,A_m}$ of~$\mathcal{I}$ such that~$m\geq 2$, discretization~$\mathscr{M}_K^n = \bigcup_{j =1}^m \mathscr{M}_{A_j, K}^n$, such that~$\emptyset \neq \mathscr{M}_{A_j, K}^n \subseteq \mathscr{M}_{A_j,K}$ for every~$j = 1, \hdots, m$, set of elimination rounds~$\mathcal{R}$,~$\eta>0$
	
	\smallskip
	
	\textbf{Set:}~$t\leftarrow 0$,~$r\leftarrow0$,~$\mathcal{I}_0\leftarrow\mathcal{I}$,~$\mathcal{J}_0 \leftarrow \{1, \hdots, m\}$, 
	$\mathscr{M}_{A_j, K, 0}^{n} \leftarrow \mathscr{M}_{A_j, K}^n$ for~$j = 1, \hdots, m$.
	
	\smallskip
	
	\While{$n-t\geq |\mathcal{I}_r|$}{
		
		\For{$i\in\mathcal{I}_r$}{
			$t\leftarrow t+1$ \\
			$\tilde{\pi}_{n,t}(Z_{t-1})=i$ 
		}
		$r\leftarrow r+1$\\
		$\mathcal{I}_r\leftarrow \mathcal{I}_{r-1}$ \\
		$\mathcal{J}_r\leftarrow \mathcal{J}_{r-1}$ \\
		
		\For{$j \in \mathcal{J}_r$}{
			\If{$r\in\mathcal{R}$ and~$|\mathcal{J}_r| \geq 2$}{
				\smallskip
				
				$h \leftarrow \max\{\mathsf{T}(\langle \delta, \mathbf{\hat{F}}_{t,n}\rangle):  \delta\in \mathscr{M}_{A_l, K, r-1}^n, l \in \mathcal{J}_r \backslash \{j\}\}$ 
				
				\smallskip
				
				$\mathscr{M}_{A_j, K, r}^n \leftarrow  \cbr[1]{\gamma\in\mathscr{M}_{A_j, K, r-1}^n: h-\mathsf{T}(\langle \gamma, \mathbf{\hat{F}}_{t,n}\rangle)\leq u_\eta(r,n)}$
				
				\smallskip
				
				\If{$\mathscr{M}_{A_j, K, r}^n = \emptyset$}{$\mathcal{I}_r \leftarrow \mathcal{I}_r \backslash A_j$\\
					$\mathcal{J}_r \leftarrow \mathcal{J}_r \backslash \{j\}$ }
			}
			\Else{
				$\mathscr{M}_{A_j, K, r}^n \leftarrow \mathscr{M}_{A_j, K, r-1}^n$}
		}

	}
	
	\If{$t<n$}{
		Let~$(j_{1}, \hdots, j_{|\mathcal{I}_r|})$ be the elements of~$\mathcal{I}_r$, ordered from smallest to largest.
		
		\For{$i=1,\hdots,(n-t)$}{$\tilde{\pi}_{n, t+i}(Z_{t+i-1})= j_i$}
	}
	$\tilde{\pi}_{n,n+1}(Z_n)=\min\argmax \cbr[1]{\mathsf{T}(\langle\delta,\hat{\mathbf{F}}_{n,n}\rangle):\delta\in \bigcup_{j \in \mathcal{J}_r} \mathscr{M}^n_{A_j, K, r}}$
\end{policy}

Now that we have defined the policy in the general case, we are ready to state a theorem regarding its maximal expected regret. The sequential nature of the SE policy makes the analysis more involved than that of the SA policy, which leads to a slightly more complicated upper bound than the one from Theorem~\ref{thm:NMAPup}.
\begin{theorem}\label{thm:FSApartition}Suppose Assumptions \ref{as:dgp}, \ref{as:MAIN}, and \ref{as:MB} hold, and that~$\mathscr{M}_K$ is as in Example~\ref{ex:compa} and satisfies Assumption~\ref{as:M}. Then, the SE policy~$\tilde{\pi}$ as in Policy~\ref{pol:FSE_partition} satisfies
	\begin{align*}
		&\sup_{\substack{F^i \in \mathscr{D} \\ i = 1, \hdots, K }}\E \left[ r_n(\tilde{\pi}, \mathscr{M}_K) \right]
		\leq \eps(n) + 4C \sqrt{\frac{K}{n}} D_{n,K} + 4C\frac{K^{1-\frac{3\eta}{4+2\eta}} }{\sqrt{n}} \times (1+18.5\eta^{-1})^2, ~\forall n \geq 2K,
	\end{align*}
	where~$D_{n,K} :=  \sqrt{K\log\del[0]{17K}} \wedge  \sqrt{\log(3.01 \times K(n+2))}$.
\end{theorem}

The upper bound on maximal expected regret is a sum of three terms. The first term is the optimization error resulting from targeting the best recommendation in~$\mathscr{M}_K^n$, rather than the best recommendation in~$\mathscr{M}_K$. Inspection of the proof shows that the second term is an upper bound on the expected regret on the event where not all maximizers of~$\delta\mapsto\mathsf{T}(\langle \delta, \mathbf{F}\rangle)$ over~$\mathscr{M}_K^n$ have been eliminated after~$n$ assignments, and that the last term upper bounds the probability that all maximizers of~$\delta\mapsto\mathsf{T}(\langle \delta, \mathbf{F}\rangle)$ over~$\mathscr{M}_K^n$ are eliminated by~$\tilde{\pi}$ after~$n$ assignments. We note that in Theorem~\ref{thm:FSApartitionstg} in Appendix~\ref{sec:SEP} we prove a slightly more general result, from which Theorem~\ref{thm:FSApartition} follows. The generality is bought at the expense of more complicated notation, but it provides a stronger bound valid for all~$n > K$, which shows in more detail how the choice of~$\underline{r}$ affects the upper bound. 

For~$K$ fixed and any choice of~$\eta > 0$, the bound in Theorem \ref{thm:FSApartition} achieves the minimax optimal dependence in~$n$ if one chooses a discretization with optimization error~$\eps(n)$ of no larger order than~$n^{-1/2}$ (cf.~the discussion after Theorem \ref{Thm: UniformLowermixed}). For regimes where~$K$ grows with~$n$, the dependence of the upper bound in~Theorem \ref{thm:FSApartition} on~$n$ and~$K$ is optimal (up to a factor of order~$\sqrt{\log(nK)}$) as long as~$\eta \geq 1$. 

Although the policy~$\tilde{\pi}$ cannot improve on~$\hat{\pi}$ in terms of \emph{worst-case} expected regret (and there is not much room for improvement given our lower bound results), there are many empirically relevant potential outcome distributions for which the policy~$\tilde{\pi}$ is more efficient than~$\hat{\pi}$. 
 A prominent example of such distributions is the following:
\begin{example}
 For simplicity, assume that the functional of interest~$\mathsf{T}$ is quasi-convex and that~$\mathscr{M}_K=\mathscr{E}_K$ such that no mixtures need to be studied. Consider~$K>2$ potential outcome distributions~$F^1,\hdots,F^K$ for which it holds that
	\begin{align*}
		\mathsf{T}(F^1)\approx \mathsf{T}(F^2)\qquad\text{and}\qquad 
		\mathsf{T}(F^3)=\hdots =\mathsf{T}(F^K)\ll \mathsf{T}(F^1),
	\end{align*}
	that is there are two top treatments (Treatments 1 and 2) that are approximately equally good and~$K-2$ very inferior treatments. Instead of assigning all~$K$ treatments~$n/K$ times as in the basic version of the static policy (assuming here that~$n/K$ is an integer), it is often possible to identify that Treatments~$3$ to~$K$ are inferior based on much fewer than~$n/K$ assignments to each of these. This allows one to assign Treatments 1 and 2 more often, which is beneficial since these are the two real candidates for the best treatment. In particular if~$K$ is large this can lead to better recommendations as one can sample much more often from the two distributions that are difficult to distinguish than in the static policy. Indeed, this is the idea underlying the third simulation setting in Section~\ref{sec:Ginisim} where the outcome distribution with maximal Gini-welfare is targeted. There, it is illustrated numerically that in such a setting the sequential elimination policy can result in much lower regret than the static assignment policy.
\end{example}

\subsection{Beyond incompatibility constraints}\label{sec:beyond}

So far, we have considered the case where~$\mathscr{M}_K$ satisfies an incompatibility constraint, which we exploited in the construction of the sequential policy. As our final investigation, we now ask whether a sequential policy with optimal expected regret guarantees can also be designed without exploiting that type of constraint. Recall that we have already argued that strongly inferior treatments should exist in order that a sequential policy can eliminate treatments. Because~$\mathscr{M}_K$ is known to the decision maker, this can be checked before deciding on which policy to use. 

\begin{example}
One example where strongly inferior treatments can exist, without an incompatibility constraint being satisfied, are capacity constrained sets~$\mathscr{M}_K$, where, for some index~$i$, it holds for every~$\delta \in \mathscr{M}_K$ that either~$\delta_i = 0$ or~$\delta_i \geq \epsilon$ for some~$\epsilon \in (0, 1)$. This corresponds to ``entry cost" restrictions, where treatment~$i$ can only be rolled out if a certain minimal proportion~$\epsilon$ is assigned to that treatment. Then, if all optimal mixture vectors are such that their~$i$-th coordinate is~$0$, treatment~$i$ is clearly strongly inferior. An illustration of such a situation in comparison to an incompatibility constraint is given in Figure~\ref{fig:exe}, where there are three treatments. The set~$\mathscr{M}_K$ is colored red, whereas its complement in the simplex (i.e., those weights that are ruled out) is shown in gray. In Part~(A) of that figure, we show a situation corresponding to an incompatibility constraint, where Treatments~$1$ and~$3$ are elements of the same group and Treatment~$2$ constitutes the second group. In Part~(B) of that figure we show an ``entry cost'' constraint on Treatment~$2$. Note that in~(A) Treatments~1 and~3, or Treatment~2 can be strongly inferior treatments, whereas in~(B), only Treatment~2 can be strongly inferior.
\begin{figure}
	\centering
	\includegraphics[width=0.7\linewidth]{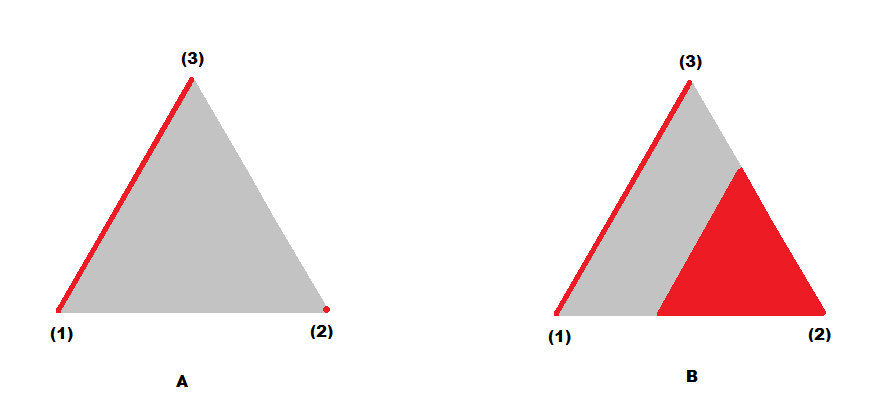}
	\caption{Incompatibility constraint (A) vs.~an entry cost constraint (B).}
	\label{fig:exe}
\end{figure}
\end{example}

The policy we propose proceeds in rounds (analogous to the sequential policies already discussed). In every round, treatments that have not been eliminated in previous rounds are assigned once. The decision maker needs to a priori specify a subset of ``elimination rounds''~$\mathcal{R}$, in which treatments can potentially be eliminated. Fix a discretization~$\mathscr{M}_{K}^n$ of~$\mathscr{M}_K$, which is updated after each round, starting with~$\mathscr{M}_{K, 0}^n = \mathscr{M}_{K}^n$. In each round~$r$ our policy determines~$\mathscr{M}_{K, r}^n$ by dropping from that set all~$\gamma \in \mathscr{M}_{K, r-1}^n$ for which
\begin{equation}\label{eqn:elim21} 
\max_{\delta \in \mathscr{M}_{K, r-1}^n} \mathsf{T}(\langle \delta, \mathbf{F} \rangle) > \mathsf{T}(\langle \gamma, \mathbf{F} \rangle) + 2u_{\eta}(r, n),
\end{equation}
cf.~Equation~\eqref{eqn:ueta} where~$u_{\eta}(r, n)$ was introduced. If~$\mathscr{M}_{K, r}^n$ does not contain a single weights vector with positive~$i$-th coordinate, treatment~$i$ is eliminated, i.e., is no longer assigned in subsequent rounds. Note that compared to the previously introduced sequential policies (that were tailored to incompatibility constraints), the elimination threshold~$2u_{\eta}(r, n)$ is now \emph{more conservative}, because we can no longer exploit any particular structure of~$\mathscr{M}_K$. This is the price we have to pay for obtaining a policy that can be applied under weaker structural assumptions on~$\mathscr{M}_K$. 

A detailed description of~$\check{\pi}$ is given in Policy \ref{pol:FSE_seq}. The recommendation is based on an empirical-success rule and does not depend on external randomization, which is therefore dropped from the notation. Note also that one can only eliminate treatments, if the user chooses the discretization~$\mathscr{M}_K^n \subseteq \mathscr{M}_K$ in such a way that it contains ``sparse'' elements (i.e., vectors with coordinates equaling~$0$), because otherwise the policy will never eliminate treatments. 
\begin{policy}[t!]\label{pol:FSE_seq}
\caption{Sequential Elimination~$\check{\pi}$}

\smallskip

\textbf{Input:}~$n \in \N$,~$K \leq n$, discretization~$\mathscr{M}_K^n$, set of elimination rounds~$\mathcal{R}$,~$\eta>0$

\smallskip

\textbf{Set:}~$t\leftarrow 0$,~$r\leftarrow0$,~$\mathcal{I}_0\leftarrow\mathcal{I}$,
$\mathscr{M}_{K, 0}^{n} \leftarrow \mathscr{M}_{K}^n$.

\smallskip

\While{$n-t\geq |\mathcal{I}_r|$}{

\For{$i\in\mathcal{I}_r$}{
$t\leftarrow t+1$ \\
$\check{\pi}_{n,t}(Z_{t-1})=i$ 
}
$r\leftarrow r+1$\\
$\mathcal{I}_r\leftarrow \mathcal{I}_{r-1}$ \\

\If{$r\in\mathcal{R}$ and~$|\mathcal{I}_r| \geq 2$}{
\smallskip
$h \leftarrow \max\{\mathsf{T}(\langle \delta, \hat{\mathbf{F}}_{t,n} \rangle): \delta \in \mathscr{M}_{K, r-1}^n\}$ \\
$\mathscr{M}_{K, r}^n \leftarrow \{\gamma \in \mathscr{M}_{K, r-1}^n: h - \mathsf{T}(\langle \gamma, \hat{\mathbf{F}}_{t,n} \rangle) \leq 2 u_{\eta}(r, n)\}$ \\
$\mathcal{I}_r \leftarrow \{i \in \mathcal{I}_r: \text{there exists a } \gamma \in \mathscr{M}_{K, r}^n \text{ with } \gamma_i > 0\}$
}

}

\If{$t<n$}{
Let~$(j_{1}, \hdots, j_{|\mathcal{I}_r|})$ be the elements of~$\mathcal{I}_r$, ordered from smallest to largest.

\For{$i=1,\hdots,(n-t)$}{$\check{\pi}_{n, t+i}(Z_{t+i-1})= j_i$}
}
$\check{\pi}_{n,n+1}(Z_n)=\min\argmax \cbr[1]{\mathsf{T}(\langle\delta,\hat{\mathbf{F}}_{n,n}\rangle):\delta\in \mathscr{M}_{K, r}^n}$
\end{policy}

\begin{theorem}\label{thm:FSAcheck}Suppose Assumptions \ref{as:dgp}, \ref{as:MAIN}, and \ref{as:MB} hold, and that~$\mathscr{M}_K$ satisfies Assumption~\ref{as:M}. Then, the SE policy~$\check{\pi}$ as in Policy~\ref{pol:FSE_seq} satisfies the same regret upper bound as the policy~$\tilde{\pi}$ in Theorem~\ref{thm:FSApartition}.
\end{theorem}
One advantage of the policy~$\check{\pi}$ over~$\tilde{\pi}$ is that it can be run and has optimality guarantees regardless of whether~$\mathscr{M}_K$ is incompatibility constrained or not, because~$\check{\pi}$ does not exploit any structural properties of~$\mathscr{M}_K$. This is achieved through a more conservative constant~$2u_{\eta}(r, n)$.

\section{Numerical results}\label{sec:num}
We now illustrate the theoretical results established in this article by means of numerical examples. Simulations will be conducted for~$\mathsf{T}$ the Gini-welfare measure introduced in Example~\ref{ex:gini} with~$\mathscr{M}_K = \mathscr{E}_K$, for the mean functional with~$\mathscr{M}_K$ incorporating different types of capacity and similarity constraints, and for the headcount ratio discussed in Example~\ref{eqn:pov}, a widely used poverty measure.
In all cases we study the expected regret of the static assignment policy~$\hat{\pi}$ and, for those types of~$\mathscr{M}_K$ as in Example~\ref{ex:compa}, we also study the performance of the sequential elimination policy~$\tilde{\pi}$. For simplicity in the simulations, the SA policy is implemented assigning the~$K$ treatments available cyclically in the course of the experimentation phase (i.e., with the partition $\Pi_{n,i}=\cup_{s=0}^\infty \cbr[0]{i+Ks}\cap\cbr[0]{1,\hdots,n}$), which leads to balanced groups. Note that considering cyclical assignments is without loss of generality here, as the performance of the SA policy only depends on the partition via the number of subjects assigned to every treatment. The SE policy is implemented with
$\eta=0.0001$, and with elimination rounds $\mathcal{R}=\cbr[0]{1,\hdots,\lfloor n/2\rfloor}$, i.e., allowing for elimination after each round.\footnote{Experiments (not reported) yielded that lower values of~$\eta>0$, i.e.,~more aggressive elimination of treatments, resulted in lower regret for the SE policy which we therefore recommend for practice.}

We study the expected regret of the policies for all~$n$ that are integer multiples of~$100$ between~$100$ and~$5{,}000$ and approximate the expected regret by the average regret (i.e., the arithmetic mean) over~$5{,}000$ replications for each sample size.

A brief summary of the findings is as follows. The SA and SE policy both generally incur low average regret, confirming our theoretical optimality results for these policies. For sufficiently small sample sizes, the recommendations of the two policies are identical, as no treatments are eliminated by the SE policy in this case. However, for larger samples sizes, the SE policy eliminates clearly suboptimal treatments and thus incurs a lower average regret. In this sense, our numerical results indicate that in terms of regret there is nothing lost from using the SE policy for such~$\mathscr{M}_K$, as it always performs at least as well as the SA policy and sometimes strictly better. The SE policy is, however, computationally more burdensome, as it must compare the values of the functional after every elimination round. Note, however, that the number of elimination rounds can be chosen by the user to lighten the numerical cost of using the SE policy.

\subsection{Gini-welfare}\label{sec:Ginisim}
Recall the Gini-welfare measure from Example~\ref{ex:compa}. There, it is shown that this functional satisfies Assumption~\ref{as:MAIN} with $\mathscr{D}=D_{cdf}([0,1])$ and $C=2$, and is quasi-convex. We know from the discussion in Section~\ref{sec:qconv} that any~$\mathscr{M}_K$ can thus without loss of generality be replaced by the set of its extreme points. In situations where~$\mathscr{M}_K$ contains~$\mathscr{E}_K$, it therefore suffices to work with~$\mathscr{E}_K$, the case we focus on here (and which allows an application of the SE policy). We consider~$K=5$ treatments, and where~$Y_{i,t}$ for~$i = 1, \hdots, K$ are (independently) distributed with~$F^i$ the cdf of a Beta-distribution with parameters~$\alpha_i$ and~$1$. Three settings of parameters are considered.
\begin{enumerate}
	\item \emph{Approximately equal Gini-welfare}: $\alpha_1=5$ and $\alpha_i=4.5$ for $i=2,\hdots,5$. These distributions have a Gini welfare of $0.76$ and $0.74$, respectively.  
	\item \emph{Approximately equal Gini-welfare within three classes}: i) $\alpha_1=5$ and $\alpha_2=4.5$, ii) $\alpha_3=\alpha_4=2$ and iii) $\alpha_5=1$. These distributions have a Gini-welfare of  $0.76,\ 0.74,\ 0.53$ and $0.33$, respectively.
	\item \emph{Two strong treatments and three inferior ones}: $\alpha_1=5,\ \alpha_2=4.5$ and $\alpha_i=1$ for $i=3,\hdots,5$. These distributions have a Gini-welfare of $0.76,\ 0.74$ and $0.33$, respectively. 
\end{enumerate}
Before we proceed to the results, we note that since~$\mathscr{E}_K$ is finite, we used the discretization~$\mathscr{M}_K^n = \mathscr{E}_K$ in both policies.
\subsubsection{Results}
The numerical results for the Gini-welfare measure are summarized in Figure \ref{Fig:1}. The top two panels in Figure \ref{Fig:1} contain the results for Setting 1 where all treatments are approximately equally good. The Gini welfares of the five treatments are so close to each other that no treatment is ever eliminated by the SE policy for any of the sample sizes considered. Thus, the SA and SE policies are identical and incur the same average regret as witnessed by the left and right panel alike. From the left panel it is seen that even for the smallest sample sizes considered, the SA and SE policies both incur rather small average regret.

The middle two panels contain the results for Setting 2, where there are three classes of treatments. The most important difference to Setting 1 is that the regrets of the SA and SE policies now differ. In fact, by eliminating the very inferior Treatment 5 after (on average) around 250 rounds (and the moderately inferior Treatments 3 and 4 after around 850 rounds), the SE policy is able to allocate more assignments to Treatments 1 and 2, which are the ones that are hard to distinguish. As a result of this, the SE policy recommends the superior Treatment 1 more often than the SA policy and hence generally incurs a lower average regret. This gain is almost uniform in $n$ except for at $n=1{,}100$ where the relative regret is~1.007. This slightly higher average regret is not due to the SE policy eliminating the best treatment --- it never did in our simulations. Rather it is a consequence of approximating expected regret by average regret. This also explains why the average regret of the SE policy is not monotonically decreasing in~$n$ in any of the three settings. The right panel shows that the relative gains of the SE over the SA policy can be more than 50\% in Setting 2.

The bottom panel contains the results for Setting 3 where there are two strong treatments and three inferior ones. Both policies have a small average regret. However, the SE policy is able to eliminate the inferior treatments after (on average) about 250 rounds. This results in the SE policy incurring an average regret that is lower by several orders of magnitude than the one of the SE policy for the largest sample sizes considered.

Summarizing, the numerical results for Gini-welfare and~$\mathscr{M}_K = \mathscr{E}_K$ suggest that the SE policy is never inferior to the SA policy in terms of average regret, and can be superior whenever there are treatments that are not too similar to the best (in relation to sample size). The computational burden of the SE policy is furthermore negligible in case~$\mathscr{M}_K = \mathscr{E}_K$. Thus, if a sequential assignment scheme is practically feasible, the SE policy should be used for the Gini-welfare and~$\mathscr{M}_K = \mathscr{E}_K$.

\begin{figure}
	\begin{center}
		Setting 1
		
		\vspace{-0.5cm}
		
		\includegraphics[height=6.5cm, width=7cm]{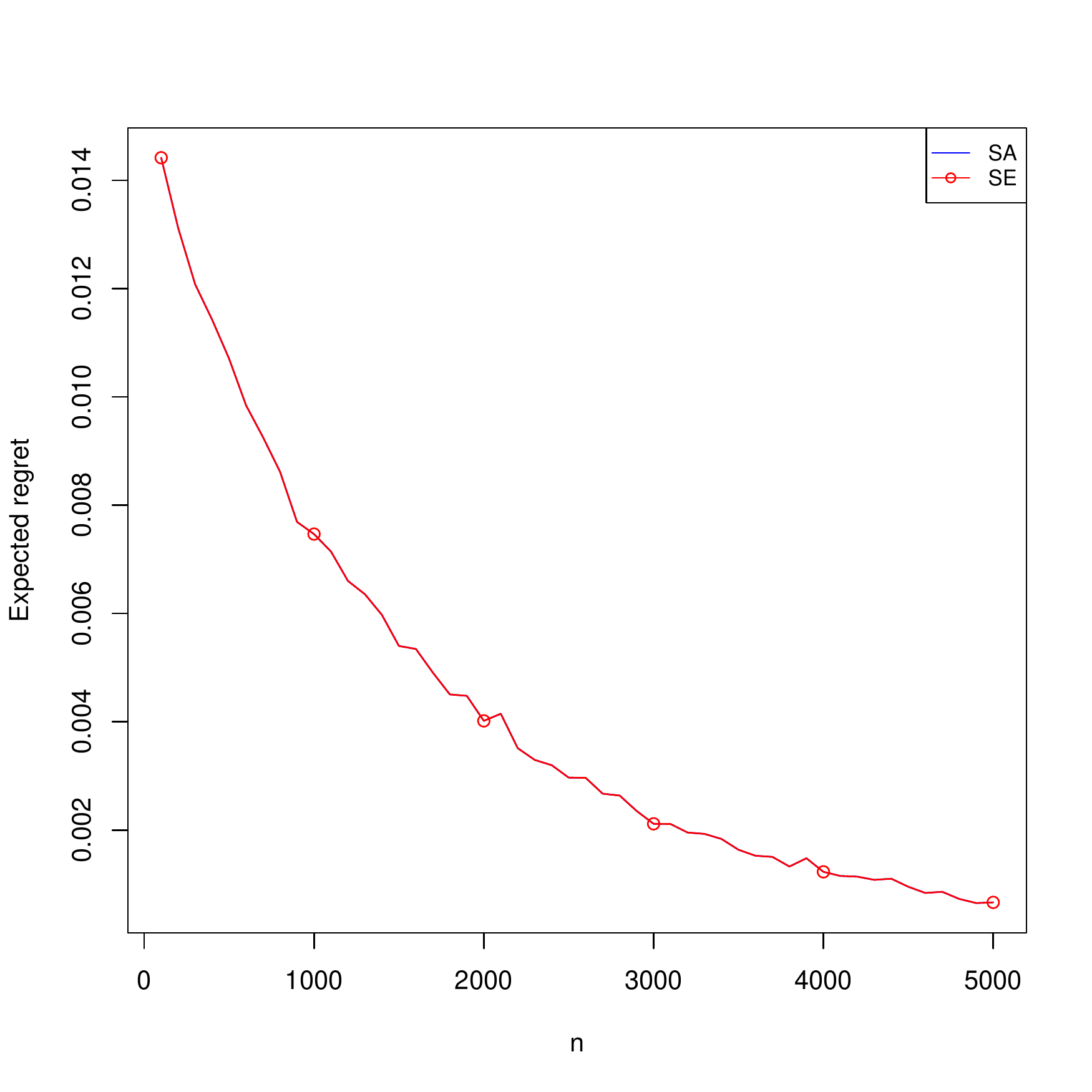}
		\includegraphics[height=6.5cm, width=7cm]{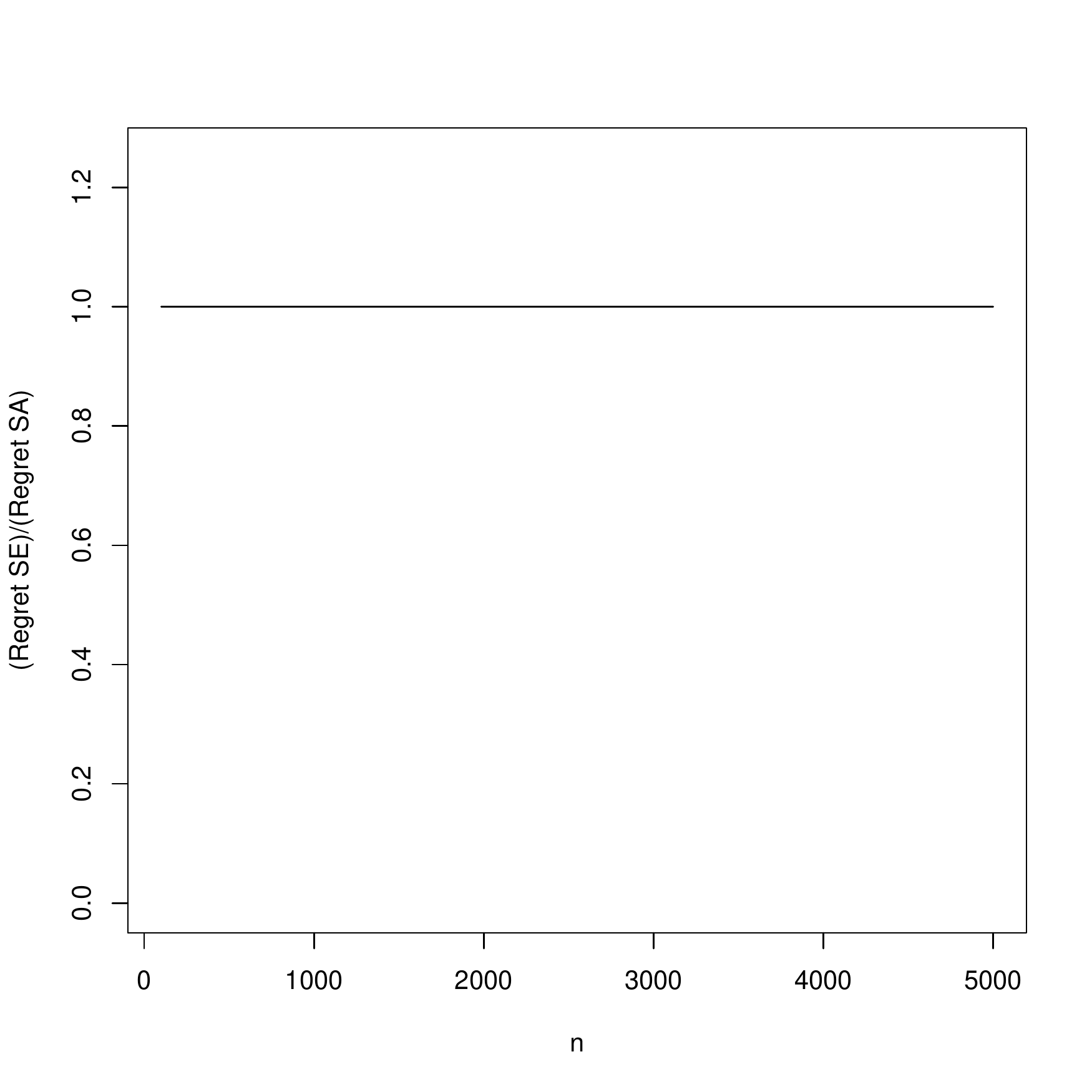}
		
		Setting 2
		
		\vspace{-0.5cm}
		
		\includegraphics[height=6.5cm, width=7cm]{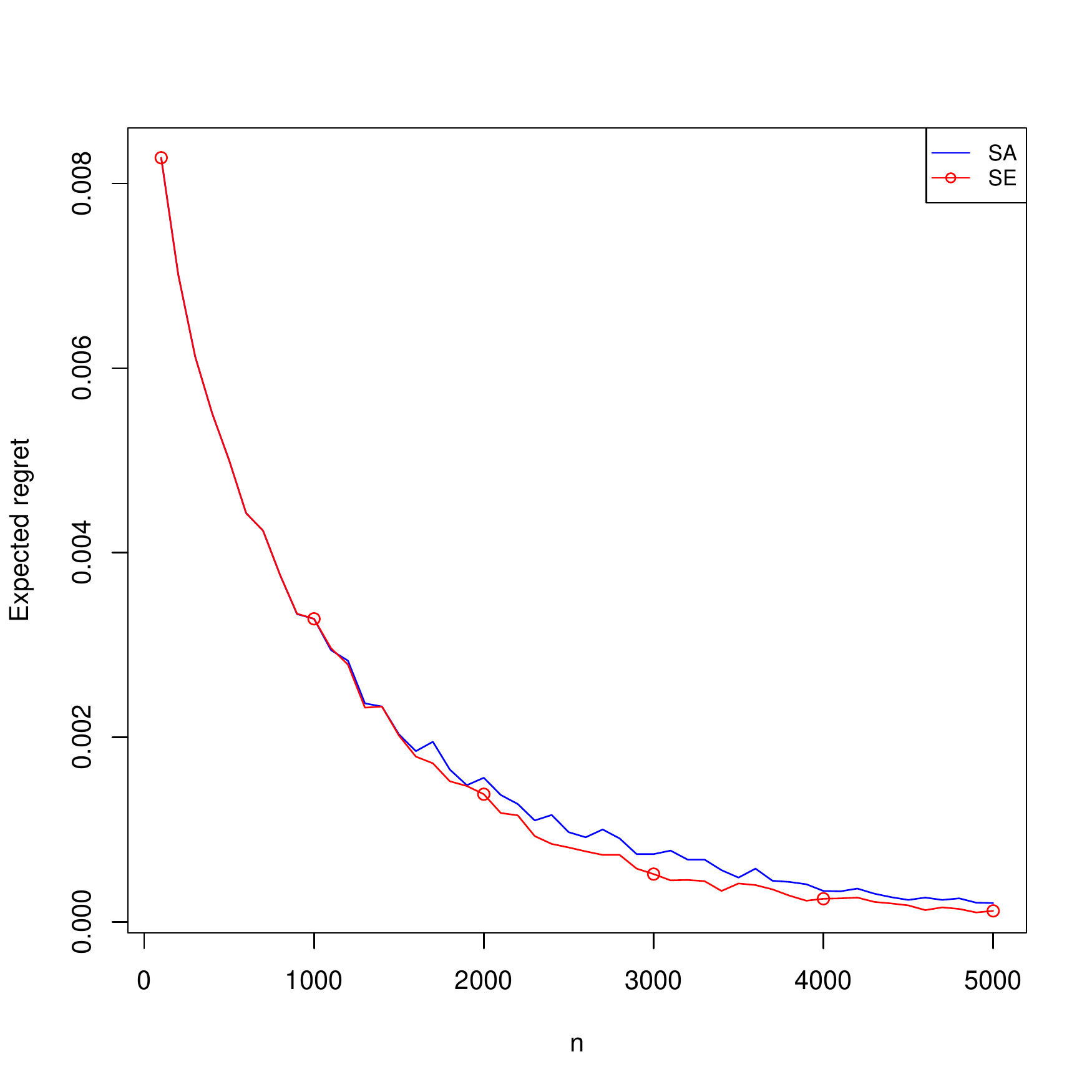}
		\includegraphics[height=6.5cm, width=7cm]{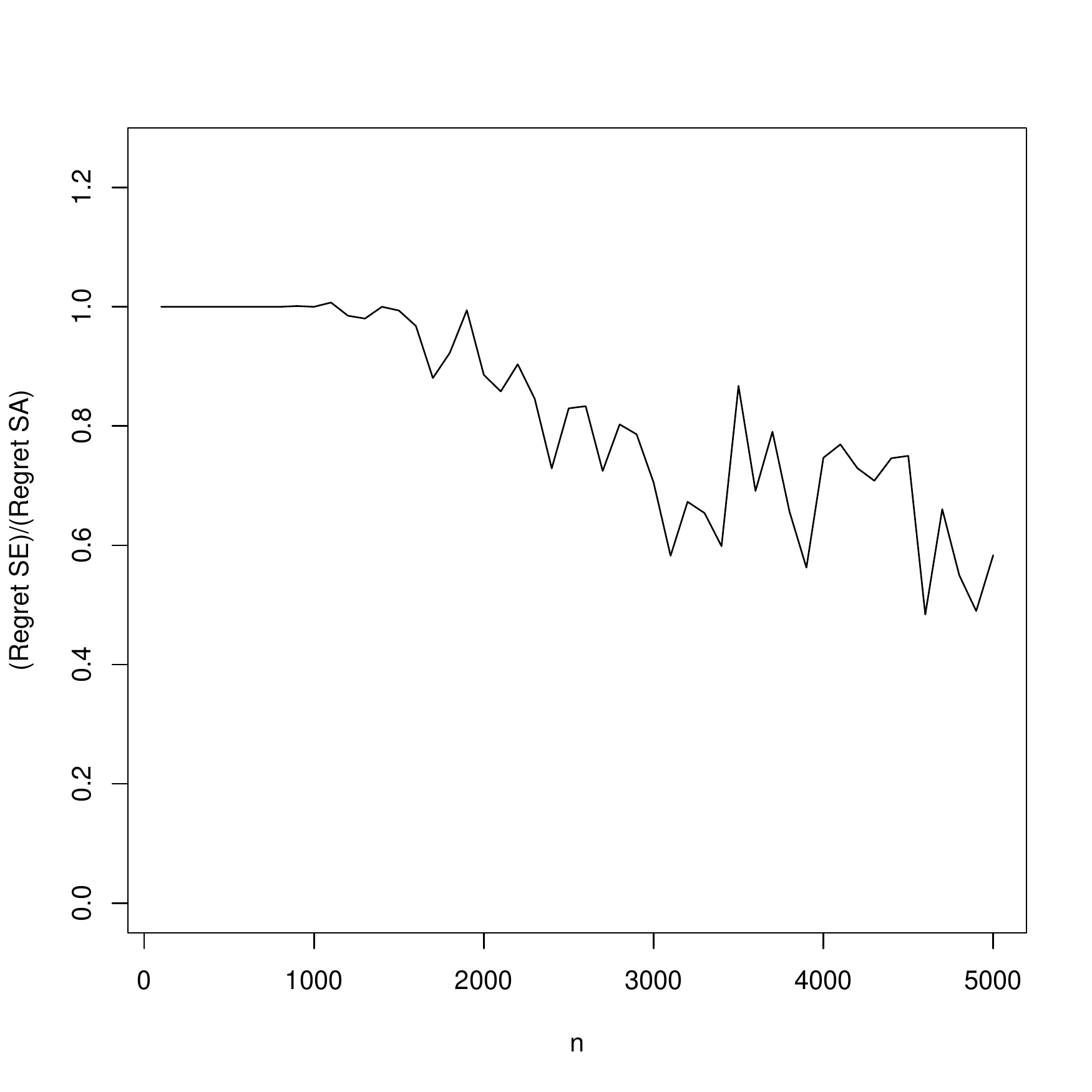}
		
		Setting 3
		
		\vspace{-0.5cm}
		
		\includegraphics[height=6.5cm, width=7cm]{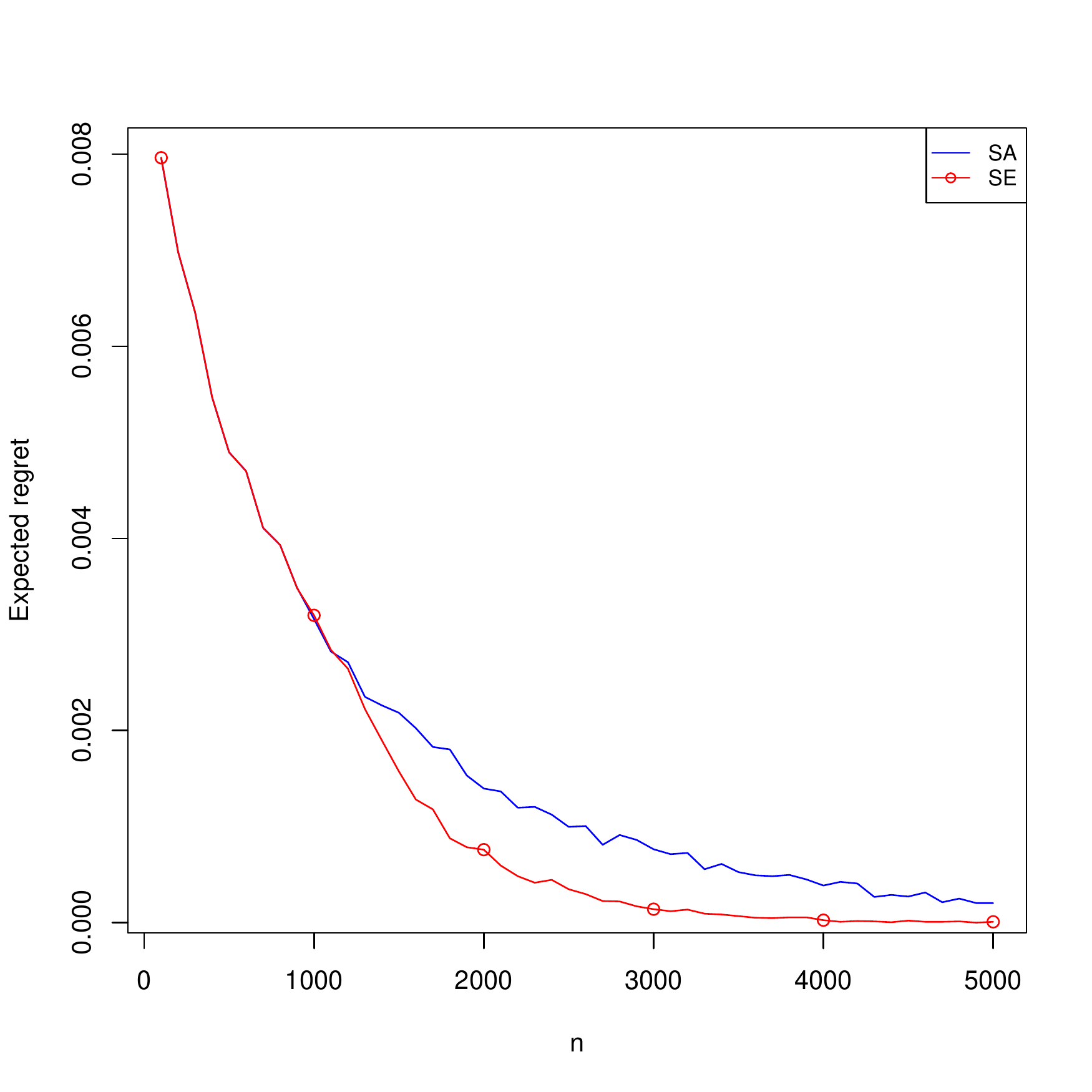}
		\includegraphics[height=6.5cm, width=7cm]{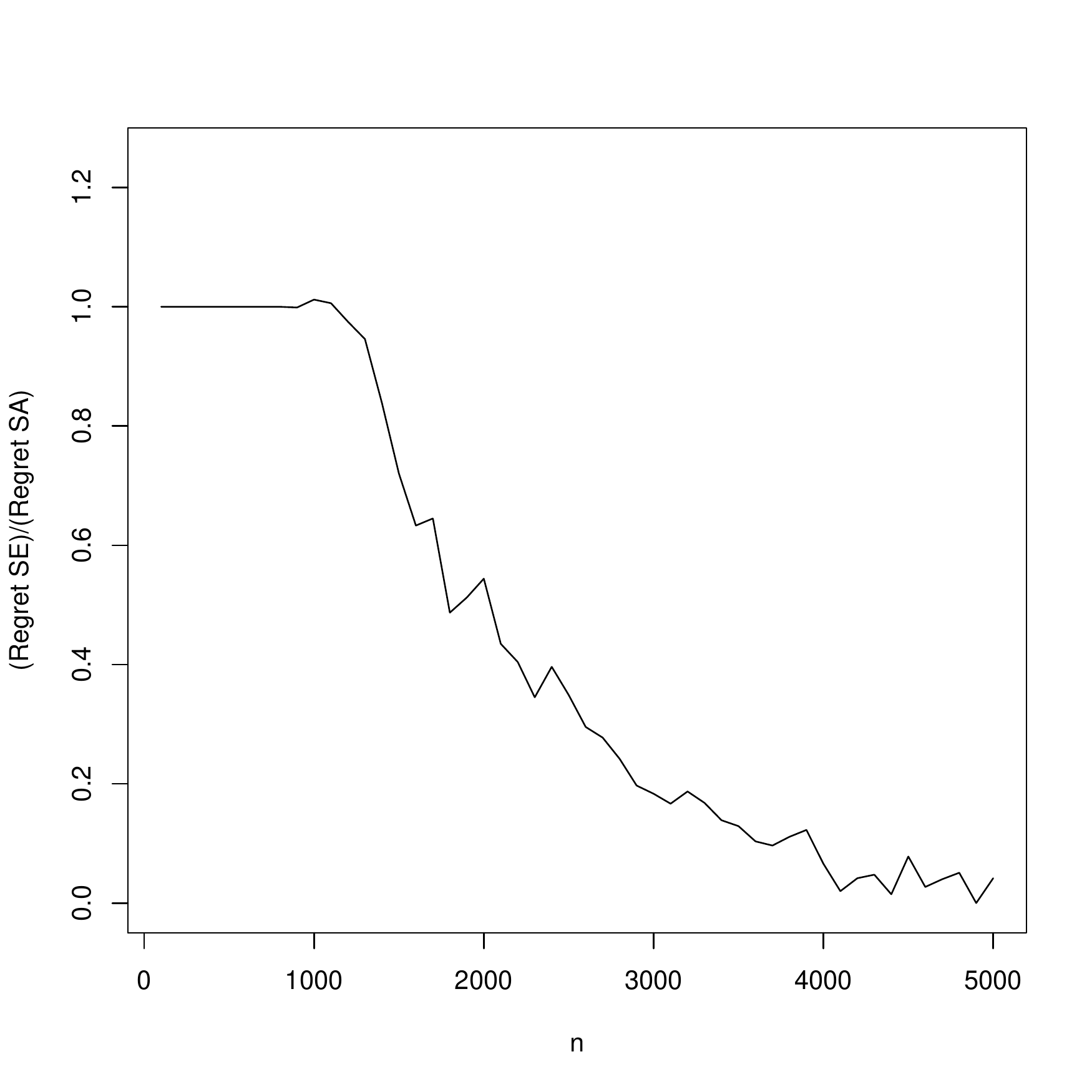}
		\vspace{-.8cm}
	\end{center}
	\caption{{\footnotesize Absolute and relative average regrets for Settings~1-3 for the Gini-welfare measure. Each row corresponds to a setting. The left column contains the absolute average regrets, while the right column contains the relative average regret of the SE to the SA policy. Thus, numbers less than one for the relative average regret mean that the SE has a lower average regret than the SA policy.}}
	\label{Fig:1}
\end{figure}

\subsection{Mean functional with capacity and similarity constraints}

When a decision maker targets the treatment with the highest expectation, and is free to roll out a single best treatment, the situation is very similar (apart from the functional used) to the situation in the previous section. However, if the decision maker needs to take capacity or similarity constraints into consideration, this is no longer true, as such constraints typically imply restrictions of the type~$\mathscr{E}_K \not \subseteq \mathscr{M}_K$. In this section, we study the performance of the SA policy under such restrictions.\footnote{In none of the situations we consider here, there are strongly inferior treatments. Hence, we do not implement any version of the SE policy.} Again, in all of the settings we consider, there are~$K=5$ treatments, and~$F^i$ is the cdf of a Beta-distribution with parameters~$\alpha_i$ and~$1$. We consider the parameters~$\alpha_1=1.2,\ \alpha_2=1.1,\ \alpha_3=1,\ \alpha_4=0.9$ and $\alpha_5=0.8$, the corresponding Beta-distributions having means $0.55,0.52,0.5,0.47$ and $0.44$, respectively. The restrictions~$\mathscr{M}^{(1)}_K \supseteq \mathscr{M}^{(2)}_K \supseteq \mathscr{M}^{(3)}_K$ we study are as follows.
\begin{enumerate}
	\item \emph{The capacities of all treatments are equally restricted}: $$
	\mathscr{M}_{K}^{(1)}=
	\{
	\delta \in \mathscr{S}_K : 0.1 \leq \delta_i \leq 0.5, \text{ for } i = 1, \hdots, K
	\}.$$ Here, the optimal weights vector is~$\delta = (0.5, 0.2, 0.1, 0.1, 0.1)'$. As expected, this puts as much weight as possible on the treatments with the highest means.
	\item \emph{The joint capacity of the first two treatments is further restricted}: $$\mathscr{M}_{K}^{(2)}= \cbr[0]{\delta\in\mathscr{M}_{K}^{(1)}:2\delta_1+\delta_2\leq 0.5}.$$ Now the total weight on the best two treatments can not be too high. In particular, Treatment 1 is ``expensive." The optimal weights vector is~$\delta = (0.1, 0.3, 0.4, 0.1, 0.1)'$. Observe that now Treatment 1 is assigned as little as possible, despite having the highest mean.
	\item \emph{A similarity constraint concerning Treatments 3 and 5 is added:}  $$\mathscr{M}_{K}^{(3)}=\cbr[0]{\delta\in \mathscr{M}_{K}^{(2)}:|\delta_3-\delta_5|\leq 0.1}.$$ The optimal weights vector here is~$\delta = (0.1, 0.3, 0.2, 0.3, 0.1)'$. Now the weights of Treatment 3, which were highest in the previous setting, must be close to those of the individual treatment with the lowest mean.
\end{enumerate}
Before we proceed to the numerical results, we note that since the constraints are all linear, the recommendation in the SA policy can be found by solving (for~$i = 1, 2, 3$) the linear program
\begin{equation*}
	\max_{\delta \in \mathscr{M}^{(i)}_K} \mathsf{T}(\langle \delta, \mathbf{\hat{F}}_{n, \Pi_n}\rangle)  = 
	\max_{\delta \in \mathscr{M}^{(i)}_K}
	\sum_{j = 1}^K \delta_j \hat{\mu}_j
\end{equation*}
via the simplex-algorithm; here~$\hat{\mu}_j$ denotes the mean of the empirical cdf based on the observations assigned to treatment~$j$.

\subsubsection{Results}
Figure \ref{Fig:2} contains the results for Settings 1-3 just described. The regret decreases with $n$ in all three settings in accordance with the theoretical guarantees described in Theorem~\ref{thm:NMAPup}. Furthermore, even for small $n$, the regret is ``low'' for all three sets $\mathscr{M}^{(i)}_K$ considered. The figure also illustrates that the smaller the feasible set is, the lower the regret is for small sample sizes. This is in accordance with the lower bounds in Theorem~\ref{Thm: UniformLowermixed}, which involve smaller constants for smaller feasible sets.
\begin{figure}[H]
	\begin{center}
		\includegraphics[height=6.5cm, width=7cm]{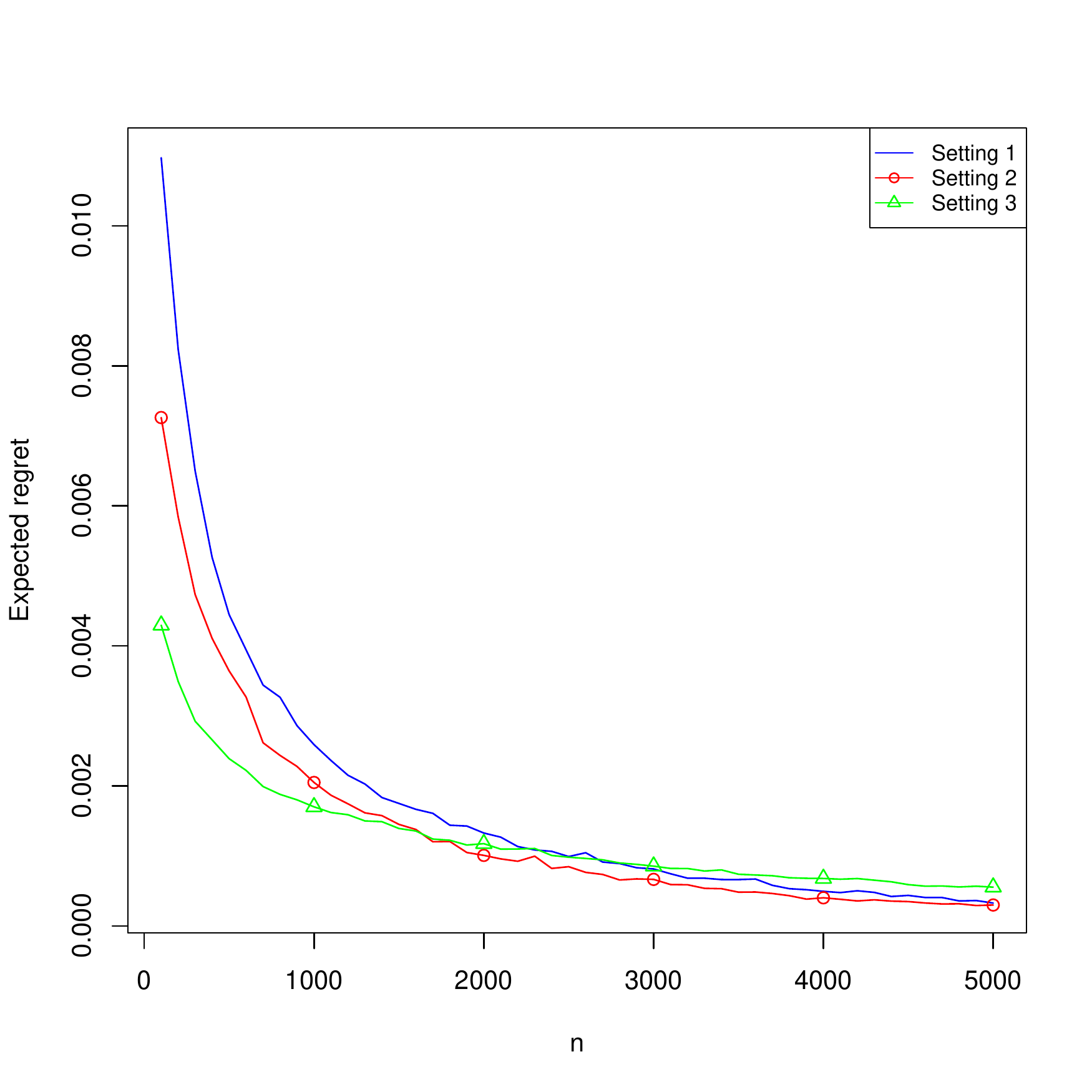}
	\end{center}
	\vspace{-.8cm}
	\caption{{\footnotesize Average regret of the SA policy for the mean functional under capacity and similarity constraints in Settings 1-3 described in the text.}}
	\label{Fig:2}
\end{figure}

\subsection{Headcount ratio}
We finally study a situation in which the decision maker targets a common poverty measure: namely the (negative) headcount ratio as defined in Equation~\eqref{eqn:headcount} of Example~\ref{eqn:pov}. In contrast to the Gini welfare measure and the mean, this functional is not quasi-convex. Thus, even when~$\mathscr{E}_K\subseteq\mathscr{M}_K$, the mixture that minimizes the headcount ratio need not be an element of~$\mathscr{E}_K$. One example of this type is the following simple setting:
\begin{align}
	F^1(x)&= 1.25x \mathds{1}_{\{ x<0.25 \} }+ 0.3125 \mathds{1}_{\{ 0.25 \leq x<0.3125 \} }+ x \mathds{1}_{\{ x \geq 0.3125 \} },\notag\\
	F^2(x)&= 3x \mathds{1}_{\{ x<0.1 \}}+0.3\mathds{1}_{\{ 0.1 \leq x<0.15 \}}+  2 x \mathds{1}_{\{ 0.15 \leq x < 1/2 \}}+\mathds{1}_{\{ x \geq  1/2 \}},\label{eq:hcsimdist}\\
	F^i(x)&=a_i \mathds{1}_{\{ x<1/2 \} }+ (p_i x+1-p_i) \mathds{1}_{\{ x \geq 1/2 \} },\qquad \text{for }i\in\cbr[0]{3,4,5},\notag
\end{align}
with~$p_i=2(1-a_i)$ and~$a_3=0.5,\ a_4=0.7,\ a_5=0.9$, respectively. It can be verified that the headcount ratios of these five cdfs are~$0.31,\ 0.3,\ 0.5,\ 0.7$ and~$0.9$, respectively (mentioned in the order of~$F^1$ to~$F^5$). Intuitively, the larger~$a_i$ is, the larger is the mass at zero of~$F^i$, i.e, the larger is the proportion of very poor people. We let
\begin{align*}
	\mathscr{M}_K=\cbr[0]{\delta\in\mathscr{S}_5: \delta_3=\delta_4=\delta_5=0}\cup \cbr[0]{e_3(5)}\cup \cbr[0]{e_4(5)}\cup \cbr[0]{e_5(5)},
\end{align*}
such that mixtures of~$F^1$ and~$F^2$ are allowed, while~$F^3$,~$F^4$, and~$F^5$ can only be assigned on their own. This is an instance of an incompatibility constraint as described in Example~\ref{ex:compa}. It can be shown that~$\delta=\del[1]{\frac{736}{3171},\frac{2435}{3171},0,0,0}'$ minimizes the headcount ratio over~$\mathscr{M}_K$ with a corresponding minimal value of~$0.2739$. The values of the headcount ratios corresponding to mixtures of~$F^1$ and~$F^2$ are shown in Figure~\ref{fig:hc} (which nicely illustrates that the lowest headcount ratio is obtained by combining two treatments). Importantly, we note that the minimizer~$\delta\not\in\mathscr{E}_K$ even though~$\mathscr{E}_K\subseteq \mathscr{M}_K$. The discretization used is 
\begin{align*}
	\mathscr{M}_K^n=\cbr[0]{\delta\in\mathscr{S}_5: \delta_3=\delta_4=\delta_5=0,\ 100\delta_1\in\Z}\cup \cbr[0]{e_3(5)}\cup \cbr[0]{e_4(5)}\cup \cbr[0]{e_5(5)},
\end{align*}
and we set~$C=2.5$ in accordance with Example~\ref{eqn:pov}. 

\begin{figure}
\centering
\includegraphics[height=6.5cm, width=7cm]{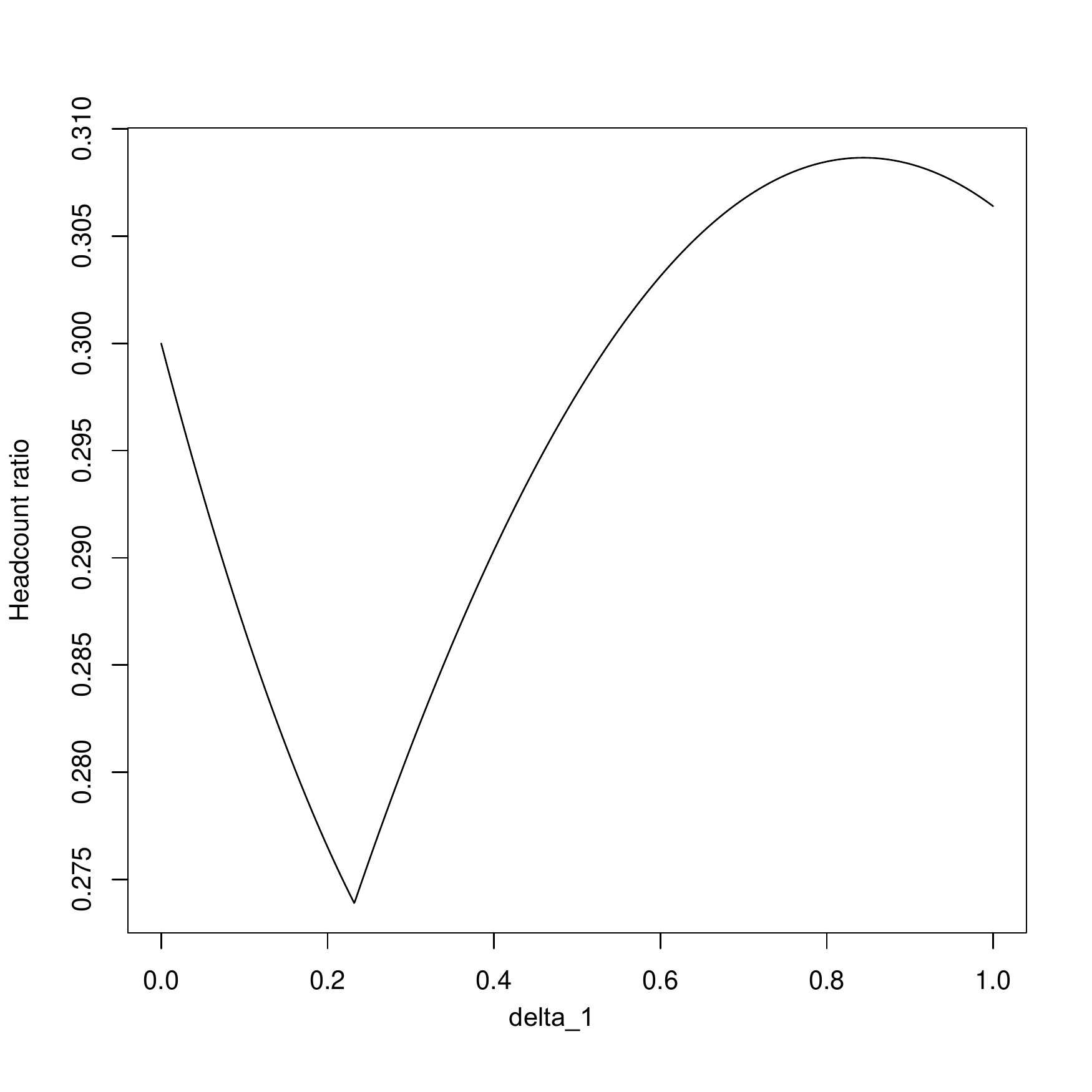}
\caption{\footnotesize{Headcount ratios for mixtures of the form~$\delta_1 F^1 + (1-\delta_1)F^2$ in dependence on~$\delta_1 \in [0, 1]$.}}
\label{fig:hc}
\end{figure}

\subsubsection{Results}
The numerical results for the headcount ratio with outcome distributions as in~\eqref{eq:hcsimdist} are contained in Figure~\ref{Fig:3}. The left panel, which contains the  absolute average regret for the SA and the SE polices, shows that both of these generally incur a low average regret and that this is decreasing in the sample size. Observe also that the SE policy is, up to simulation error, uniformly better (in~$n$) than the SA policy. The right panel contains the relative average regret of the SE policy to that of the SA policy. It can be seen that for~$n\geq 3{,}000$ the average regret of the SE policy is always at least~25\% lower (and up to~38\% lower) than that of the SA policy. As was the case for the Gini welfare measure, this is due to the SE policy correctly eliminating inferior treatments and dedicating more sampling effort to, in this case, Treatments~1 and~2 on which the best mixture puts positive weights.  

\begin{figure}[H]
	\begin{center}
		\includegraphics[height=6.5cm, width=7cm]{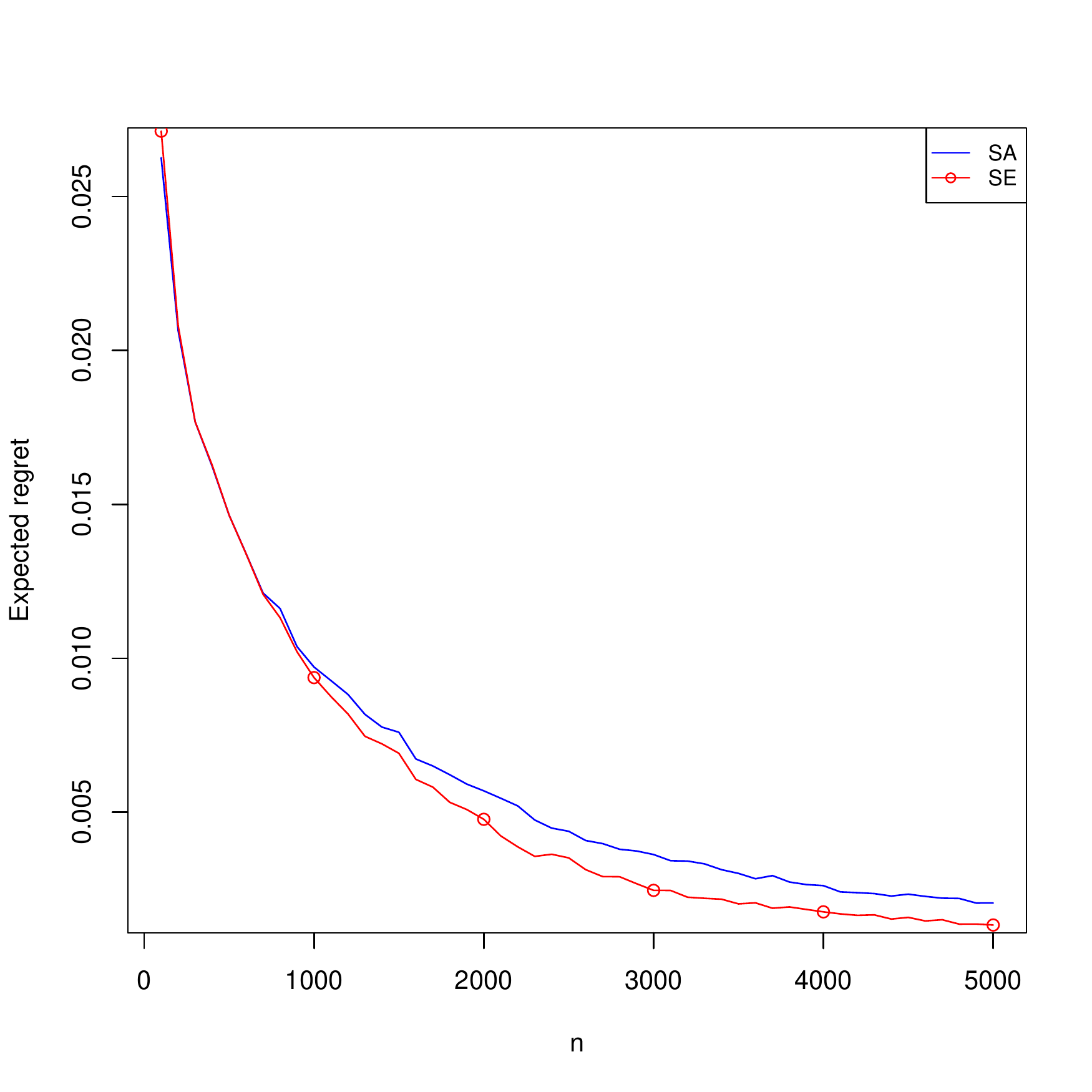}
		\includegraphics[height=6.5cm, width=7cm]{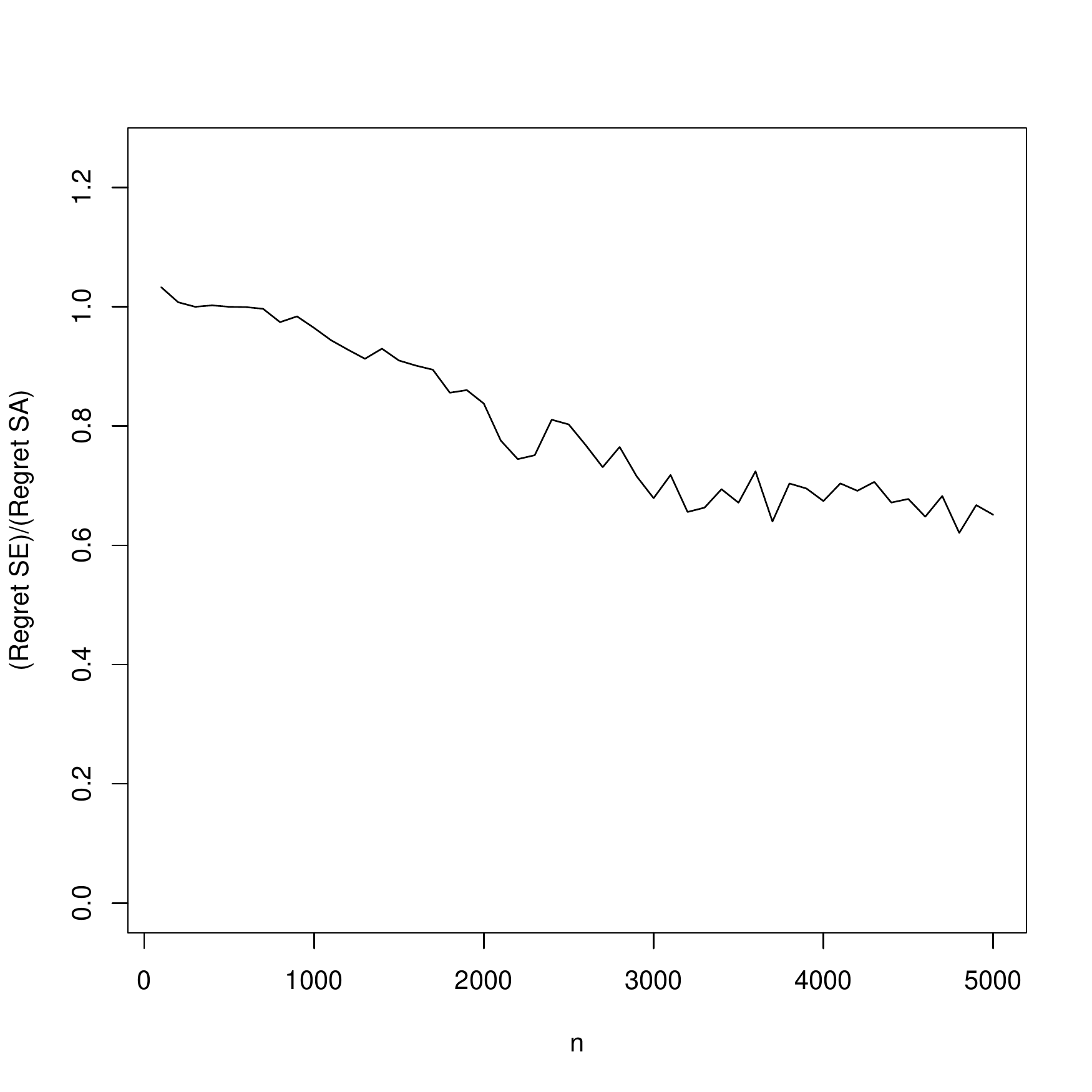}
	\end{center}
	\vspace{-.8cm}
	\caption{{\footnotesize Absolute and relative average regrets for the headcount ratio with outcome distributions as in~\eqref{eq:hcsimdist}. The left plot contains the absolute average regrets, while the right plot contains the relative average regret of the SE to the SA policy. Thus, numbers less than one for the relative average regret mean that the SE has a lower average regret than the SA policy.}}
	\label{Fig:3}
\end{figure}

\subsubsection{Mixing inferior and superior treatments}
In the setting of~\eqref{eq:hcsimdist} \emph{all} mixtures of Treatments 1 and 2 are superior to the remaining Treatments 3,~4 and~5. This may make the elimination of Treatments 3,~4 and~5 too easy.  We therefore consider now a setting where this is no longer the case. The potential outcome distributions are in the same family as those used in~\eqref{eq:hcsimdist}. To be precise, we let
 \begin{align}
 	G^1(x)&= 3x \mathds{1}_{\{ x<0.1 \}}+0.3\mathds{1}_{\{ 0.1 \leq x<0.15 \}}+  2 x \mathds{1}_{\{ 0.15 \leq x < 1/2 \}}+\mathds{1}_{\{ x \geq  1/2 \}},\notag\\
 	G^i(x)&=b_i \mathds{1}_{\{ x<1/2 \} }+ (q_i x+1-q_i) \mathds{1}_{\{ x \geq 1/2 \} },\qquad \text{for }i\in\cbr[0]{2,3,4,5};\label{eq:sim2}
 \end{align}
 with~$q_i=2(1-b_i)$ and~$b_2=0.7,\ b_3=0.32,\ b_4=0.7,\ b_5=0.9$, respectively. It can be verified that the headcount ratios of these five cdfs are~$0.3,\ 0.7,\ 0.32,\ 0.7$ and~$0.9$, respectively (mentioned in the order of~$G^1$ to~$G^5$).
We allow for the superior Treatment 1 to be mixed with the dominated Treatment 2 while the remaining treatments can only be assigned on their own, that is
\begin{align*}
	\mathscr{M}_K=\cbr[0]{\delta\in\mathscr{S}_5: \delta_3=\delta_4=\delta_5=0}\cup \cbr[0]{e_3(5)}\cup \cbr[0]{e_4(5)}\cup \cbr[0]{e_5(5)}.
\end{align*}
It can now be shown that~$\delta=(1,0,0,0,0)$ minimizes the headcount ratio over~$\mathscr{M}_K$ with a corresponding value of~$0.3$ (the headcount ratio corresponding to assigning Treatment 1 on its own). The discretization of~$\mathscr{M}_K$ that we use is
\begin{align*}
	\mathscr{M}_K^n=\cbr[0]{\delta\in\mathscr{S}_5: \delta_3=\delta_4=\delta_5=0,\ 100\delta_1\in\Z}\cup \cbr[0]{e_3(5)}\cup \cbr[0]{e_4(5)}\cup \cbr[0]{e_5(5)},
\end{align*}
and we again set~$C=2.5$ in accordance with Example~\ref{eqn:pov}.

\subsubsection{Additional results}
The results for the headcount ratio with outcome distributions as in~\eqref{eq:sim2} are contained in Figure~\ref{Fig:4}. The findings are qualitatively like those in Figure~\ref{Fig:3} --- as long as no treatments are eliminated the SA and SE policy make the same recommendation but once the SE policy begins eliminating inferior treatments it more often recommends the regret minimizing mixture~$(1,0,0,0,0)$. 

We also experimented with settings in which the headcount ratio of the second best treatment (Treatment 3) is much closer to (further from) that of the best treatment (Treatment 1). In the former case, the recommendation problem can become so difficult that neither the SA not SE policy can distinguish between the two top treatments at the studied sample sizes and the relative regret is essentially one. If, on the other hand, the headcount ratio of the second best treatment is much higher than that of the best treatment then the recommendation problem can become so easy that neither the SA nor the SE policy make any mistakes and again they are equally good. Let us stress that even in such settings nothing is lost from using the SE policy and the above simulations show that there are setting in which one obtains substantially lower regret by using the SE rather than the SA policy.

\begin{figure}[H]
	\begin{center}
		\includegraphics[height=6.5cm, width=7cm]{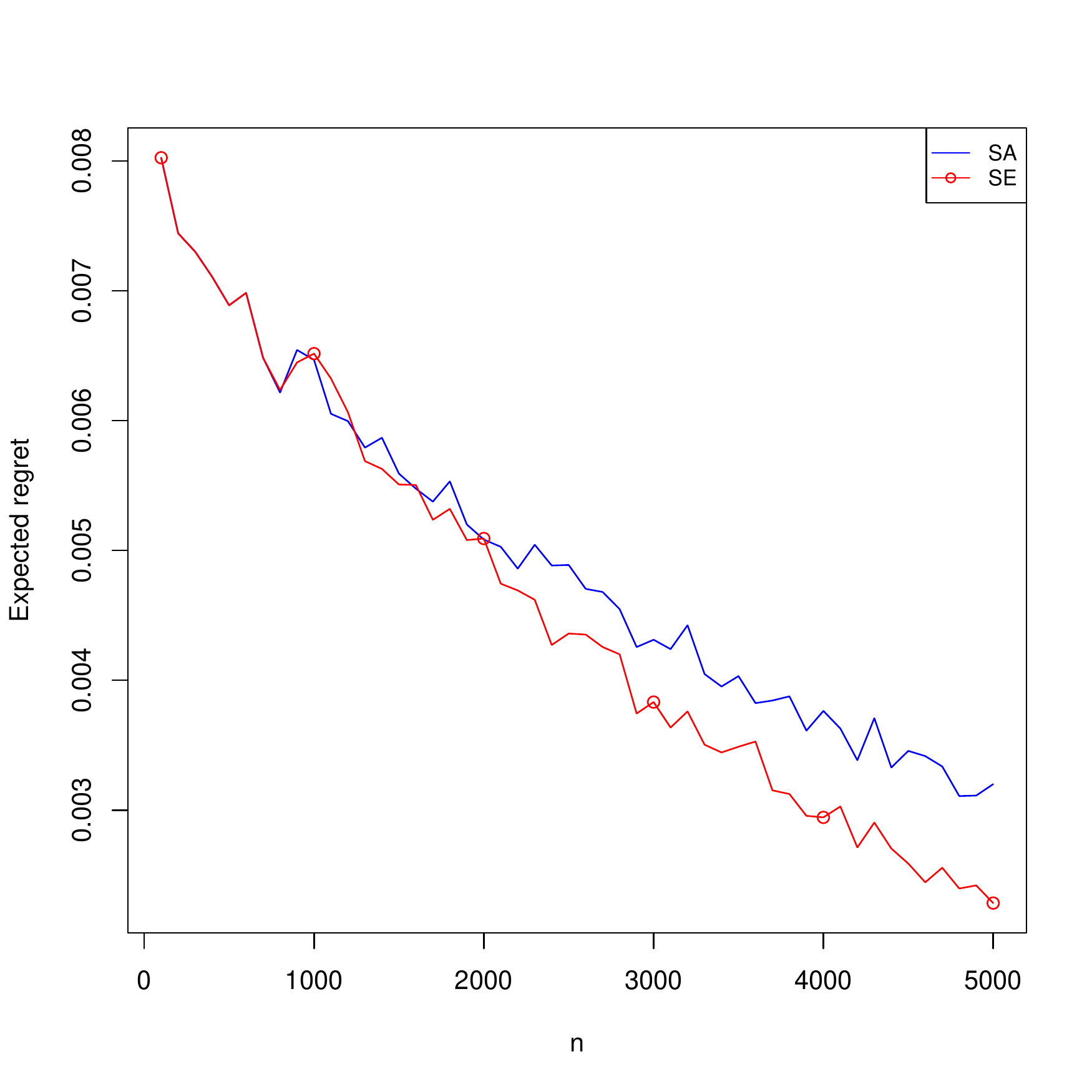}
		\includegraphics[height=6.5cm, width=7cm]{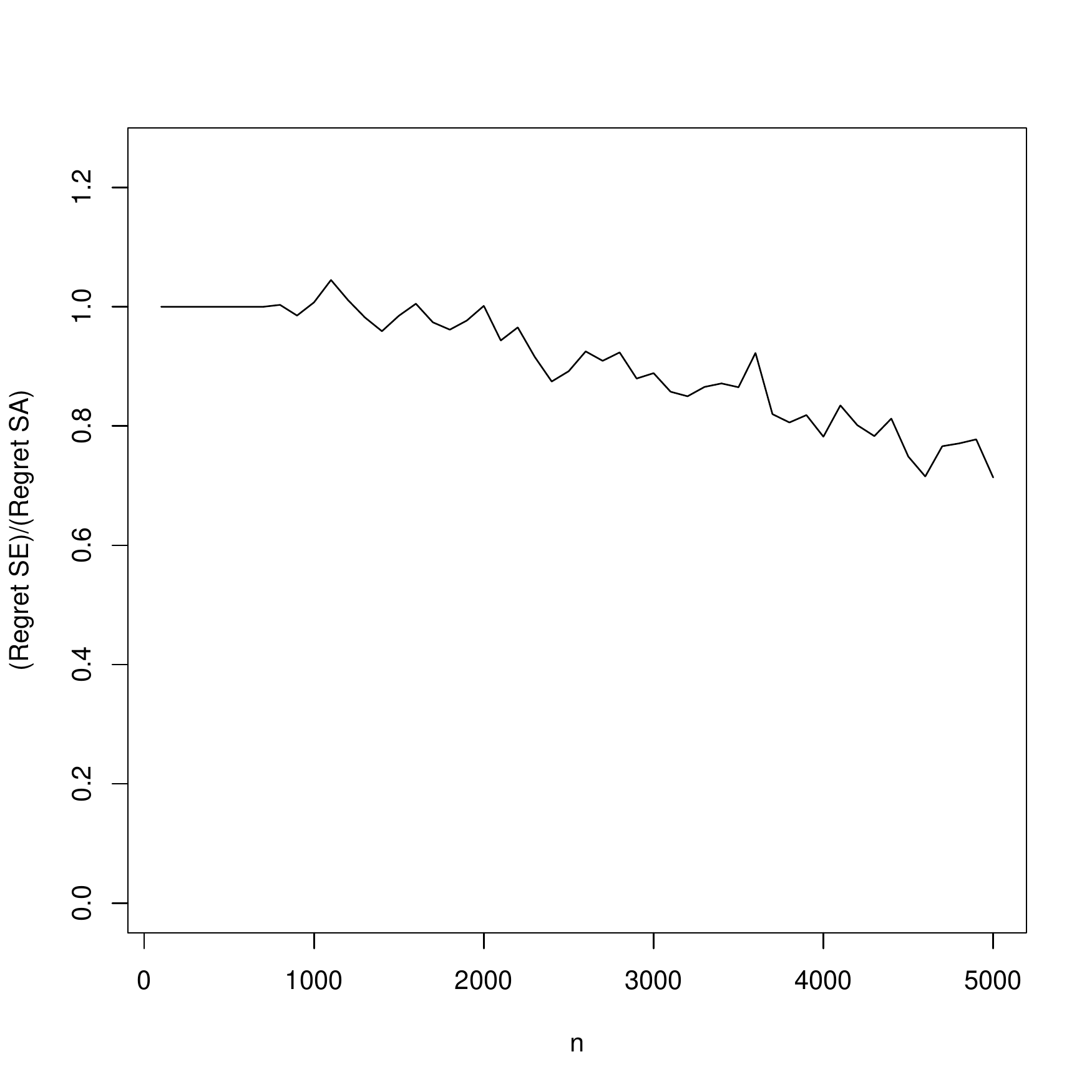}
	\end{center}
	\vspace{-.8cm}
	\caption{{\footnotesize Absolute and relative average regrets for the headcount ratio with outcome distributions as in~\eqref{eq:sim2}. The left plot contains the absolute average regrets, while the right plot contains the relative average regret of the SE to the SA policy. Thus, numbers less than one for the relative average regret mean that the SE has a lower average regret than the SA policy.}}
	\label{Fig:4}
\end{figure}

\section{Incorporating covariates}\label{sec:cov}

One relevant extension of the setup we considered is the situation in which for each subject one also observes a vector of covariates~$X_t \in \mathcal{X}$ with (unknown) distribution~$P_X$, say. While we do not spell out any theoretical results in this setup, we shall heuristically outline in the following how a decision maker can incorporate covariate information using a natural extension of our plug-in based approach. Denoting the conditional distribution of the potential outcome vector~$Y_t$ given~$X_t$ by~$\mathbf{F}(y|x)$ (a distribution on~$\R^K$), any mixture~$\delta: \mathcal{X} \to \mathscr{M}_K$ (in principle~$\mathscr{M}_K$ could then also depend on the covariates but we do not highlight this possibility here) now induces a population cdf via
\begin{equation*}
	\int \langle \delta(x), \mathbf{F}(\cdot | x)\rangle dP_X(x).
\end{equation*}
Note that in a setting with covariates the weights vector~$\delta$ now actually is a vector-valued \emph{function} taking as an argument the covariate vector.
The goal of the policy maker here is to determine an optimal function~$\delta$, i.e., a maximizer of
\begin{equation}\label{eqn:theoX}
	\delta \mapsto  \mathsf{T}\left(
	\int \langle \delta(x), \mathbf{F}(\cdot | x)\rangle dP_X(x)
	\right).
\end{equation}
Following our plug-in approach, given an estimate of the conditional distribution~$\hat{\mathbf{F}}(\cdot|x)$ and an estimate~$\hat{P}_X$ of~$P_X$, the mixture recommended at the end of the experimental phase would be a maximizer of 
\begin{equation*}
	\delta \mapsto  \mathsf{T}\left(
	\int \langle \delta(x), \hat{\mathbf{F}}(\cdot | x)\rangle d\hat{P}_X(x)
	\right).
\end{equation*}
To illustrate this in some detail, consider a specific approach based on partitioning the covariate space, an approach that has been used in related settings, e.g., \cite{rigollet2010nonparametric},~\cite{perchet2013multi} or~\cite{kpv2}. 
Essentially, one would partition the covariate space into~$B$ clusters, say, obtain the estimated conditional cdfs $\hat{\mathbf{F}}_b$ (a cdf on~$\mathbb{R}^K$) for~$j = 1, \hdots, B$ by aggregating outcomes within each cluster, and use these estimates as plug-ins to maximize
\begin{equation}\label{eqn:empX}
	(\delta_1, \hdots, \delta_B) \mapsto  \mathsf{T}\left(
	\sum_{j = 1}^B \hat{p}_j
	\langle \delta^{(j)}, \hat{\mathbf{F}}_j \rangle 
	\right).
\end{equation}
Here~$\hat{p}_j$ denotes the relative frequency of observations falling into cluster~$j$, and each~$\delta^{(j)}$ is an element of~$\mathscr{M}_K$. Optimizing~\eqref{eqn:empX} in~$(\delta^{(1)}, \hdots, \delta^{(B)})$ one obtains the cluster-specific allocation vectors~$(\hat{\delta}^{(1)}, \hdots, \hat{\delta}^{(B)})$. The optimal size of the clusters (ensuring that~\eqref{eqn:empX} well approximates \eqref{eqn:theoX}) depends on the sample size~$n$ and on smoothness assumptions imposed on~$F(y|x)$. In the important case where~$\mathcal{X}$ is finite and not too large no clustering is necessary at all, i.e.,~$B$ would coincide with the number of elements of~$\mathcal{X}$. Static and sequential assignment schemes are possible and would be implemented on the cluster level (but would need to be adjusted so as to take into account the global objective in~\eqref{eqn:empX} in the recommendation and in eliminating treatments on the cluster level). 

\section{Conclusions}
This paper has studied the problem of a policy maker who conducts a sequential experiment to recommend the treatment that is best according to a desirable functional characteristic of the outcome distribution of the treatments. Although our theory allows the policy maker to target a wide class of functionals, the problem is particularly intricate when targeting a functional that is not quasi-convex or when the set of treatments that can be rolled out is restricted as one must then learn the best mixture of treatments. We have characterized the difficulty of this decision problem and shown how it depends on the type of constraints faced. When it is feasible (and sensible, i.e., the constraints do not rule out strongly inferior treatments), we recommend to assign subjects using one of our sequential policies due to their superior regret behavior. 

\begin{appendix}
	\section{Auxiliary results}\label{app:aux}
	\emph{In this section}, let~$\cbr[0]{X_t}_{t\in\N}$ be an i.i.d.~sequence of random variables with cdf $F$; furthermore, for every $t\in\N$, we denote by~$\hat{F}_t$ the empirical cdf based on~$X_1, \hdots, X_t$, i.e.,~~$\hat{F}_t(x) := \frac{1}{t}\sum_{s=1}^t\mathds{1}_{\cbr[0]{X_s\leq x}}$ for every~$x \in \R$. 
	
	We shall call (the distribution of) a random variable~$X$ \emph{sub-Gaussian}, if there exist positive real numbers~$D$ and~$\sigma^2$ such that~$\P\del[0]{|X|\geq x}\leq De^{-\frac{x^2}{2\sigma^2}}$ for every~$x>0$. Note that we do not require~$X$ to have expectation~$0$; we only require a two-sided exponential tail bound.
	\begin{lemma}\label{lem:mgf}
		Let~$X$ be sub-Gaussian. Then, for every~$s>0$ and~$k>0$,
		\begin{equation}\label{eqn:mgfexp}
			\E [e^{sX}] \leq 1+D\sqrt{2\pi}s\sigma e^{\frac{s^2\sigma^2}{2}}
			\leq
			\del[2]{1+D \sqrt{k/e}\sqrt{2\pi}} e^{(\frac{1}{2}+\frac{1}{2 k})s^2\sigma^2}.
		\end{equation}
	\end{lemma}
	
	\begin{proof}
		Note that~$\mathbb{E}(e^{sX}) \leq \mathbb{E}(e^{s|X|}) = \int_{(0, \infty)} \P( e^{s|X|} > x) dx$, which is bounded from above by
		$$ \int_{(0, \infty)}\P( |X| > \log(x)/s) dx  \leq 1 + D \int_{(1, \infty)} e^{-(\log(x)/s)^2/(2\sigma^2)}dx,$$
		which (after the change of variables~$x \mapsto e^{sx}$) is seen to coincide with~$$1 + s D \int_{(0, \infty)} e^{-x^2/(2\sigma^2) + sx} dx = 
		1 + s D e^{\frac{s^2\sigma^2}{2}} \int_{(0, \infty)} e^{-\frac{1}{2\sigma^2}(x - s\sigma^2)^2} dx \leq
		1+D\sqrt{2\pi}s\sigma e^{\frac{s^2\sigma^2}{2}}.$$ 
		
		To prove the second inequality in the lemma, use that for every~$x\geq 0$ it holds that~$\sqrt{x}\leq  \sqrt{k/e} \exp(x/2 k)$ [since~$e x/k\leq \exp(x/k)$] with~$x = s^2 \sigma^2$.
	\end{proof}
	
	We now summarize some consequences of Lemma~\ref{lem:mgf} (which hold for every~$t \in \N$): Together with the Dvoretzky-Kiefer-Wolfowitz-Massart (DKWM)-inequality, the first inequality in~\eqref{eqn:mgfexp} shows that (cf. also~\cite{shorack2009empirical}, p.~357) 
	\begin{equation}\label{eqn:mgf1}
		\mathbb{E}\left(\exp(s t \|\hat{F}_t - F\|_{\infty})\right) \leq 1 + \sqrt{2\pi} s \sqrt{t} e^{s^2t/8} \quad \text{ for every } s > 0;
	\end{equation}
	and the second inequality in~\eqref{eqn:mgfexp} shows that for every~$k > 0$
	\begin{equation}\label{eqn:mgf2}
		\mathbb{E}\left(\exp(s t \|\hat{F}_t - F\|_{\infty})\right) \leq 	
		\del[2]{1+ 2\sqrt{\frac{2\pi k}{e}}} e^{(1+k^{-1})s^2t/8} \quad \text{ for every } s > 0,
	\end{equation}
	and thus in particular
	\begin{equation}\label{eqn:mgf3}
		\mathbb{E}\left(\exp(s t \|\hat{F}_t - F\|_{\infty})\right) \leq 	
		\del[2]{1+ 2 \sqrt{2\pi }} e^{(1 + e^{-1})s^2t/8} \quad \text{ for every } s > 0.
	\end{equation}
	Doob's submartingale inequality, together with the inequality in Equation~\eqref{eqn:mgf2}, delivers the upper bound in the following lemma.
	\begin{lemma}\label{lem:maxDKW}
		For every~$n\in\N$, every~$x>0$, and every~$k>0$ it holds that
		\begin{align*}
			\P\del[2]{\max_{1\leq t\leq n}t\enVert[0]{\hat{F}_t-F}_\infty\geq x}
			\leq
			\del[2]{1+2\sqrt{k/e} \sqrt{2\pi}} e^{-\frac{2 x^2}{(1+1/k)n}}. 
		\end{align*}
	\end{lemma}
	\begin{proof}
		We fix~$n\in\N$ and~$s,x>0$. It holds that
		\begin{align*}
			\P\del[2]{\max_{1\leq t\leq n}t\enVert[0]{\hat{F}_t-F}_\infty\geq x}
			=
			\P\del[2]{\max_{1\leq t\leq n}e^{st\enVert[0]{\hat{F}_t-F}_\infty}\geq e^{sx}}.
		\end{align*}
		We now verify that~$\cbr[0]{e^{st\enVert[0]{\hat{F}_t-F}_\infty}}_{t\in\N}$ is a submartingale w.r.t.~$\cbr[0]{\mathcal{F}_t}_{t\in\N}$, the natural filtration generated by~$\cbr[0]{X_t}_{t\in\N}$. By Jensen's inequality for conditional expectations, it follows that for any pair of natural numbers~$t_1<t_2$,
		\begin{align}\label{eq:subMG}
			\E \del[1]{e^{st_2\enVert[0]{\hat{F}_{t_2}-F}_\infty}|\mathcal{F}_{t_1}}
			\geq
			e^{s\E\del[0]{t_2\enVert[0]{\hat{F}_{t_2}-F}_\infty|\mathcal{F}_{t_1}}}.
		\end{align}
		Next, since
		\begin{align*}
			t_2\enVert[0]{\hat{F}_{t_2}-F}_\infty
			=
			\sup_{z\in\R}\envert[2]{\sum_{r=1}^{t_2}(\mathds{1}\cbr[0]{X_{r}\leq z}-F(z))}
			\geq
			\envert[2]{\sum_{r=1}^{t_2}(\mathds{1}\cbr[0]{X_{r}\leq z_0}-F(z_0))},
		\end{align*}
		for all~$z_0\in\R$, another application of Jensen's inequality for conditional expectations yields
		\begin{align*}
			\E\del[1]{t_2\enVert[0]{\hat{F}_{t_2}-F}_\infty|\mathcal{F}_{t_1}}
			\geq
			\envert[2]{\sum_{r=1}^{t_2}\E\del[1]{\mathds{1}\cbr[0]{X_{r}\leq z_0}-F(z_0)|\mathcal{F}_{t_1}}}
			=
			\envert[2]{\sum_{r=1}^{t_1}(\mathds{1}\cbr[0]{X_{r}\leq z_0}-F(z_0))},
		\end{align*}
		the equality following from the independence of the random variables~$X_t$. Thus, we conclude that
		\begin{align*}
			\E\del[1]{t_2\enVert[0]{\hat{F}_{t_2}-F}_\infty|\mathcal{F}_{t_1}}
			\geq
			\sup_{z\in\R}\envert[2]{\sum_{r=1}^{t_1}(\mathds{1}\cbr[0]{X_{r}\leq z}-F(z))}
			=
			t_1\enVert[0]{\hat{F}_{t_1}-F}_\infty,
		\end{align*}
		which together with \eqref{eq:subMG} establishes that~$\cbr[0]{e^{st\enVert[0]{\hat{F}_t-F}_\infty}}_{t\in\N}$ is a (positive) submartingale. Hence, Doob's submartingale inequality together with Equation~\eqref{eqn:mgf2} shows for all~$s>0$ that
		\begin{align*}
			\P\del[2]{\max_{1\leq t\leq n}t\enVert[0]{\hat{F}_t-F}_\infty\geq x}
			\leq
			\frac{\E e^{sn\enVert[0]{\hat{F}_n-F}_\infty}}{e^{sx}} 
			\leq
			\del[2]{1+2 \sqrt{k/e} \sqrt{2\pi}} e^{(\frac{1}{2}+\frac{1}{2 k})\frac{n}{4}s^2-sx}.
		\end{align*}
		Minimizing the right hand side of the above display over~$s>0$ yields that
		\begin{align*}
			\P\del[2]{\max_{1\leq t\leq n}t\enVert[0]{\hat{F}_t-F}_\infty\geq x}
			\leq
			\del[2]{1+2\sqrt{k/e} \sqrt{2\pi}} e^{-\frac{2 x^2}{(1+1/k)n}}.
		\end{align*}
	\end{proof}
	
	\begin{lemma}\label{lem:sum_indsubgauss}
		Let~$X$ and~$Y$ be independent sub-Gaussian random variables with the same parameters~$D$ and~$\sigma^2$. Then, for every~$z>0$ and every~$k>0$,		
		\begin{align*}
			\P\del[0]{X+Y\geq z}
			\leq
			\del[2]{1+D\sqrt{k/e}\sqrt{2\pi}}^2e^{-\frac{z^2}{4(1+1/k)\sigma^2}}.
		\end{align*}
	\end{lemma}	
	
	\begin{proof}
		Note that by the independence of~$X$ and~$Y$, for all~$s>0$,~$k>0$, and~$z > 0$,
		\begin{align*}
			\P\del[0]{X+Y\geq z}
			=
			\P\del[0]{e^{s(X+Y)}\geq e^{sz}}
			\leq
			\frac{\E e^{sX}\E e^{sY}}{e^{sz}}
			\leq
			\del[2]{1+D\sqrt{k/e} \sqrt{2\pi}}^2 e^{(1+\frac{1}{k})s^2\sigma^2-s z},
		\end{align*}
		where the second estimate follows from Lemma \ref{lem:mgf}. Minimizing the right-hand side of the above display in~$s>0$ yields that
		\begin{align*}
			\P\del[0]{X+Y\geq z}
			\leq
			\del[2]{1+D \sqrt{k/e} \sqrt{2\pi}}^2e^{-\frac{z^2}{4(1+1/k)\sigma^2}}.
		\end{align*}
	\end{proof}
	
	In the following result, we interpret the maximum over an empty set as $0$.
	
	\begin{corollary}\label{cor:summax_DKW}
		Let~$K \in \N$, and let~$\cbr[0]{X_{1,t}}_{t\in\N},\hdots,\cbr[0]{X_{K,t}}_{t\in\N}$ be independent sequences of random variables. Suppose further that for every~$i =  1, \hdots, K$ the random variables~$\cbr[0]{X_{i,t}}_{t\in\N}$ are i.i.d.~with cdf~$F_{i}$. Denote by~$\hat{F}_{i,n}$ the empirical cdf based on~$X_{i, 1}, \hdots, X_{i, n}$. Then, for any~$A\subseteq \cbr[0]{1,\hdots,K}$, any~$x>0$,~$k > 0$ and~$n \in \N$
		\begin{align*}
			\P\del[2]{\max_{i\in A}\enVert[0]{\hat{F}_{i,n}-F_i}_\infty + \max_{i \notin A}\enVert[0]{\hat{F}_{i,n}-F_i}_\infty\geq x}
			\leq
			\del[2]{1+K\sqrt{k/e}\sqrt{8\pi}}^2e^{-\frac{x^2n}{(1+1/k)}}.
		\end{align*}
	\end{corollary}
	
	\begin{proof}
		By a union bound and the DKWM inequality~$\max_{i\in B}\enVert[0]{\hat{F}_{i,n}-F_i}_\infty$ for~$B \in \{A, A^c\}$ are sub-Gaussian with~$D=2K$ and~$\sigma^2=1/(4n)$. As the two random variables are also independent, Lemma \ref{lem:sum_indsubgauss} yields the desired estimate.
	\end{proof}
	\section{Proofs and additional results for Section~\ref{sec:lowbound}}\label{app:prfs1}
	
	\subsection{A general lower bound}\label{sec:LB}
	
	In this section, we start with a lower bound result that will be instrumental for proving Theorem~\ref{thm:lbn} and the first part of Theorem~\ref{Thm: UniformLowermixed}. We need to introduce some further notation. Assume that~$\mathscr{M}_K$ satisfies Assumption~\ref{as:M}, i.e.,~$\mathscr{M}_K$ is closed and contains at least two distinct elements. Given~$\mathcal{H} = \{\mathcal{H}_1, \hdots, \mathcal{H}_l\}$, a partition of~$\{1, \hdots K\}$, we define for every~$i = 1, \hdots, l$ the~$(K\times|\mathcal{H}_i|)$-dimensional matrix~$B_i = (e_j(K))_{j \in \mathcal{H}_i}$, and set~$\kappa(\mathcal{H}, \mathscr{M}_K)$ equal to
	\begin{equation}\label{eqn:kappaH}
		\sup \left\{
		\left[
		\max_{\delta \in \mathscr{M}_K} v'\delta - \sup_{\delta \in \mathscr{T}} v'\delta\right] \wedge \min_{i = 1}^l
		\left[
		\max_{\delta \in \mathscr{M}_K} (v + B_iw_i)'\delta 
		-
		\sup_{\delta \in  \mathscr{M}_K \setminus \mathscr{T}} (v + B_iw_i)'\delta
		\right]\right\},
	\end{equation}
	where the supremum is taken over all nonempty Borel sets~$\mathscr{T} \subsetneqq \mathscr{M}_K$, all~$v \in [-1,1]^K$, and all~$w_i \in [-1,1]^{|\mathcal{H}_i|}$ ($i = 1, \hdots, l$). Note that for the partition~$\mathcal{H} = \{\{1\}, \hdots, \{K\}\}$ the expression~$\kappa(\mathcal{H}, \mathscr{M}_K)$ coincides with~$\kappa(\mathscr{M}_K)$. Note further that~$\kappa(\mathcal{H}, \mathscr{M}_K) \geq 0$ (under Assumption~\ref{as:M}). The following lemma is crucial for establishing Theorem~\ref{thm:lbn}; the main result in this section is given subsequently. 
	\begin{lemma}\label{lem:Hsimp}
		Let~$\mathscr{M}_K$ satisfy Assumption~\ref{as:M}. Then
		\begin{equation*}
			4 \times \kappa\left(\{\{1, \hdots, K\}\}, \mathscr{M}_K\right) \geq \mathrm{diam}^2(\mathscr{M}_K) := \max \left\{ \|\nu - \gamma\|^2 : (\nu, \gamma) \in \mathscr{M}_K \times \mathscr{M}_K \right \}.
		\end{equation*}
	\end{lemma}
	\begin{proof}
		Let~$(\nu, \gamma)$ solve~$\max_{(\nu, \gamma) \in \mathscr{M}_K \times \mathscr{M}_K} \|\nu - \gamma\|^2$, and set~$v = (\nu - \gamma)/2$ and~$w = -2v$. The Cauchy-Schwarz inequality then shows that~$\nu$ solves~$\max_{\delta \in \mathscr{M}_K} v'(\delta - \gamma) = \|\nu - \gamma\|^2/2$, and~$\gamma$ solves~$\max_{\delta \in \mathscr{M}_K} (-v)'(\delta -  \nu) = \|\nu - \gamma\|^2/2$. For the non-empty Borel set $$\mathscr{T} :=\{
		\delta \in \mathscr{M}_K : v'(\delta - \gamma) \leq \|\nu - \gamma\|^2/4
		\} \subsetneqq \mathscr{M}_K,$$ we obtain
		\begin{equation*}
			\begin{aligned}
				&\max_{\delta \in \mathscr{M}_K} v'\delta - \sup_{\delta \in \mathscr{T}} v'\delta = \max_{\delta \in \mathscr{M}_K} v'(\delta - \gamma) - \sup_{\delta \in \mathscr{T}} v'(\delta - \gamma) \geq \frac{\|\nu - \gamma\|^2}{4}, \\
				&\max_{\delta \in \mathscr{M}_K} (v + w)'\delta - \sup_{\delta \in \mathscr{M}_K \setminus \mathscr{T}} (v + w)'\delta = 
				\max_{\delta \in \mathscr{M}_K} (-v)'(\delta-\nu) - \sup_{\delta \in \mathscr{M}_K \setminus \mathscr{T}} (-v)'(\delta-\nu) \geq \frac{\|\nu - \gamma\|^2}{4},
			\end{aligned}
		\end{equation*}
		which shows that for this choice of~$\mathscr{T}$,~$v$, and~$w = B_1 w_1 = w_1$ (and for the partition~$H = \{\{1, \hdots, K\}\}$) the term in braces in~\eqref{eqn:kappaH} is not smaller than~$\|\nu - \gamma\|^2/2$.
	\end{proof}
	\begin{theorem}\label{prop: UniformLowermixed}
		Suppose Assumptions~\ref{as:dgp},~\ref{as:MAIN},~\ref{as:M} and~\ref{as:lbcov} hold. 
		Then there exists a constant~$c > 0$, independent of~$K$,~$n$ and~$\mathscr{M}_K$, such that for every policy~$\pi$ with recommendations in~$\mathscr{M}_K$, and any randomization measure~$\P_G$, it holds that
		\begin{equation*}
			\sup_{F^1, \hdots, F^K \in \{J_{\tau}:\tau \in [0,1]\}} \mathbb{E}[r_n(\pi, {\mathscr{M}_K})] \geq c \max_{\mathcal{H} \in \{ \text{Partitions of } \{1, \hdots, K\} \}}\kappa(\mathcal{H}, \mathscr{M}_K) \sqrt{|\mathcal{H}|/n}, ~~\forall n \geq K,
		\end{equation*}
		where the supremum is taken over all potential outcome vectors with independent marginals and cdfs in~$\{J_{\tau}:\tau \in [0,1]\}$.
	\end{theorem}
	
	\begin{proof}
		We show that there exists a constant~$c$ as in the statement of the theorem, such that for every policy~$\pi$ with recommendations in~$\mathscr{M}_K$, any randomization measure~$\P_G$, and any partition~$\mathcal{H} = \{\mathcal{H}_1, \hdots, \mathcal{H}_l\}$ of~$\{1, \hdots, K\}$ it holds that
		\begin{equation}
			\sup_{F^1, \hdots, F^K \in \{J_{\tau}:\tau \in [0,1]\}} \mathbb{E}[r_n(\pi, {\mathscr{M}_K})] \geq c \kappa(\mathcal{H}, \mathscr{M}_K) \sqrt{l/n} , \quad \text{ for every } n \geq K.
		\end{equation}
		Fix a partition~$\mathcal{H}$, and abbreviate~$\kappa := \kappa(\mathcal{H}, \mathscr{M}_K)$ in what follows. If~$\kappa = 0$ there is nothing to show (recall that~$\kappa \geq 0$). We therefore assume throughout that~$\kappa > 0$. Let~$\varepsilon := 2/\sqrt{17} < 1/2$. Arguing as in Step~0 of the proof of Theorem 3.9 in \cite{kpv2}, we note that Assumption~\ref{as:lbcov} implies the existence of a~$\zeta > 0$, such that~$H_v := J_{1/2 + v}$ satisfies
		\begin{equation}\label{eqn:bilip}
			\mathsf{KL}^{1/2}(\mu_{H_{v}},\mu_{H_{w}})
			\leq \zeta |v-w|  \quad \text{ for every } v,  w \text{ in } [-\varepsilon, \varepsilon],
		\end{equation}
		where~$\mu_{H_v}$ denotes the (unique) probability measure corresponding to the cdf~$H_v$,~$\mathsf{KL}(\cdot,\cdot)$ denotes Kullback-Leibler divergence, and
		\begin{equation}\label{eqn:Tlowup}
			c_-(w - v) \leq \mathsf{T}(H_w) - \mathsf{T}(H_v) \leq C (w - v) \quad \text{ for every } v \leq w \text{ in } [-\varepsilon, \varepsilon],
		\end{equation}
		implying in particular that~$w \mapsto \mathsf{T}(H_w)$ is continuous and strictly increasing on~$[-\varepsilon, \varepsilon]$.
		
		Now, let~$\pi$ be a policy with recommendations in~$\mathscr{M}_K$. Fix~$n \geq K$, and denote~$\pi_{n, t} = \pi_t$ for~$t = 1, \hdots, n+1$. For~$t=1, \ldots, n$ and every~$ u  \in [-\varepsilon, \varepsilon]^K$ we denote by~$\mathbb{P}^t_{\pi, u}$ the distribution induced by the random vector~$Z_t$ (see~Equation~\eqref{eqn:zdef}) on the Borel sets of~$\mathbb{R}^{2t}$ when~$G_t$ are~i.i.d.~with distribution~$\mathbb{P}_G$, and~$Y_t = (Y_{1,t}, \hdots, Y_{K, t})$~with independent coordinates and marginal cdfs
		\begin{equation}\label{defdist}
			F^i = H_{u_i} \text{ for } i = 1, \hdots, K;
		\end{equation}
		correspondingly, we denote~$\mathbf{F}_{u}:= (H_{u_1}, \hdots, H_{u_K})$. We also denote the product measure~$\mathbb{P}^t_{\pi, u} \otimes \mathbb{P}_G$ by~$\tilde{\mathbb{P}}^t_{\pi, u}$, and write~$\mathbb{E}^t_{\pi, u}$ and~$\tilde{\mathbb{E}}^t_{\pi, u}$ for the expectation operators corresponding to~$\mathbb{P}^t_{\pi, u}$ and~$\tilde{\mathbb{P}}^t_{\pi, u}$, respectively. 
		
		For every~$\delta \in \mathscr{S}_K$ and every~$u \in [-\varepsilon, \varepsilon]^K$, we observe that
		\begin{equation}\label{eqn:mixwrite}
			\langle \delta, \mathbf{F}_{u} \rangle 
			= \sum_{j = 1}^K \delta_j J_{1/2 + u_j} 
			= J_{1/2 + \delta'u} = H_{\delta'u}.
		\end{equation}
		
		Recalling the definition of~$\kappa$ and the notation introduced before the statement of Theorem~\ref{prop: UniformLowermixed}, we choose~$v \in [-1,1]^K$ and~$w_i \in [-1,1]^{|\mathcal{H}_i|}$ for~$i = 1, \hdots, l$ and a Borel set~$\mathscr{T}$ satisfying~$\emptyset \neq \mathscr{T} \subsetneqq \mathscr{M}_K$, such that
		\begin{equation}\label{eqn:choosest}
			\min\left(
			\max_{\delta \in \mathscr{M}_K} v'\delta - \sup_{\delta \in \mathscr{T}} v'\delta,~~ \min_{i = 1}^l
			\left[
			\max_{\delta \in \mathscr{M}_K} (v + B_iw_i) '\delta
			-
			\sup_{\delta \in \mathscr{M}_K \setminus \mathscr{T}} (v + B_iw_i) '\delta
			\right]
			\right) > \kappa/2 > 0.
		\end{equation}
		Abbreviate~$\epsilon := \sqrt{l/n} \varepsilon$. Define~$v^* = (\epsilon/2) v \in [-\varepsilon, \varepsilon]^K$. Because the functions~$\pi_1, \hdots, \pi_n$ map into~$\mathcal{I}$ and since~$\mathcal{H}_1, \hdots, \mathcal{H}_l$ is a partition of~$\mathcal{I}$, it must hold that~$\sum_{j = 1}^l \sum_{t = 1}^n \mathds{1} \{ \pi_t \in \mathcal{H}_j \} = n$. Hence, there exists an index~$j^* \in \{1, \hdots, l\}$, such that
		\begin{equation}\label{eqn:obsin}
			\tilde{\mathbb{E}}^{n-1}_{\pi, v^*} \left[\sum_{t = 1}^n \mathds{1} \{ \pi_t \in \mathcal{H}_{j^*}  \} \right] \leq \frac{n}{l}.
		\end{equation}
		Fix such an index~$j^*$, define~$\overline{w} := v + B_{j^*} w_{j^*}$, and rescale this vector to obtain~$w^* := (\epsilon/2) \overline{w} \in [-\varepsilon, \varepsilon]^K$ (recalling that~$v$ and~$w_{j^*}$ have coordinates in~$[-1,1]$). By construction~$v^*_j = w_j^*$ for every~$j \notin \mathcal{H}_{j^*}$. 
		
		Equation~\eqref{eqn:mixwrite} implies
		\begin{equation}
			\max_{\delta \in \mathscr{M}_K} \mathsf{T}(\langle \delta, \mathbf{F}_{v^*} \rangle) - \sup_{\delta \in \mathscr{T}} \mathsf{T}(\langle \delta, \mathbf{F}_{v^*} \rangle) = \max_{\delta \in \mathscr{M}_K} \mathsf{T}(H_{\delta'v^*} ) - \sup_{\delta \in \mathscr{T}} \mathsf{T}(H_{\delta'v^*}).
		\end{equation}
		As observed after Equation~\eqref{eqn:Tlowup}, the function~$a \mapsto \mathsf{T}(H_a)$ is continuous and strictly increasing on~$[-\varepsilon, \varepsilon]$. Furthermore,~$\max_{\delta \in \mathscr{M}_K} [v^*]'\delta - \sup_{\delta \in \mathscr{T}} [v^*]'\delta >\epsilon \kappa/4$ as a consequence of~\eqref{eqn:choosest}. It thus follows from the lower bound in Equation~\eqref{eqn:Tlowup} that
		\begin{equation*}
			\max_{\delta \in \mathscr{M}_K} \mathsf{T}(\langle \delta, \mathbf{F}_{v^*} \rangle) - \sup_{\delta \in \mathscr{T}} \mathsf{T}(\langle \delta, \mathbf{F}_{v^*} \rangle) \geq 
			c_- \epsilon \kappa/4.
		\end{equation*}
		Likewise, it follows that
		\begin{equation*}
			\max_{\delta \in \mathscr{M}_K} \mathsf{T}(\langle \delta, \mathbf{F}_{w^*} \rangle) - \sup_{\delta \notin \mathscr{T}} \mathsf{T}(\langle \delta, \mathbf{F}_{w^*} \rangle) = \max_{\delta \in \mathscr{M}_K} \mathsf{T}(H_{\delta'w^*}) - \sup_{\delta \notin \mathscr{T}} \mathsf{T}(H_{\delta'w^*}) \geq c_- \epsilon \kappa/4.
		\end{equation*}
		Therefore, on the event~$\{\pi_{n+1} \in \mathscr{T} \}$ we have~$r_n(\pi, \mathscr{M}_K; \mathbf{F}_{v^*}; Z_n, G_{n+1}) \geq  c_- \epsilon \kappa/4$; whereas on the event~$\{\pi_{n+1} \notin \mathscr{T} \}$ we have~$r_n(\pi, \mathscr{M}_K; \mathbf{F}_{w^*}; Z_n, G_{n+1}) \geq  c_- \epsilon \kappa/4$. We conclude that 
		\begin{align*}
			\sup_{F^1, \hdots, F^K \in \{J_{\tau}:\tau \in [0,1]\}} \mathbb{E}[&r_n(\pi, {\mathscr{M}_K})] \geq  \frac{1}{2} \left[\tilde{\mathbb{E}}^n_{\pi, v^*}(r_n(\pi, {\mathscr{M}_K}))
			+ \tilde{\mathbb{E}}^n_{\pi, w^*}(r_n(\pi, {\mathscr{M}_K}))
			\right] \\ &\geq 
			\frac{c_- \epsilon \kappa}{8} \left[\tilde{\mathbb{E}}^n_{\pi, v^*}(\mathds{1}\{\pi_{n+1} \in \mathscr{T} \})
			+ \tilde{\mathbb{E}}^n_{\pi, w^*}(\mathds{1}\{\pi_{n+1} \notin \mathscr{T} \})
			\right] \\
			&\geq 	\frac{c_- \epsilon \kappa}{32} \exp \left( 
			- \mathsf{KL}(\mathbb{P}^n_{\pi, v^*},
			\mathbb{P}^n_{\pi, w^*})
			\right),
		\end{align*}
		where for the third inequality we used that the term in brackets is the sum of the Type~I and Type~II errors of the test~$\mathds{1}\{\pi_{n+1} \in \mathscr{T} \}$ for testing~$\tilde{\P}^n_{\pi, v^*}$ against~$\tilde{\P}^n_{\pi, w^*}$, together with, e.g., Theorem 2.2(iii) in~\cite{tsybakov2009introduction} and~$\mathsf{KL}(\tilde{\mathbb{P}}^n_{\pi, v^*},
		\tilde{\mathbb{P}}^n_{\pi, w^*}) = \mathsf{KL}(\mathbb{P}^n_{\pi, v^*},
		\mathbb{P}^n_{\pi, w^*})$, the latter following, e.g., from the chain rule for the Kullback-Leibler divergence (see, e.g., Lemma~B.1 in~\cite{kpv1} for a suitable statement). 
		
		Next note that (cf.~also Equation~\eqref{eqn:zdef}) for every~$u \in [-\varepsilon, \varepsilon]^K$ we can write 
		\begin{equation}\label{eqn:kern}
			\mathbb{P}^n_{\pi, u} = 
			\bigg( \sum_{j = 1}^K \mu_{H_{u_j}}
			\mathds{1} \{\pi_{n} = j\}
			\bigg)
			\otimes 
			\tilde{\mathbb{P}}^{n-1}_{\pi, u },
		\end{equation}
		as the term in parentheses defines a stochastic kernel on~$\mathcal{B}(\mathbb{R}) \times (\mathbb{R}^{2(n-1)} \times \mathbb{R})$. It thus follows from the chain rule for the Kullback-Leibler divergence, and since~$v^*_j = w^*_j$ holds for all~$j \notin \mathcal{H}_{j^*}$, that
		\begin{equation*}
			\mathsf{KL}(\mathbb{P}^n_{\pi, v^*},
			\mathbb{P}^n_{\pi,w^*}) \leq \mathsf{KL}(
			\mathbb{P}^{n-1}_{\pi,  v^* }
			, \mathbb{P}^{n-1}_{\pi,  w^* }
			) + \tilde{\mathbb{E}}^{n-1}_{\pi, v^* } \left( \mathds{1}\{\pi_n \in \mathcal{H}_{j^*} \} \right)\max_{j \in \mathcal{H}_{j^*}} \mathsf{KL}(\mu_{H_{v^*_{j}}}, \mu_{H_{w^*_{j}}}).
		\end{equation*}
		By induction and Equations~\eqref{eqn:bilip} and~\eqref{eqn:obsin}, it hence follows that
		\begin{equation*}
			\mathsf{KL}(\mathbb{P}^n_{\pi, v^*},
			\mathbb{P}^n_{\pi,u^*}) 
			\leq \frac{n}{l}  \zeta^2 \max_{j \in \mathcal{H}_{j^*}}(v_j^* - w_j^*)^2 \leq \frac{n}{l}  \zeta^2 \epsilon^2.
		\end{equation*}
		We thus arrive at
		\begin{equation*}
			\sup_{F^1, \hdots, F^K \in \{J_{\tau}:\tau \in [0,1]\}} \mathbb{E}[r_n (\pi, {\mathscr{M}_K})] \geq \frac{c_- \epsilon \kappa}{32}  \exp\left(- \frac{n}{l}  \zeta^2 \epsilon^2 \right) =
			\sqrt{l/n} \frac{c_- \varepsilon \kappa}{32} \exp\left(-  \zeta^2 \varepsilon^2 \right),
		\end{equation*}
		where the equality follows from~$\epsilon = \sqrt{l/n} \varepsilon$. 
	\end{proof}
	
	\subsection{Proof of Theorem~\ref{thm:lbn}}
	
	This follows immediately from Theorem~\ref{prop: UniformLowermixed} and Lemma~\ref{lem:Hsimp}.
	
	\subsection{Proof of Theorem~\ref{Thm: UniformLowermixed}}
	
	Equation~\eqref{eqn:rlow3} follows immediately from Theorem~\ref{prop: UniformLowermixed}. It remains to show that
	\begin{equation}\label{eqn:LB1A}
		\kappa(\mathscr{M}_K) \geq  \frac{1}{272} \times \max_{\delta \in \mathscr{M}_K} \|\delta\|^2 \times \min_{j = 1}^K (\max_{\delta \in \mathscr{M}_K} \delta_j - \min_{\delta \in \mathscr{M}_K} \delta_j)^3;
	\end{equation}
	and that
	\begin{equation}\label{eqn:LB2D}
		\kappa(\mathscr{M}_K) \geq (\min_{j = 1}^K \max_{\delta \in \mathscr{M}_K} \delta_j - 1/2)_+^2,
	\end{equation}
	where for a real number~$x$ we denote its positive part~$\max(x, 0)$ by~$(x)_+$.
	
	We collect some notation used in this proof. For~$j = 1, \hdots, K$ define
	\begin{equation*}
		\underline{\delta}_j := \min_{\delta \in \mathscr{M}_K} \delta_j, ~\overline{\delta}_j := \max_{\delta \in \mathscr{M}_K} \delta_j, ~\hat{\delta}_j := (\underline{\delta}_j+\overline{\delta}_j)/2, \text{ and } d_j := \frac{\overline{\delta}_j - \underline{\delta}_j}{2} = \overline{\delta}_j - \hat{\delta}_j = \hat{\delta}_j - \underline{\delta}_j \geq 0.
	\end{equation*}
	By Assumption~\ref{as:M} the minima and maxima in the previous display are all well defined. 
	
	We organize the proofs of the statements in Equations~\eqref{eqn:LB1A} and~\eqref{eqn:LB2D} in the following two subsections, respectively.
	
	\subsubsection{Proof of Equation~\eqref{eqn:LB1A}}
	
	Write the lower bound in Equation~\eqref{eqn:LB1A} as~$\lambda := 
	\|\mu\|^2\min_{i = 1}^K \frac{d_i^3}{34}$ for~$\mu \in \arg\max_{\delta \in \mathscr{M}_K} \|\delta\|$. Since~$\kappa(\mathscr{M}_K) \geq 0$ always holds, there is nothing to show if~$\lambda = 0$, and we shall thus assume~$\lambda > 0$ (implying $d_j > 0$ for $j = 1, \hdots, K$). To show~$\kappa(\mathscr{M}_K) \geq \lambda$, we now construct a non-empty Borel set~$\mathscr{T} \subsetneqq \mathscr{M}_K$, a vector~$v \in [-1,1]^K$, and weights~$w_i \in [-1,1]$ such that
	\begin{align}
		\label{eqn:ts1}
		&\max_{\delta \in \mathscr{M}_K} v'\delta - \sup_{\delta \in \mathscr{T}} v'\delta \geq \lambda, \\
		\label{eqn:ts2}
		&\max_{\delta \in \mathscr{M}_K} ~(v'\delta + w_i \delta_i)
		-
		\sup_{\delta \in \mathscr{M}_K \setminus \mathscr{T}} (v'\delta + w_i \delta_i)
		\geq \lambda \text{ for } i = 1, \hdots, K.
	\end{align}
	%
	We set~$v := \tau \mu \in [0,1]^K$ for~$\tau := \frac{8}{17} \min_{i = 1, \hdots, K} d_i$, 
	define the (Borel) set
	\begin{equation*}
		\mathscr{T} := \{ \delta \in \mathscr{M}_K : \|\mu\|^2 - \mu'\delta  \geq \tau^{-1} \lambda\} = \{ \delta \in \mathscr{M}_K : v'\mu - v'\delta  \geq \lambda \},
	\end{equation*}
	and, for~$i = 1, \hdots, K$, set~$w_i$ to~$-1$ if~$\mu_i \geq \hat{\delta}_i$, and to~$1$ if~$\mu_i < \hat{\delta}_i$. To proceed (and to show that~$\mathscr{T}$ is a non-empty strict subset of~$\mathscr{M}_K$, in particular), we need an auxiliary result: We claim that for every~$i = 1, \hdots, K$
	\begin{equation}\label{eqn:ineqT}
		\|\mu\|^2 -\mu'\gamma  \geq  \|\mu\|^2 
		d_i^2/16,~~ \forall \gamma \in \big\{\delta \in \mathscr{M}_K: 
		w_i \delta_i  \geq w_i [2\hat{\delta}_i+(1-w_i ) \underline{\delta}_i + (1+w_i) \overline{\delta}_i]/4
		\big\}.
	\end{equation}
	Note that the set in the previous display is non-empty, which follows upon choosing~$\gamma$ with $i$-th coordinate equal to~$\underline{\delta}_i$ (equal to~$\overline{\delta}_i$) if~$w_i = -1$ (if $w_i = 1$). To prove the claim fix~$i$, let~$\gamma \in \mathscr{M}_K$ be such that~$w_i \gamma_i \geq w_i (2\hat{\delta}_i+(1-w_i ) \underline{\delta}_i + (1+w_i) \overline{\delta}_i)/4$, and bound
	\begin{equation*}
		\frac{d_i^2}{4} \leq (\mu_i - \gamma_i)^2 \leq \|\mu-\gamma\|^2 = \|\mu\|^2 + \|\gamma\|^2 - 2\|\mu\|\|\gamma\|\cos(\theta),
	\end{equation*}
	for~$\theta$ such that~$\cos(\theta) = \mu'\gamma/(\|\mu\|\|\gamma\|)$. Noting that~$\|\mu\|\|\gamma\| \leq 1$, as~$\mu$ and~$\gamma$ are elements of~$\mathscr{S}_K$, we obtain for~$x := \|\mu\|/\|\gamma\| \geq 1$ that
	\begin{equation*}
		\cos(\theta) \leq \frac{1}{2}\left(x + x^{-1} - d_i^2/4 \right). 
	\end{equation*}
	If~$1 \leq x\leq 1+\frac{d_i^2}{8}$, then the upper bound just derived can be further upper bounded by~$1-\frac{d_i^2}{16}$, which allows us to conclude that
	\begin{equation}\label{eqn:betw}
		\mu'\gamma - \mu'\mu = \cos(\theta)\|\gamma\|\|\mu\| - \|\mu\|^2 \leq -\|\mu\|^2d_i^2/16.
	\end{equation}
	Multiplying by~$-1$ we obtain the first inequality in~\eqref{eqn:ineqT}. If~$x >  1+d_i^2/8$, then we obtain
	\begin{equation*}
		\mu'\gamma - \mu'\mu \leq \|\mu\|\|\gamma\| - \|\mu\|^2 = \|\mu\|^2(x^{-1}-1) \leq- \frac{\|\mu\|^2}{8}\left(\frac{d_i^2}{1+d_i^2/8}\right),
	\end{equation*}
	which is bounded from above by the upper bound in~\eqref{eqn:betw}, and which proves the claim.
	
	The equality~$\tau^{-1} \lambda = \|\mu\|^2\min_{i = 1}^K d_i^2/16$ implies that~$\mathscr{T}$ contains the non-empty set in~\eqref{eqn:ineqT} (for every~$i$). Hence~$\mathscr{T}$ is non-empty. Furthermore,~$\mu \notin \mathscr{T}$, as~$\tau^{-1} \lambda > 0$. This establishes that~$\mathscr{T}$ is strictly contained in~$\mathscr{M}_K$. By definition of~$\mathscr{T}$, this now delivers~\eqref{eqn:ts1}.
	
	We now prove~\eqref{eqn:ts2}, and shall use that for every~$\delta \in \mathscr{M}_K$ such that~$\delta \neq \mu$, the Cauchy-Schwarz inequality implies
	\begin{equation}\label{eqn:CSEc}
		v' \delta = \tau \mu'\delta < \tau \|\mu\| \|\delta\| \leq \tau \|\mu\|^2 = v'\mu,
	\end{equation}
	in particular~$\max_{\delta \in \mathscr{M}_K} v'\delta = v'\mu$. 
	Now, fix~$i \in \{1, \hdots, K\}$. Note that the definition of~$\mathscr{T}$,~$\tau^{-1} \lambda = \|\mu\|^2\min_{i = 1}^K d_i^2/16$ and the statement in~\eqref{eqn:ineqT} implies~$$
	w_i\delta_i < 
	w_i \left(2\hat{\delta}_i+(1-w_i ) \underline{\delta}_i + (1+w_i) \overline{\delta}_i\right)/4 
	\text{ for every } \delta \in \mathscr{M}_K \setminus \mathscr{T}.$$
	Suppose first that~$w_i = -1$. Then, the previous display implies~$\delta_i > (\hat{\delta}_i + \underline{\delta}_i)/2$. Let~$\eta \in \mathscr{M}_K$ be such that~$\eta_i = \underline{\delta}_i$, and let~$\delta \in \mathscr{M}_K \backslash \mathscr{T}$. Applying~\eqref{eqn:CSEc} we obtain
	\begin{equation*}
		(\eta'v + w_i\eta_i) - (\delta'v + w_i \delta_i) = \eta'v + (\delta_i - \underline{\delta}_i) - \delta'v \geq (\hat{\delta}_i/2 - \underline{\delta}_i/2) - \mu'v = \frac{d_i}{2} - \tau\|\mu\|^2,
	\end{equation*}
	which by the definition of~$\tau$ and~$\|\mu\|^2\leq 1$ is bounded from below by
	\begin{equation}
		\frac{d_i}{2} (1-16/17) = d_i/34 \geq \lambda,
	\end{equation}
	which proves Equation~\eqref{eqn:ts2}. Assume next that~$w_i = 1$, let~$\eta \in \mathscr{M}_K$ be such that~$\eta_i = \overline{\delta}_i$, and~$\delta \in \mathscr{M}_K \backslash \mathscr{T}$. Then one has~$\delta_i < (\hat{\delta}_i + \overline{\delta}_i)/2$, which gives
	\begin{equation*}
		(\eta'v + w_i\eta_i) - (\delta'v + w_i \delta_i) \geq \eta'v + (\overline{\delta}_i/2 -  \hat{\delta}_i/2) - \delta'v \geq \frac{d_i}{2} - \tau\|\mu\|^2,
	\end{equation*}
	which again proves Equation~\eqref{eqn:ts2}.

	\subsubsection{Proof of Equation~\eqref{eqn:LB2D}}
	
	If~$\min_{j = 1}^K \overline{\delta}_j - 1/2 \leq 0$, there is nothing to show. Hence, we assume from now on that~$\min_{j = 1}^K \overline{\delta}_j - 1/2 =: \rho > 0$, and thus in particular
	\begin{equation}\label{eqn:uprho}
		\overline{\delta}_j - 1/2 \geq \rho, \text{ for every } j = 1, \hdots, K.
	\end{equation}
	It remains to verify that~$\kappa(\mathscr{M}_K) \geq \rho^2$. Assuming~\eqref{eqn:uprho} there must exist an index~$j^* \in \{1, \hdots, K\}$, such that
	\begin{equation}\label{eqn:underdelta}
		\underline{\delta}_{j^*}  \leq \frac{1}{2} - \rho.
	\end{equation}
	Now, set~$v = (0, \hdots, 0, \rho, 0, \hdots, 0)' \in [-1,1]^K$ ($\rho$ being at the~$j^*$-th coordinate), and define for~$i = 1, \hdots,K$ the weights~$w_i = 1$ for~$i \neq j^*$ and~$w_{j^*} = -1$. Finally, define~$\mathscr{T} := \{\delta \in \mathscr{M}_K : \delta_{j^*} \leq 1/2\}$ (a non-empty Borel set, that does not coincide with~$\mathscr{M}_K$). 
	
	To verify~$\kappa(\mathscr{M}_K) \geq \rho^2$ we now show
	\begin{equation}
		\min\left(
		\max_{\delta \in \mathscr{M}_K} v'\delta - \sup_{\delta \in \mathscr{T}} v'\delta,~ \min_{i = 1}^K
		\left[
		\max_{\delta \in \mathscr{M}_K} (v'\delta + w_i  \delta_i )
		-
		\sup_{\delta \in \mathscr{M}_K \setminus \mathscr{T}} (v'\delta + w_i \delta_i)
		\right]
		\right) \geq \rho^2.
	\end{equation}
	We start with 
	\begin{equation*}
		\max_{\delta \in \mathscr{M}_K} v'\delta - \sup_{\delta \in \mathscr{T}} v'\delta = \max_{\delta \in \mathscr{M}_K} \rho \delta_{j^*} - \sup_{\delta \in \mathscr{M}_K: \delta_{j^*} \leq 1/2} \rho \delta_{j^*} \geq \rho \overline{\delta}_{j^*} - \rho/2 \geq \rho^2,
	\end{equation*}
	where the second inequality follows from~\eqref{eqn:uprho}. Next, consider an index~$i \neq j^*$ and write 
	\begin{equation*}
		\max_{\delta \in \mathscr{M}_K} (v'\delta + w_i \delta_i) 
		-
		\sup_{\delta \in \mathscr{M}_K \setminus \mathscr{T}} (v'\delta + w_i \delta_i)  = 
		\max_{\delta \in \mathscr{M}_K}
		(\delta_{j^*}\rho + \delta_i )
		-
		\sup_{\delta \in \mathscr{M}_K: \delta_{j^*} > 1/2 } (\delta_{j^*}\rho + \delta_i).
	\end{equation*}
	The maximum to the right can trivially be bounded from below by~$\overline{\delta}_i \geq 1/2 + \rho$; the supremum can be upper bounded by (observing that~$\rho \in (0, 1/2]$)
	\begin{equation*}
		\sup_{\delta \in \mathscr{M}_K: \delta_{j^*} > 1/2 } (\delta_{j^*}\rho + \delta_i) \leq 
		\sup_{\delta \in \mathscr{M}_K: \delta_{j^*} > 1/2 }[\delta_{j^*}\rho + (1-\delta_{j^*})] \leq 1 + (\rho - 1)/2 = 1/2 + \rho/2.
	\end{equation*}
	Combining what we have just observed yields
	\begin{equation*}
		\max_{\delta \in \mathscr{M}_K} (v'\delta + w_i \delta_i) 
		-
		\sup_{\delta \in \mathscr{M}_K \setminus \mathscr{T}} (v'\delta + w_i \delta_i) \geq  1/2 + \rho - [1/2 + \rho/2] = \rho/2 \geq \rho^2.
	\end{equation*}
	Finally, consider
	\begin{equation*}
		\max_{\delta \in \mathscr{M}_K} (v'\delta + w_{j^*} \delta_{j^*})
		-
		\sup_{\delta \in \mathscr{M}_K \setminus \mathscr{T}} (v'\delta + w_{j^*} \delta_{j^*})
		=
		\max_{\delta \in \mathscr{M}_K}
		\delta_{j^*}(\rho - 1)
		-
		\sup_{\delta \in \mathscr{M}_K: \delta_{j^*} > 1/2 } \delta_{j^*}(\rho - 1).
	\end{equation*}
	Since~$\rho-1$ is negative, the maximum to the right equals~$\underline{\delta}_{j^*}(\rho-1)$. For the same reason, the supremum  is bounded from above by~$(\rho-1)/2$. Together this implies the lower bound~$$(\rho-1)(\underline{\delta}_{j^*} - 1/2) = (1-\rho)(1/2 - \underline{\delta}_{j^*}) \geq \rho (1-\rho) \geq \rho^2,$$
	where we used Equation~\eqref{eqn:underdelta} and~$\rho \in (0, 1/2]$.
	
	\section{Proof of Theorem~\ref{thm:NMAPup}}\label{app:es}
	
	Before we prove the theorem, we verify that the functions~$\pi_{n,t}$ defined in Policy~\ref{pol:es} are measurable under Assumption~\ref{as:MB}: 
	It is obvious that for any~$n \in \N$ the functions~$\hat{\pi}_{n,t}$ for~$t = 1, \hdots, n$ are measurable, simply because they are constant. Assumption~\ref{as:MB} guarantees, for every~$n \in \N$ and every~$\delta \in \mathscr{S}_K$, the measurability of~$y \mapsto \mathsf{T}(\langle \delta, \mathbf{\hat{F}}_{n, \Pi_n}\rangle)$ (interpreted as a function on~$\R^n$). Thus, the measurability of the selection~$\hat{\pi}_{n, n+1}$ immediately follows, noting that for every~$\bar{\delta} \in \mathscr{M}_K^n$ we may write~$\{ \hat{\pi}_{n,n+1} = \bar{\delta}\}$ as the intersection~$$ \bigcap_{\delta \in \mathscr{M}_K^n : \delta \geq \bar{\delta}  } \{
	\mathsf{T}(\langle \bar{\delta}, \mathbf{\hat{F}}_{n, \Pi_n}\rangle) \geq
	\mathsf{T}(\langle \delta, \mathbf{\hat{F}}_{n, \Pi_n}\rangle) \} \cap
	\bigcap_{\delta \in \mathscr{M}_K^n : \delta < \bar{\delta} } \{
	\mathsf{T}(\langle \bar{\delta}, \mathbf{\hat{F}}_{n, \Pi_n}\rangle) >
	\mathsf{T}(\langle \delta, \mathbf{\hat{F}}_{n, \Pi_n}\rangle) \}.$$

	Now, we move on to the proof of Theorem~\ref{thm:NMAPup}: Fix~$n \geq K$, and denote~$\hat{\pi}_{n,t} = \hat{\pi}_t$ for~$t = 1, \hdots, n+1$. We start with the observation that, using convexity of~$\mathscr{D}$ and of~$D_{cdf}([a,b])$ together with Assumption~\ref{as:MAIN} for the second inequality,
	\begin{align*}
		\big|\max_{\delta \in \mathscr{M}^n_K} \mathsf{T}(\langle \delta, \mathbf{\hat{F}}_{n, \Pi_n} \rangle)
		-
		\max_{\delta \in \mathscr{M}_K^n} \mathsf{T}(\langle \delta, \mathbf{F} \rangle)\big| 
		&\leq 
		\max_{\delta \in \mathscr{M}_K^n}
		|\mathsf{T}(\langle \delta, \mathbf{\hat{F}}_{n, \Pi_n}\rangle) - \mathsf{T}(\langle \delta, \mathbf{F} \rangle)| \\
		&\leq C \max_{\delta \in \mathscr{M}_K^n}
		\| \langle \delta, \mathbf{\hat{F}}_{n, \Pi_n}\rangle
		-
		\langle \delta, \mathbf{{F}}\rangle
		\|_{\infty} \\  &\leq C \max_{j = 1}^K \|\hat{F}_{n, \Pi_n}^j - F^j\|_{\infty}.
	\end{align*}
	Now,~$\hat{\pi}_{n+1} = \min \argmax_{\delta \in \mathscr{M}^n_K} \mathsf{T}(\langle \delta, \mathbf{\hat{F}}_{n, \Pi_n}\rangle)$ and~\eqref{eqn:opte} shows that the regret~$r_n(\hat{\pi}, \mathscr{M}_K) =\max_{\delta \in \mathscr{M}_K} \mathsf{T}(\langle \delta, \mathbf{F} \rangle) - \mathsf{T}(\langle \hat{\pi}_{n+1}, \mathbf{F} \rangle)$ is upper bounded by
	\begin{equation*}
		\max_{\delta \in \mathscr{M}_K^n} \mathsf{T}(\langle \delta, \mathbf{F} \rangle) - \max_{\delta \in \mathscr{M}_K^n} \mathsf{T}(\langle \delta, \mathbf{\hat{F}}_{n, \Pi_n}\rangle ) 
		+
		\mathsf{T}(\langle \hat{\pi}_{n+1}, \mathbf{\hat{F}}_{n, \Pi_n}\rangle ) 
		- 
		\mathsf{T}(\langle \hat{\pi}_{n+1}, \mathbf{F} \rangle) + \varepsilon(n).
	\end{equation*}
	From~$\hat{\pi}_{n+1} \in \mathscr{M}_K^n$ we obtain
	\begin{equation*}
		\mathsf{T}(\langle \hat{\pi}_{n+1}, \mathbf{\hat{F}}_{n, \Pi_n} \rangle) - \mathsf{T}(\langle \hat{\pi}_{n+1}, \mathbf{F} \rangle) \leq \max_{\delta \in \mathscr{M}_K^n}
		|\mathsf{T}(\langle \delta, \mathbf{\hat{F}}_{n, \Pi_n}\rangle) - \mathsf{T}(\langle \delta, \mathbf{F} \rangle)|,
	\end{equation*}
	for which an upper bound has already been developed above. Summarizing, we obtain
	\begin{equation}\label{eqn:rup}
		r_n(\hat{\pi}, \mathscr{M}_K) \leq  \varepsilon(n) + 2C \max_{j = 1}^K \|\hat{F}_{n, \Pi_n}^j - F^j\|_{\infty}.
	\end{equation}
	Denote~$M_j := \|\hat{F}_{n, \Pi_n}^j - F^j\|_{\infty}$. Using Jensen's inequality, Equation~\eqref{eqn:mgf1} and~$\beta_n \leq |\Pi_{n,j}|$, for every~$j$, we get for every~$t > 0$ that
	\begin{equation}
		\exp\left[t \mathbb{E}(\max_{j = 1}^K M_j)\right] \leq \sum_{j = 1}^K \mathbb{E}(e^{\frac{t}{|\Pi_{n,j}|} |\Pi_{n,j}| M_j}) \leq K\left(1+\sqrt{ 2\pi/\beta_n}te^{t^2/(8\beta_n)} \right),
	\end{equation}
	or equivalently~
	$$\mathbb{E}(\max_{j = 1}^K M_j) \leq t^{-1} \log\left(
	K\left(1+\sqrt{ 2\pi/\beta_n}te^{t^2/(8\beta_n)} \right)
	\right).$$ 
	Upon inserting~$t = \sqrt{8\beta_n \log(K)a}$ for some~$a \geq 1$, the upper bound becomes
	\begin{equation*}
		\left(\frac{1}{\sqrt{8a}} + \sqrt{a} \frac{\log(1+4\sqrt{\pi  \log(K^a)} K^a)}{\sqrt{8}\log(K^a)} \right) \sqrt{\log(K)/\beta_n},
	\end{equation*}
	which, by Lemma~\ref{lem:calc} below (with~$c = 4\sqrt{\pi}$, thus~$e^{e/c^2} < 2^a$), is bounded from above by
	\begin{equation*}
		\frac{1}{\sqrt{8}}\left(\frac{1}{\sqrt{a}}  + \sqrt{a} \frac{\log(1+4\sqrt{\pi \log(2^a)}2^a )}{\log(2^a)} \right) \sqrt{\log(K)/\beta_n}.
	\end{equation*}
	This yields
	\begin{equation}\label{eqn:updec}
		\mathbb{E}(\max_{j = 1}^K M_j) \leq \inf_{a \geq 1} \left(\frac{\log(2+8\sqrt{\pi \log(2^a)}2^a )}{\log(2)\sqrt{8a}} \right) \sqrt{\frac{\log(K)}{\beta_n}}	 \leq 1.505 \times \sqrt{\frac{\log(K)}{\beta_n}}	,
	\end{equation}
	where the last inequality is obtained by setting~$a = 3$. Together with Equation~\eqref{eqn:rup}, this proves the theorem.
	
	\begin{lemma}\label{lem:calc}
		Let~$c > 0$. Then~$x \mapsto \frac{\log(1+c\sqrt{\log(x)}x)}{\log(x)}$ is strictly decreasing on~$(e^{e/c^2}, \infty)$.
	\end{lemma}
	
	\begin{proof}
		It suffices to show that~$x \mapsto \log(1+c x e^{x^2})/(x^2)$ is strictly decreasing on~$(\sqrt{e}/c, \infty)$. The derivative of this function at~$x \in (\sqrt{e}/c, \infty)$ equals 
		\begin{equation*}
			\frac{2cx^2e^{x^2}+ce^{x^2}}{x^2(cxe^{x^2}+1)}-\frac{2\log(cxe^{x^2}+1)}{x^3},
		\end{equation*}
		which is negative if and only if~$2\log(ce^{x^2}x+1) + ce^{x^2}x(2\log(ce^{x^2}x+1)-2x^2-1)$ is positive, which holds because of~$2\log(ce^{x^2}x+1) > 2x^2 + 1$, a consequence of~$cx \geq e^{1/2}$. 
	\end{proof}
	
	\section{Proof of Theorem~\ref{thm:FSApartition}}\label{sec:SEP}
	
	Instead of proving Theorem~\ref{thm:FSApartition}, we establish the following slightly stronger theorem, where we denote~$f_{n,K}(u):=\frac{1}{2}+\frac{u}{n}-\frac{uK}{2n}+\frac{1}{n}$ (for~$n > 0$).
	
	\begin{theorem}\label{thm:FSApartitionstg}Suppose Assumptions \ref{as:dgp}, \ref{as:MAIN}, and \ref{as:MB} hold, and that~$\mathscr{M}_K$ is as in Example~\ref{ex:compa} and satisfies Assumption~\ref{as:M}. Then, the SE policy~$\tilde{\pi}$ satisfies 
		\begin{equation}\label{eqn:UPBD}
			\sup_{\substack{F^i \in \mathscr{D} \\ i = 1, \hdots, K }}\E \left[ r_n(\tilde{\pi}, \mathscr{M}_K) \right] \leq \varepsilon(n) + C \times (2 \min(A_{n,K}, B_{n,K}) + C_{n,K}), \text{ for every } n > K,
		\end{equation}
		where
		\begin{align*}
			A_{n,K} &:= \frac{1+e}{\sqrt{2}e(1-K/n)}\sqrt{\log\del[1]{K\sbr[1]{1+(1+\sqrt{8\pi})\sqrt{2\pi}}}}\sqrt{f_{n,K}(\underline{r})}\frac{K}{\sqrt{n}} \\
			&\leq 
			\frac{1}{(1-K/n)}\sqrt{\log\del[1]{17 \times K }}\frac{K}{\sqrt{n}}
		\end{align*}
		\begin{align*}
			B_{n,K} &:= \sqrt{\frac{{1+e}}{{2e(1-K/n)}}\log\del[2]{K\sbr[1]{\underline{r}+\left\lceil\frac{n-\underline{r}K}{2} \right\rceil-\lfloor n/K\rfloor +1}\del[1]{1+\sqrt{8\pi}}}\frac{K}{n}} \\
			&\leq 
			\sqrt{\frac{1}{{(1-K/n)}}\log\del[2]{6.02\times K\sbr[1]{3+\frac{n}{2} - \frac{n}{K}}}\frac{K}{n}}
		\end{align*}
		and
		\begin{align*} 
			C_{n,K} &:= \min\biggl(\frac{\del[1]{1+K\sqrt{16\pi/(\eta e)}}^2}{(\sqrt{n}K)^{1+\eta/(2+\eta)}}\sum_{r\in\mathcal{R}}\frac{1}{r^{1+\eta/(2+\eta)}},\frac{2K^{(1-\eta)/2}}{n^{(1+\eta)/4}}\sum_{r\in\mathcal{R}}\frac{1}{r^{(1+\eta)/2}}\biggr) \\
			&\leq 
			\min\biggl(\frac{\del[1]{1+K\sqrt{16\pi/(\eta e)}}^2}{(\sqrt{n}K)^{1+\eta/(2+\eta)}}\frac{2(1+\eta)}{\eta},\frac{2K^{(1-\eta)/2}}{n^{(1+\eta)/4}}\sum_{r\in\mathcal{R}}\frac{1}{r^{(1+\eta)/2}}\biggr),
		\end{align*}
		where 
		\begin{align*}
			\frac{\del[1]{1+K\sqrt{16\pi/(\eta e)}}^2}{(\sqrt{n}K)^{1+\eta/(2+\eta)}}\frac{2(1+\eta)}{\eta} \leq 
			\frac{1+18.5 \times K^2/\eta }{(\sqrt{n}K)^{1+\eta/(2+\eta)}}\frac{4(1+\eta)}{\eta} \leq 
			\frac{K^{1-\frac{3\eta}{4+2\eta}} }{\sqrt{n}} \times 4(1+18.5\eta^{-1})^2.
		\end{align*}
	\end{theorem}
	
	To prove Theorem~\ref{thm:FSApartitionstg}, fix~$n \in \N$,~$n > K$, a partition~$\cbr[0]{A_1,\hdots,A_m}$ of~$\mathcal{I}$ such that~$m\geq 2$, a discretization~$\mathscr{M}_K^n = \bigcup_{j =1}^m \mathscr{M}_{A_j, K}^n$ (finite), such that~$\emptyset \neq \mathscr{M}_{A_j, K}^n \subseteq \mathscr{M}_{A_j,K}$ for every~$j = 1, \hdots, m$,  a set of elimination rounds~$\mathcal{R}$, and~$\eta>0$. Abbreviate~$\tilde{\pi}_{n,t} = \tilde{\pi}_t$ for~$t = 1, \hdots, n+1$.
	
	Before we establish the upper bound in Theorem~\ref{thm:FSApartitionstg}, we verify that the functions~$\tilde{\pi}_{n,t}$ (implicitly) defined in~Policy~\ref{pol:FSE_partition} are Borel measurable. For~$t = 1, \hdots, K$ the function~$\tilde{\pi}_t$ is constant, and hence Borel measurable. Consider next the case where~$K+1\leq t\leq n$. Since the policy does not depend on external randomization, we need to verify that~$\{z \in [a,b]^{t-1}: \tilde{\pi}_{t}(z) = s\}$ is a Borel set for every~$s \in \{1, \hdots, K\}$. Fix such an~$s$. It is tedious but not difficult to see that~$\{z \in [a,b]^{t-1}: \tilde{\pi}_{t}(z) = s\}$ can be written as a finite number of intersections and unions of sets (and their complements) of the form
	\begin{equation}
		A(\mathscr{S}, \Pi, c, \gamma) = 
		\{z \in [a,b]^{t-1}: \max_{\delta \in \mathscr{S}} \mathsf{T}(\langle \delta, \mathbf{F}_{\Pi}(z) \rangle) - \mathsf{T}(\langle \gamma, \mathbf{F}_{\Pi}(z)  \rangle) \leq c \};
	\end{equation}
	where~$\mathscr{S}$ is a finite subset of~$\mathscr{S}_K$;~$\Pi = (\Pi_1, \hdots, \Pi_K)$ where the~$\Pi_i$ are disjoint and non-empty subsets of~$\{1, \hdots, t-1\}$ (but $\Pi$ may not constitute a partition of~$\{1, \hdots, t-1\}$); where~$\mathbf{F}_{\Pi}(z) = (F_{\Pi_1}(z), \hdots, F_{\Pi_K}(z))$ with~$F_{\Pi_j}(\cdot) = |\Pi_j|^{-1} \sum_{i \in \Pi_j} \mathds{1}\{z_i \leq \cdot\}$;~$c > 0$; and~$\gamma \in \mathscr{S}_K$. That every such~$A(\mathscr{S}, \Pi, c, \gamma)$ is Borel measurable follows immediately from Assumption~\ref{as:MB}, which implies that~$\tilde{\pi}_{t}$ is measurable. The measurability of~$\tilde{\pi}_{n,n+1}$ is shown analogously, observing that for every~$\delta \in \mathscr{M}_K^n$ the set~$\{z \in [a,b]^n : \tilde{\pi}_{n,n+1}(z) = \delta \}$ can be written as a finite union and intersection of sets (and their complements) as in the previous display (but with~$t-1$ replaced by~$n$).
	
	To establish the upper bound in~\eqref{eqn:UPBD}, denote by~$r_C$ the last round completed, and denote the last elimination round completed by~$\bar{r}:=\max\cbr[0]{r\in\mathcal{R}:r\leq r_C}$. We have~$\bar{r}\leq \min(r_C,\lfloor n/2\rfloor)$, since~$r\leq \lfloor n/2\rfloor$ for all~$r\in\mathcal{R}$. Note that~$r_C$ and~$\bar{r}$ are random variables (Borel measurability can be verified as in the beginning of the proof; we suppress the dependence of~$r_C$ and~$\bar{r}$ on~$\omega\in \Omega$, for~$(\Omega, \mathcal{A}, \mathbb{P})$ the underlying probability space), and note that~$\mathcal{I}_{\bar{r}}=\mathcal{I}_{r_C}$,~$\mathcal{J}_{\bar{r}} = \mathcal{J}_{r_C}$, and~$\mathscr{M}_{A_j, K,\bar{r}}^{n}=\mathscr{M}_{A_j, K, r_C}^{n}$ for every~$j \in \mathcal{J}_{\bar{r}}$, since no elimination takes place after round~$\bar{r}$. Set
	\begin{equation*}
		\mathscr{M}_{K, \bar{r}}^n := 
		\bigcup_{j \in \mathcal{J}_{\bar{r}}} \mathscr{M}^n_{A_j, K, \bar{r}}.
	\end{equation*}
	Define the event (Borel measurability can again be verified as in the argument in the beginning of the proof) where not all~$\delta\in\argmax_{\delta \in \mathscr{M}_K^n} \mathsf{T}(\langle \delta, \mathbf{F} \rangle)$ have been eliminated after~$\bar{r}$ rounds
	\begin{align*}
		\mathcal{G} := \{\omega \in \Omega: \mathscr{M}_{K, \bar{r}}^n \cap \argmax_{\delta \in \mathscr{M}_K^n} \mathsf{T}(\langle \delta, \mathbf{F} \rangle)  \neq \emptyset \}.
	\end{align*} 
	
	Consider~$\omega \in \mathcal{G}$ in this paragraph: if~$|\mathcal{I}_{\bar{r}}|=1$, i.e., when there is only a single treatment left after the last completed elimination round, then the weights vector which puts mass~$1$ on that treatment must be an element of~$\argmax_{\delta \in \mathscr{M}_K^n} \mathsf{T}(\langle \delta, \mathbf{F} \rangle)$, and the regret is at most~$\varepsilon(n)$. Assume that~$\omega$ is such that~$|\mathcal{I}_{\bar{r}}|=|\mathcal{I}_{r_C}|\geq 2$. Then,~$\omega \in \mathcal{G}$ implies
	\begin{align*}
		r_n(\tilde{\pi}, \mathscr{M}_K)
		=
		\max_{\delta\in\mathscr{M}_K}\mathsf{T}(\langle\delta,\mathbf{F}\rangle)-\mathsf{T}(\langle\tilde{\pi}_{n+1},\mathbf{F}\rangle)
		\leq
		\max_{\delta\in\mathscr{M}_{K,\bar{r}}^{n}}\mathsf{T}(\langle\delta,\mathbf{F}\rangle)-\mathsf{T}(\langle\tilde{\pi}_{n+1},\mathbf{F}\rangle)+\eps(n).
	\end{align*}
	The right-hand side of the above display can also be written as
	\begin{align}\label{eqn:argrefapp1}
		\max_{\delta \in \mathscr{M}_{K,\bar{r}}^{n}} \mathsf{T}(\langle \delta, \mathbf{F} \rangle) - \mathsf{T}(\langle \tilde{\pi}_{n+1}, \mathbf{\hat{F}}_{n,n} \rangle ) 
		+
		\mathsf{T}(\langle \tilde{\pi}_{n+1}, \mathbf{\hat{F}}_{n,n} \rangle ) 
		- 
		\mathsf{T}(\langle \tilde{\pi}_{n+1}, \mathbf{F} \rangle) + \varepsilon(n),
	\end{align}
	which, since by definition~$\tilde{\pi}_{n+1}=\min \argmax\cbr[1]{ \mathsf{T}(\langle \delta, \mathbf{\hat{F}}_{n,n} \rangle):\delta\in\mathscr{M}_{K,\bar{r}}^{n}}$, equals
	\begin{equation*}
		\max_{\delta \in \mathscr{M}_{K,\bar{r}}^{n}} \mathsf{T}(\langle \delta, \mathbf{F} \rangle) - \max_{\delta \in \mathscr{M}_{K,\bar{r}}^{n}} \mathsf{T}(\langle \delta, \mathbf{\hat{F}}_{n,n}\rangle ) 
		+
		\mathsf{T}(\langle \tilde{\pi}_{n+1}, \mathbf{\hat{F}}_{n,n}\rangle ) 
		- 
		\mathsf{T}(\langle \tilde{\pi}_{n+1}, \mathbf{F} \rangle) + \varepsilon(n).
	\end{equation*}
	Next observe that~$\max_{\delta \in \mathscr{M}_{K,\bar{r}}^{n}} \mathsf{T}(\langle \delta, \mathbf{F} \rangle) - \max_{\delta \in \mathscr{M}_{K,\bar{r}}^{n}} \mathsf{T}(\langle \delta, \mathbf{\hat{F}}_{n,n}\rangle )~$ and~$\mathsf{T}(\langle \tilde{\pi}_{n+1}, \mathbf{\hat{F}}_{n,n}\rangle ) - 
	\mathsf{T}(\langle \tilde{\pi}_{n+1}, \mathbf{F} \rangle)$ are bounded from above by
	\begin{align*}
		\max_{\delta\in\mathscr{M}_{K,\bar{r}}^{n}}\envert[1]{\mathsf{T}(\langle \delta, \mathbf{\hat{F}}_{n,n}\rangle ) 
			- 
			\mathsf{T}(\langle \delta, \mathbf{F} \rangle)}
		\leq
		C\max_{\delta\in\mathscr{M}_{K,\bar{r}}^{n}}\enVert[1]{\langle \delta, \mathbf{\hat{F}}_{n,n}\rangle -\langle \delta, \mathbf{F} \rangle}_\infty 
		\leq
		C\max_{i\in \mathcal{I}_{\bar{r}}}\enVert[1]{\hat{F}_{i,n,n}-F^i}_\infty,
	\end{align*}
	where the first inequality follows from convexity of~$\mathscr{D}$ and Assumption \ref{as:MAIN}. Note that~$\underline{r}:=\min\mathcal{R}$ is the number of times each treatment in~$\mathcal{I}$ has been assigned by the end of the first elimination round. Together with~$|\mathcal{I}_{\bar{r}}|=|\mathcal{I}_{r_C}|\geq 2$ this implies~
	\begin{equation}\label{eqn:argrefapp2}
		\lfloor n/K\rfloor \leq S_{i}(n) \leq \underline{r}+ \left\lceil\frac{n-\underline{r}K}{2} \right\rceil =: R, \text{ for every } i\in \mathcal{I}_{\bar{r}},
	\end{equation}
	(note that~$R$ is non-random, as~$\underline{r} =\min \mathcal{R} \leq n/K$~is fixed), from which it follows that
	\begin{align*}
		\max_{i\in \mathcal{I}_{\bar{r}}}\enVert[1]{\hat{F}_{i,n,n}-F^i}_\infty
		&\leq
		\max_{i=1}^K\sbr[2]{\enVert[1]{\hat{F}_{i,n,n}-F^i}_\infty \mathds{1}{\cbr[1]{S_i(n)\in\cbr[0]{\lfloor n/K\rfloor,\hdots,R}}}}\\
		&=
		\max_{i=1}^K\sbr[2]{\sum_{s=\lfloor n/K\rfloor}^{R}\enVert[1]{\hat{F}_{i,n,n}-F^i}_\infty \mathds{1}{\cbr[1]{S_i(n)=s}}}.
	\end{align*}
	
	Using that~$r_n(\tilde{\pi},\mathscr{M}_K)\leq C$ for~$\omega \in \Omega \setminus \mathcal{G}$ by Assumption~\ref{as:MAIN}, we conclude that for every~$\omega \in \Omega$ it holds that
	\begin{equation*}
		r_n(\tilde{\pi}, \mathscr{M}_K) \leq 2 C \max_{i=1}^K\sbr[2]{\sum_{s=\lfloor n/K\rfloor}^{R}\enVert[1]{\hat{F}_{i,n,n}-F^i}_\infty \mathds{1}{\cbr[1]{S_i(n)=s}}} + \varepsilon(n) + C \mathds{1}\{\Omega \setminus \mathcal{G}\}.
	\end{equation*}
	The theorem now follows from Parts~3 and~4 of the auxiliary Lemma \ref{lem:OptionalSkipping} given next (the remaining upper bounds on~$A_{n,K}$,~$B_{n,K}$ and~$C_{n,K}$ being obvious).
	
	While we continue to use the notation and notational conventions introduced already, to formulate and prove Lemma \ref{lem:OptionalSkipping}, we define the sequence of random variables~$f_t$ as:
	\begin{align*}
		f_t :=\begin{cases}
			\tilde{\pi}_{n,t}(Z_{t-1}) & \text{for } 1\leq t \leq n\\
			(t~\mathrm{mod}~K)+1 & \text{ for } t > n.
		\end{cases}
	\end{align*} 
	For every~$i\in\mathcal{I}$ and~$t\in\N$ denote the random variable~$\tau_{i,t}$ by 
	\begin{align*}
		\tau_{i,t}:=\min\cbr[2]{r\in\N: \sum_{j=1}^r\mathds{1}\cbr[0]{f_j=i}=t}.
	\end{align*}

	\begin{lemma}\label{lem:OptionalSkipping} 
		It holds that 
		\begin{enumerate}
			\item The sequences of random variables~$\del[1]{Y_{i,\tau_{	i,t}}}_{t\in\N}$ and~$\del[1]{Y_{i,t}}_{t\in\N}$ have the same distribution for every~$i\in\mathcal{I}$.
			\item The sequences of random variables~$\del[1]{Y_{1,\tau_{1,t}}}_{t\in\N},\hdots,\del[1]{Y_{K,\tau_{K,t}}}_{t\in\N}$ are independent. 
			\item It holds that
			\begin{equation*}
				\E \max_{i=1}^K\sbr[2]{\sum_{s=\lfloor n/K\rfloor}^{R}\enVert[1]{\hat{F}_{i,n,n}-F^i}_\infty \mathds{1}{\cbr[1]{S_i(n)=s}}} \leq A_{n,K} \wedge B_{n,K}.
			\end{equation*}
			\item~$1-\P(\mathcal{G})$ is bounded from above by
			\begin{equation*}
				\frac{\del[1]{1+K\sqrt{16\pi/(\eta e)}}^2}{(\sqrt{n}K)^{1+\eta/(2+\eta)}}\sum_{r\in\mathcal{R}}\frac{1}{r^{1+\eta/(2+\eta)}} \leq \frac{\del[1]{1+K\sqrt{16\pi/(\eta e)}}^2}{(\sqrt{n}K)^{1+\eta/(2+\eta)}}\frac{2(1+\eta)}{\eta},
			\end{equation*}
			and also by
			\begin{equation*}
				\frac{2K^{(1-\eta)/2}}{n^{(1+\eta)/4}}\sum_{r\in\mathcal{R}}\frac{1}{r^{(1+\eta)/2}}.
			\end{equation*}
		\end{enumerate}
	\end{lemma} 
	\begin{proof}
		Parts~1 and~2 of the lemma follow from an optional sampling theorem. A suitable result in our context is Theorem~1 in \cite{belisle2008independence}, cf.~also their Section 5.3. The conditions in their Theorem~1 are easy to verify using the natural filtration generated by the i.i.d.~sequence~$(Y_t)_{t \in \N}$.
		
		To prove Part~3, for every~$s = 1, \hdots, n$ set~$\check{F}_{i,s}(\cdot):=\frac{1}{s}\sum_{r=1}^s\mathds{1}\cbr[0]{Y_{i,\tau_{i,r}}\leq \cdot}$, and note that for~$\omega$ such that~$S_i(n) = s$ it holds that~$\hat{F}_{i,n,n} = \check{F}_{i,s}$ by construction of~$f_t$. Therefore, the left-hand side in the inequality in Part~3 coincides with
		\begin{align*}
			\E \max_{i=1}^K\sbr[2]{\sum_{s=\lfloor n/K\rfloor}^{R}\enVert[0]{\check{F}_{i,s}-F^i}_\infty \mathds{1}{\cbr[1]{S_i(n)=s}}}
			\leq
			\E \max_{i=1}^K\max_{s\in\cbr[0]{\lfloor n/K\rfloor,\hdots,R}}\enVert[0]{\check{F}_{i,s}-F^i}_\infty,
		\end{align*}
		where the inequality follows from~$\sum_{s = \lfloor n/K\rfloor}^R \mathds{1}\{S_i(n) = s\} \leq 1$. Part~2 of this lemma implies that~$\check{M}_i:=\max_{s\in\cbr[0]{\lfloor n/K\rfloor,\hdots,R}}\enVert[0]{\check{F}_{i,s}-F^i}_\infty$ for~$ i\in\mathcal{I}$ are independent random variables. Furthermore, by Part~1 of this lemma and since~$(Y_{i,t})$ is an i.i.d.~sequence, the random variables~$\check{M}_i$ for~$i = 1, \hdots, K$ and~$M_i$ for~$i = 1, \hdots, K$ have the same joint distribution for
		\begin{align*}
			M_i := \max_{s\in\cbr[0]{\lfloor n/K\rfloor,\hdots,R}}\enVert[0]{F_{i,s}-F^i}_\infty,\quad \text{where}\quad F_{i,s}(\cdot):=\frac{1}{s}\sum_{r=1}^s\mathds{1}\cbr[0]{Y_{i,i+(r-1)K}\leq \cdot}.
		\end{align*}
		We conclude that the random variables~$\max_{i = 1}^K \check{M}_i$ and~$\max_{i = 1}^K M_i$ have the same distribution and proceed by upper bounding~$\E\max_{i = 1}^K M_i$. To this end, note that by Lemma \ref{lem:maxDKW} (applied with~$k = e$) for any~$x>0$ the probability~$\P(M_i>x)$ is not greater than
		\begin{align*}
			\P\del[2]{\max_{s\in\cbr[0]{\lfloor n/K\rfloor,\hdots,R}}s\enVert[0]{F_{i,s}-F^i}_\infty>\lfloor n/K\rfloor x}\leq
			\del[2]{1+2 \sqrt{2\pi}}\exp\del[2]{-\frac{2x^2\lfloor n/K\rfloor^2}{(1+e^{-1})R}}.
		\end{align*}
		By definition~$R=\underline{r}+\lceil\frac{n-\underline{r}K}{2} \rceil$, thus
		\begin{align*}
			\frac{\lfloor n/K\rfloor^2}{R}
			\geq 
			\frac{(\frac{n}{K}-1)^2}{\underline{r}+\frac{n-\underline{r}K}{2}+1}
			=
			\frac{n^2(\frac{1}{K}-\frac{1}{n})^2}{n(\frac{1}{2}+\frac{\underline{r}}{n}-\frac{\underline{r}K}{2n}+\frac{1}{n})}
			=
			\frac{n(1-a)^2}{f_{n,K}(\underline{r})K^2},
		\end{align*}
		where~$f_{n,K}(u):=\frac{1}{2}+\frac{u}{n}-\frac{uK}{2n}+\frac{1}{n}$ for~$2\leq K\leq n$, and~$a = K/n$. Clearly~$f_{n,K}(\cdot)$ is non-increasing, and~$\frac{1}{K}\leq f_{n,K}(u)\leq 1$ for~$u\in[1,\lfloor \frac{n}{K}\rfloor]$. Since~$\underline{r} \in [1,\lfloor \frac{n}{K}\rfloor]$, we conclude
		\begin{align*}
			\P(M_i > x)
			\leq
			\del[2]{1+2 \sqrt{2\pi}}\exp\del[3]{-\frac{2x^2n(1-a)^2}{(1+e^{-1})f_{n,K}(\underline{r})K^2}} \text{ for any } x > 0.
		\end{align*}
		Lemma \ref{lem:mgf} with~$k = e$,~$D=\del[1]{1+\sqrt{8\pi}}$ and~$\sigma^2=\frac{(1+e^{-1})f_{n,K}(\underline{r})K^2}{4(1-a)^2n}$ yields 
		\begin{align}\label{eq:mgf1}
			\E e^{t M_i}
			\leq
			\alpha_1 e^{\alpha_2 t^2 } \text{ for any } t > 0,
		\end{align}
		with~$\alpha_1 := \sbr[1]{1+ D \sqrt{2\pi}}$ and~$\alpha_2 := (\frac{1}{2}+\frac{1}{2e})\sigma^2$. Hence, by Jensen's inequality
		\begin{align*}
			\exp\del[1]{t\E \max_{i=1}^K M_i}
			\leq
			\sum_{i=1}^K \E e^{t M_i}
			\leq
			K\alpha_1 e^{\alpha_2 t^2}, \text{ for any } t>0,
		\end{align*}
		and we obtain
		\begin{align*}
			\E \max_{i=1}^K M_i
			\leq
			\frac{\log\del[1]{K\alpha_1}}{t}+\alpha_2 t, \text{ for any } t>0.
		\end{align*}
		Setting~$t =  \sqrt{\log(K\alpha_1)/\alpha_2}$ now yields
		\begin{align}\label{eq:bound1pre}
			\E\max_{i=1}^K M_i
			\leq
			2\sqrt{\log\del[1]{K\alpha_1 }\alpha_2},
		\end{align}
		which upon inserting~$\alpha_1$,~$\alpha_2$,~$D$ and~$\sigma$ yields
		\begin{align*}
			\E\max_{i=1}^K M_i
			\leq
			\frac{1+e}{\sqrt{2}e(1-a)}\sqrt{\log\del[1]{K\sbr[1]{1+(1+\sqrt{8\pi})\sqrt{2\pi}}}}\frac{K}{\sqrt{n}}\sqrt{f_{n,K}(\underline{r})} = A_{n,K}.
		\end{align*}
		
		To show~$\E\max_{i=1}^K M_i \leq B_{n,K}$, note that for all~$t > 0$,~$\E e^{tM_i} \leq \sum_{s=\lfloor n/K\rfloor}^R\E e^{t\enVert[0]{F_{i,s}-F^i}_\infty}$ which applying the bound in Equation~\eqref{eqn:mgf3} is further upper bounded by
		\begin{align*}
			\sum_{s=\lfloor n/K\rfloor}^R\del[1]{1+\sqrt{8\pi}}e^{(1+\frac{1}{e})\frac{1}{8s}t^2}\leq
			\del[1]{R-\lfloor n/K\rfloor + 1}\del[1]{1+\sqrt{8\pi}}e^{(\frac{1}{8}+\frac{1}{8e})\frac{1}{\lfloor\frac{n}{K}\rfloor}t^2}.
		\end{align*}
		Applying the argument that led from~\eqref{eq:mgf1} to~\eqref{eq:bound1pre} with~$\alpha_1 = \del[1]{R-\lfloor n/K\rfloor + 1}\del[1]{1+\sqrt{8\pi}}$ and~$\alpha_2 = (\frac{1}{8}+\frac{1}{8e})\frac{1}{\lfloor\frac{n}{K}\rfloor}$ we obtain (using~$\lfloor n/K\rfloor \geq (n/K) -1$) that
		\begin{align*}
			\E\max_{i=1}^K M_i &\leq 
			2\sqrt{\log\del[1]{K\del[1]{R-\lfloor n/K\rfloor + 1}\del[1]{1+\sqrt{8\pi}}}}\del[1]{\frac{1}{8}+\frac{1}{8e}}^{0.5}\frac{1}{\lfloor\frac{n}{K}\rfloor^{0.5}} \leq B_{n,K}.
		\end{align*}
		
		For Part~4, pick a~$\delta^* \in \argmax_{\delta \in \mathscr{M}_K^n} \mathsf{T}(\langle\delta,\mathbf{F}\rangle)$. Let~$i^*$ be the index such that~$\delta^*\in\mathscr{M}_{A_{i^*}, K}^{n}$. For~$\omega \in \Omega \setminus \mathcal{G}$ we certainly have~$\mathscr{M}_{K, \overline{r}}^n \not \ni \delta^*$. Hence, for such~$\omega$, there exists an elimination round~$r_*(\omega) \in \mathcal{R} $, say, where~$\delta^*$ is~\emph{removed} from~$\mathscr{M}^n_{A_{i^*}, K, r_*(\omega)}$. That is, there exists an index~$l_*(\omega) \neq i^*$, say,  a~$\delta_*(\omega) \in \mathscr{M}^n_{A_{l_*(\omega)}, K, r_*(\omega)}$, and an index~$t(\omega)$ such that $$\mathsf{T}(\langle \delta_*(\omega), \hat{\mathbf{F}}_{t(\omega),n} \rangle) - \mathsf{T}(\langle \delta^*, \hat{\mathbf{F}}_{t(\omega),n} \rangle) > u_{\eta}(r_*(\omega), n).$$ Note that~$\delta_*(\omega)$ places all its weight on indices in~$A_{l_*(\omega)}$, and~$\delta^*$ places all its weight on indices in~$A_{i^*}$. Furthermore, by definition of the policy~$\tilde{\pi}$, as long as a set~$A \in \{A_1, \hdots, A_m\}$ of treatments has not been \emph{eliminated}, each treatment in that set is assigned once per round. Therefore, the cdfs in coordinates of~$\hat{\mathbf{F}}_{t(\omega),n}$ with indices in~$A_{l_*(\omega)} \cup A_{i^*}$ are all based on~$r_{*}(\omega)$ observations. Thus, we may equivalently replace~$\hat{\mathbf{F}}_{t(\omega),n}$ by~$\check{\mathbf{F}}_{r_*(\omega)} = (\check{F}_{1,r_*(\omega)}, \hdots, \check{F}_{K,r_*(\omega)})$ in the previous display, where~$\check{F}_{i,s}(\cdot):=\frac{1}{s}\sum_{r=1}^s\mathds{1}\cbr[0]{Y_{i,\tau_{i,r}}\leq \cdot}$ was defined in the Proof of Part 3 above. Hence, we proceed with
		\begin{align*}
			\Omega \setminus \mathcal{G} \subseteq &\bigcup_{r\in\mathcal{R}}\bigcup_{l \neq i^*}\bigcup_{\delta_*\in \mathscr{M}_{A_l, K}^{n}}\cbr[1]{\mathsf{T}(\langle \delta_*, \check{\mathbf{F}}_{r} \rangle)-\mathsf{T}(\langle \delta^*, \check{\mathbf{F}}_{r}\rangle)>u_\eta(r,n)}\\
			\subseteq &
			\bigcup_{r\in\mathcal{R}}\bigcup_{l\neq i^*}\bigcup_{\delta_*\in \mathscr{M}_{A_l,K}^{n}}\cbr[1]{\mathsf{T}(\langle \delta_*, \check{\mathbf{F}}_{r}\rangle)-\mathsf{T}(\langle \delta_*, \mathbf{F}\rangle)+\mathsf{T}(\langle \delta^*, \mathbf{F}\rangle)-\mathsf{T}(\langle \delta^*, \check{\mathbf{F}}_{r}\rangle)>u_\eta(r,n)}\\
			\subseteq &
			\bigcup_{r\in\mathcal{R}}\bigcup_{l \neq i^*}\cbr[2]{\max_{i\in A_l}\enVert[0]{\check{F}_{i,r}-F^i}_\infty+ \max_{i\in A_{i^*}}\enVert[0]{\check{F}_{i,r}-F^i}_\infty>\frac{u_\eta(r,n)}{C}}\\
			= &
			\bigcup_{r\in\mathcal{R}}\cbr[2]{ \max_{i\in \{1, \hdots, m\} \setminus A_{i^*}}\enVert[0]{\check{F}_{i,r}-F^i}_\infty+\max_{i\in A_{i^*}}\enVert[0]{\check{F}_{i,r}-F^i}_\infty>\frac{u_\eta(r,n)}{C}}.
		\end{align*}
		Abbreviating~$A^c_{i^*} = \{1, \hdots, m\} \setminus A_{i^*}$, we thus have
		\begin{align}\label{eq:fork}
			\P(\Omega \setminus \mathcal{G} )
			\leq
			\sum_{r\in\mathcal{R}}\P\del[2]{\max_{i\in A_{i^*}^c}\enVert[0]{\check{F}_{i,r}-F^i}_\infty+\max_{i\in A_{i^*}}\enVert[0]{\check{F}_{i,r}-F^i}_\infty>\frac{u_\eta(r,n)}{C}}.
		\end{align}
		From Part~2 of this lemma we already know that, for every~$r\in\N$, the random variables~$\enVert[0]{\check{F}_{i,r}-F^i}_\infty$ for~$i\in\mathcal{I}$ are independent. Part 1 furthermore shows that the same statement holds if~$\check{F}_{i,r}$ is replaced by~$F_{i,r}$, and that (for every~$r \in \N$ and every~$i = 1, \hdots, K$) the distributions of ~$\enVert[0]{\check{F}_{i,r}-F^i}_{\infty}$ and~$\enVert[0]{F_{i,r}-F^i}_{\infty}$ coincide. Consequently, for every~$r \in \N$ the random variables
		\begin{align*}
			\max_{i\in A_{i^*}^c}\enVert[0]{\check{F}_{i,r}-F^i}_\infty+\max_{i\in A_{i^*}}\enVert[0]{\check{F}_{i,r}-F^i}_\infty\quad \text{and} \quad \max_{i\in A_{i^*}^c}\enVert[0]{F_{i,r}-F^i}_\infty+\max_{i\in A_{i^*}}\enVert[0]{F_{i,r}-F^i}_\infty,
		\end{align*}
		have the same distribution. Applying Corollary \ref{cor:summax_DKW} with~$k=2/\eta$ to each summand in the upper bound of~\eqref{eq:fork}, after replacing the cdfs~$\check{F}_{i,r}$ by the cdfs~$F_{i,r}$, we get
		\begin{align}\label{eqn:argrefapp3}
			\P(\Omega \setminus \mathcal{G})
			&\leq \del[2]{1+K\sqrt{16\pi/(\eta e)}}^2
			\sum_{r\in\mathcal{R}}e^{-\frac{{u_\eta}^2(r,n)r}{(1+\eta/2)C^2}},
		\end{align}
		which, using that
		\begin{align*}
			u_\eta(r,n)
			&=
			C\sqrt{\frac{(1+\eta/2)(1+\eta/(2+\eta))}{r}[0.5\log(n)+\log(rK)]},
		\end{align*}
		coincides with
		\begin{align*}
			\frac{\del[1]{1+K\sqrt{16\pi/(\eta e)}}^2}{(\sqrt{n}K)^{1+\eta/(2+\eta)}}\sum_{r\in\mathcal{R}}\frac{1}{r^{1+\eta/(2+\eta)}} \leq \frac{\del[1]{1+K\sqrt{16\pi/(\eta e)}}^2}{(\sqrt{n}K)^{1+\eta/(2+\eta)}}\frac{2(1+\eta)}{\eta},
		\end{align*}
		where we used~$\sum_{r=1}^\infty\frac{1}{r^{1+b}}\leq 1+\int_1^\infty \frac{1}{x^{1+b}}dx$ for all~$b>0$.
		
		After replacing the cdfs~$\check{F}_{i,r}$ by the cdfs~$F_{i,r}$ in~\eqref{eq:fork}, we can use a union bound and the DKWM inequality to obtain
		\begin{align*}
			\P(\Omega \setminus \mathcal{G})
			&\leq
			\sum_{r\in\mathcal{R}}\sum_{i\in A_{i^*}^c}\P\del[2]{\enVert[0]{F_{i,r}-F^i}_\infty>\frac{u_\eta(r,n)}{2C}}
			+
			\sum_{r\in\mathcal{R}}\sum_{i\in A_{i^*}}\P\del[2]{\enVert[0]{F_{i,r}-F^i}_\infty>\frac{u_\eta(r,n)}{2C}}\\
			&\leq
			2K\sum_{r\in\mathcal{R}}e^{-\frac{ u_\eta^2(r,n) r}{2C^2}} =
			\frac{2K^{(1-\eta)/2}}{n^{(1+\eta)/4}}\sum_{r\in\mathcal{R}}\frac{1}{r^{(1+\eta)/2}}.
		\end{align*}
		This establishes Part 4.
	\end{proof}
	
	\section{Proof of Theorem \ref{thm:FSAcheck}}\label{sec:elimwincp}
	
	Analogously to the Proof of Theorem~\ref{thm:FSApartition}, we establish the stronger upper bound
	\begin{equation}\label{eqn:UPBD2}
		\sup_{\substack{F^i \in \mathscr{D} \\ i = 1, \hdots, K }}\E \left[ r_n(\check{\pi}, \mathscr{M}_K) \right] \leq \varepsilon(n) + C \times (2 \min(A_{n,K}, B_{n,K}) + C_{n,K}), \text{ for every } n > K,
	\end{equation}
	cf.~Theorem~\ref{thm:FSApartitionstg} for the definition of the quantities appearing on the right-hand side of~\eqref{eqn:UPBD2}. The statement can be shown arguing as in the proof of Theorem~\ref{thm:FSApartitionstg}, but we provide some details for the convenience of the reader: To prove~\eqref{eqn:UPBD2}, fix~$n \in \N$,~$n > K$, a discretization~$\mathscr{M}_K^n \neq \emptyset$ (finite),  a set of elimination rounds~$\mathcal{R}$, and~$\eta>0$. Abbreviate~$\check{\pi}_{n,t} = \check{\pi}_t$ for~$t = 1, \hdots, n+1$. That these functions (and the quantities introduced further below) are Borel measurable can be shown by arguing as in the proof of Theorem~\ref{thm:FSApartition}. To establish the upper bound in~\eqref{eqn:UPBD2}, denote by~$r_C'$ the last round completed, and denote the last elimination round completed by~$\bar{r}':=\max\cbr[0]{r\in\mathcal{R}:r\leq r_C'}$. We have~$\bar{r}'\leq \min(r_C',\lfloor n/2\rfloor)$, because~$r\leq \lfloor n/2\rfloor$ for all~$r\in\mathcal{R}$. Note that~$r_C'$ and~$\bar{r}'$ are random variables (we suppress the dependence of~$r_C'$ and~$\bar{r}'$ on~$\omega\in \Omega$, for~$(\Omega, \mathcal{A}, \mathbb{P})$ the underlying probability space), and note that~$\mathcal{I}_{\bar{r}'}=\mathcal{I}_{r_C'}$, and~$\mathscr{M}_{K,\bar{r}'}^{n}=\mathscr{M}_{K, r_C'}^{n}$, because no elimination takes place after round~$\bar{r}'$. Define the event where not all~$\delta\in\argmax_{\delta \in \mathscr{M}_K^n} \mathsf{T}(\langle \delta, \mathbf{F} \rangle)$ have been eliminated after~$\bar{r}'$ rounds
	\begin{align*}
		\mathcal{G}' := \{\omega \in \Omega: \mathscr{M}_{K, \bar{r}'}^n \cap \argmax_{\delta \in \mathscr{M}_K^n} \mathsf{T}(\langle \delta, \mathbf{F} \rangle)  \neq \emptyset \}.
	\end{align*} 
	
	Consider~$\omega \in \mathcal{G}'$ in this paragraph: if~$|\mathcal{I}_{\bar{r}'}|=1$, then the regret is at most~$\varepsilon(n)$. Assume that~$\omega$ is such that~$|\mathcal{I}_{\bar{r}'}|=|\mathcal{I}_{r_C'}|\geq 2$. Then,~$\omega \in \mathcal{G}$ implies
	\begin{align*}
		r_n(\check{\pi}, \mathscr{M}_K)
		=
		\max_{\delta\in\mathscr{M}_K}\mathsf{T}(\langle\delta,\mathbf{F}\rangle)-\mathsf{T}(\langle\check{\pi}_{n+1},\mathbf{F}\rangle)
		\leq
		\max_{\delta\in\mathscr{M}_{K,\bar{r}'}^{n}}\mathsf{T}(\langle\delta,\mathbf{F}\rangle)-\mathsf{T}(\langle\check{\pi}_{n+1},\mathbf{F}\rangle)+\eps(n),
	\end{align*}
	which (arguing as in the proof of Theorem~\ref{thm:FSApartitionstg} starting with Equation~\eqref{eqn:argrefapp1}) can be shown to be bounded from above by 
	\begin{equation*}
		2 C\max_{i\in \mathcal{I}_{\bar{r}'}}\enVert[1]{\hat{F}_{i,n,n}-F^i}_\infty+ \varepsilon(n).
	\end{equation*}
	Using that~$r_n(\check{\pi},\mathscr{M}_K)\leq C$ for~$\omega \in \Omega \setminus \mathcal{G}'$ and bounding the maximum in the previous display by arguing as in the proof of Theorem~\ref{thm:FSApartitionstg} (the argument given around Equation~\eqref{eqn:argrefapp2}) shows that for every~$\omega \in \Omega$ it holds that
	\begin{equation*}
		r_n(\check{\pi}, \mathscr{M}_K) \leq 2 C \max_{i=1}^K\sbr[2]{\sum_{s=\lfloor n/K\rfloor}^{R}\enVert[1]{\hat{F}_{i,n,n}-F^i}_\infty \mathds{1}{\cbr[1]{S_i(n)=s}}} + \varepsilon(n) + C \mathds{1}\{\Omega \setminus \mathcal{G}'\};
	\end{equation*}
	recall that~$\hat{F}_{i,n,n}$ depends on the policy used, which we do not show in our notation. The result follows from Parts~3 and~4 of Lemma \ref{lem:OptionalSkipping2} given below, for which we define the sequence of random variables~$f'_t$ via:
	\begin{align*}
		f'_t :=\begin{cases}
			\check{\pi}_{n,t}(Z_{t-1}) & \text{for } 1\leq t \leq n\\
			(t~\mathrm{mod}~K)+1 & \text{ for } t > n.
		\end{cases}
	\end{align*} 
	For every~$i\in\mathcal{I}$ and~$t\in\N$ denote the random variable~$\tau'_{i,t}$ by 
	\begin{align*}
		\tau'_{i,t}:=\min\cbr[2]{r\in\N: \sum_{j=1}^r\mathds{1}\cbr[0]{f'_j=i}=t}.
	\end{align*}
	\begin{lemma}\label{lem:OptionalSkipping2}
		It holds that 
		\begin{enumerate}
			\item The sequences of random variables~$\del[1]{Y_{i,\tau'_{i,t}}}_{t\in\N}$ and~$\del[1]{Y_{i,t}}_{t\in\N}$ have the same distribution for every~$i\in\mathcal{I}$.
			\item The sequences of random variables~$\del[1]{Y_{1,\tau'_{1,t}}}_{t\in\N},\hdots,\del[1]{Y_{K,\tau'_{K,t}}}_{t\in\N}$ are independent. 
			\item It holds that
			\begin{equation*}
				\E \max_{i=1}^K\sbr[2]{\sum_{s=\lfloor n/K\rfloor}^{R}\enVert[1]{\hat{F}_{i,n,n}-F^i}_\infty \mathds{1}{\cbr[1]{S_i(n)=s}}} \leq A_{n,K} \wedge B_{n,K}.
			\end{equation*}
			\item~$1-\P(\mathcal{G}')$ is bounded from above by
			\begin{equation*}
				\frac{\del[1]{1+K\sqrt{16\pi/(\eta e)}}^2}{(\sqrt{n}K)^{1+\eta/(2+\eta)}}\sum_{r\in\mathcal{R}}\frac{1}{r^{1+\eta/(2+\eta)}} \leq \frac{\del[1]{1+K\sqrt{16\pi/(\eta e)}}^2}{(\sqrt{n}K)^{1+\eta/(2+\eta)}}\frac{2(1+\eta)}{\eta},
			\end{equation*}
			and also by
			\begin{equation*}
				\frac{2K^{(1-\eta)/2}}{n^{(1+\eta)/4}}\sum_{r\in\mathcal{R}}\frac{1}{r^{(1+\eta)/2}}.
			\end{equation*}
		\end{enumerate}
	\end{lemma} 
	\begin{proof}
		Parts~1 and~2 are established as in the proof of Lemma~\ref{lem:OptionalSkipping} using results in \cite{belisle2008independence}. Part~3 is established analogously to the proof of Part~3 of Lemma~\ref{lem:OptionalSkipping}, replacing all occurrences of~$f_t$ and~$\tau_{i,t}$ there by~$f'_t$ and by~$\tau_{i,t}'$, respectively. For Part~4, pick a~$\delta^* \in \argmax_{\delta \in \mathscr{M}_K^n} \mathsf{T}(\langle\delta,\mathbf{F}\rangle)$. For~$\omega \in \Omega \setminus \mathcal{G}'$ we certainly have~$\mathscr{M}_{K, \overline{r}'}^n \not \ni \delta^*$. Hence, for such~$\omega$, there exists an elimination round~$r_*(\omega) \in \mathcal{R}$, say, where~$\delta^*$ is~\emph{removed} from~$\mathscr{M}^n_{K, r_*(\omega)}$. That is, there exists a~$\delta_*(\omega) \in \mathscr{M}^n_{K, r_*(\omega)}$ and an index~$t(\omega)$ such that~$$\mathsf{T}(\langle \delta_*(\omega), \hat{\mathbf{F}}_{t(\omega),n} \rangle) - \mathsf{T}(\langle \delta^*, \hat{\mathbf{F}}_{t(\omega),n} \rangle) > 2u_{\eta}(r_*(\omega), n).$$ By definition of the policy~$\check{\pi}$, as long as a treatment has not been \emph{eliminated}, each treatment is assigned once per round. The treatments corresponding to the non-zero coordinates of~$\delta_*(\omega)$ or~$\delta^*$, both vectors being subject to comparison in round~$r_*(\omega)$, have not been eliminated in previous rounds (due to the definition of the policy), and are thus all based on~$r_{*}(\omega)$ observations. Thus, we may equivalently replace~$\hat{\mathbf{F}}_{t(\omega),n}$ by~$\check{\mathbf{F}}_{r_*(\omega)} = (\check{F}_{1,r_*(\omega)}, \hdots, \check{F}_{K,r_*(\omega)})$ in the previous display, where~$\check{F}_{i,s}(\cdot):=\frac{1}{s}\sum_{r=1}^s\mathds{1}\cbr[0]{Y_{i,\tau'_{i,r}}\leq \cdot}$. Hence, we have shown that
		\begin{align*}
			\Omega \setminus \mathcal{G}' \subseteq &\bigcup_{r\in\mathcal{R}} \bigcup_{\delta_*\in \mathscr{M}_{  K}^{n}}\cbr[1]{\mathsf{T}(\langle \delta_*, \check{\mathbf{F}}_{r} \rangle)-\mathsf{T}(\langle \delta^*, \check{\mathbf{F}}_{r}\rangle)>2u_\eta(r,n)}\\
			\subseteq &
			\bigcup_{r\in\mathcal{R}} \bigcup_{\delta_*\in \mathscr{M}_{K}^{n}}\cbr[1]{\mathsf{T}(\langle \delta_*, \check{\mathbf{F}}_{r}\rangle)-\mathsf{T}(\langle \delta_*, \mathbf{F}\rangle)+\mathsf{T}(\langle \delta^*, \mathbf{F}\rangle)-\mathsf{T}(\langle \delta^*, \check{\mathbf{F}}_{r}\rangle)>2u_\eta(r,n)}\\
			\subseteq &
			\bigcup_{r\in\mathcal{R}} \cbr[2]{\max_{i\in \{1, \hdots,  K\}}\enVert[0]{\check{F}_{i,r}-F^i}_\infty >\frac{u_\eta(r,n)}{C}}.
		\end{align*}
		We thus have
		\begin{align}\label{eq:fork2}
			\P(\Omega \setminus \mathcal{G}' )
			\leq
			\sum_{r\in\mathcal{R}}\P\del[2]{\max_{i\in \{1, \hdots,  K\}}\enVert[0]{\check{F}_{i,r}-F^i}_\infty >\frac{u_\eta(r,n)}{C}}.
		\end{align}
		From Part~2 of this lemma we already know that, for every~$r\in\N$, the random variables~$\enVert[0]{\check{F}_{i,r}-F^i}_\infty$ for~$i\in\mathcal{I}$ are independent. Part 1 furthermore shows that the same statement holds if~$\check{F}_{i,r}$ is replaced by~$F_{i,r}(\cdot):=\frac{1}{r}\sum_{s=1}^r\mathds{1}\cbr[0]{Y_{i,i+(s-1)K}\leq \cdot}$, and that (for every~$r \in \N$ and every~$i = 1, \hdots, K$) the distributions of ~$\enVert[0]{\check{F}_{i,r}-F^i}_{\infty}$ and~$\enVert[0]{F_{i,r}-F^i}_{\infty}$ coincide. Consequently, for every~$r \in \N$ the random variables
		\begin{align*}
			\max_{i\in \{1, \hdots,  K\}}\enVert[0]{\check{F}_{i,r}-F^i}_\infty \quad \text{and} \quad \max_{i\in \{1, \hdots,  K\}}\enVert[0]{F_{i,r}-F^i}_\infty,
		\end{align*}
		have the same distribution. Applying Corollary \ref{cor:summax_DKW} with~$k=2/\eta$ and ``$A = \{1, \hdots, K\}$'' to each summand in the upper bound of~\eqref{eq:fork2}, after replacing the cdfs~$\check{F}_{i,r}$ by the cdfs~$F_{i,r}$, we get
		\begin{align*}
			\P(\Omega \setminus \mathcal{G}')
			&\leq \del[2]{1+K\sqrt{16\pi/(\eta e)}}^2
			\sum_{r\in\mathcal{R}}e^{-\frac{{u_\eta}^2(r,n)r}{(1+\eta/2)C^2}},
		\end{align*}
		and can finish the proof of Part 4 as in the proof of Lemma~\ref{lem:OptionalSkipping}, because this upper bound coincides with the one given in Equation~\eqref{eqn:argrefapp3}. 
	\end{proof}

\end{appendix}

\singlespacing
\bibliographystyle{chicagoa}
\bibliography{References}  
\end{document}